\documentclass[11pt]{article}
\usepackage[utf8]{inputenc}

\usepackage{amsmath}
\usepackage{amsthm}
\usepackage{amssymb}
\usepackage{algorithmic}
\usepackage{algorithm}
\usepackage{subfig}
\usepackage{color}
\usepackage[english]{babel}
\usepackage{graphicx}
\usepackage{wrapfig,epsfig}
\usepackage{epstopdf}
\usepackage{url}
\usepackage{graphicx}
\usepackage{color}
\usepackage{epstopdf}
\usepackage{scrextend}
\usepackage[T1]{fontenc}
\usepackage{bbm}
\usepackage{comment}

\usepackage[margin=1in]{geometry}

\title{Improved Sliding Window Algorithms for Clustering and Coverage via Bucketing-Based Sketches}
\author
{Alessandro Epasto\\
  \texttt{aepasto@google.com}\\
  Google Research
\and 
Mohammad Mahdian \\
  \texttt{mahdian@google.com}\\
  Google Research
\and
Vahab Mirrokni \\
  \texttt{mirrokni@google.com}\\
  Google Research
\and Peilin Zhong \\
  \texttt{peilinz@google.com}\\
  Google Research}
\date{}

\newtheorem{theorem}{Theorem}[section]
\newtheorem{lemma}[theorem]{Lemma}
\newtheorem{definition}[theorem]{Definition}

\newtheorem{fact}[theorem]{Fact}
\newtheorem{remark}[theorem]{Remark}
\newtheorem{claim}[theorem]{Claim}

\newcommand{\wt}{\widetilde}

\renewcommand{\varepsilon}{\epsilon}

\renewcommand{\div}{\mathrm{div}}
\renewcommand{\hat}{\widehat}

\newcommand{\cost}{{\rm{cost}}}
\newcommand{\wb}{\overline}

\DeclareMathOperator*{\E}{{\bf {E}}}

\DeclareMathOperator*{\Var}{{\bf {Var}}}

\DeclareMathOperator{\OPT}{OPT}

\DeclareMathOperator{\poly}{poly}

\DeclareMathOperator{\e}{ \mathbf{e}}

\allowdisplaybreaks

\begin{document}

\maketitle

\begin{abstract}
Streaming computation plays an important role in large-scale data analysis.
The sliding window model is a model of streaming computation which also captures the recency of the data.
In this model, data arrives one item at a time, but only the latest $W$ data items are considered for a particular problem.
The goal is to output a good solution at the end of the stream by maintaining a small summary during the stream.

In this work, we propose a new algorithmic framework for designing efficient sliding window algorithms via \emph{bucketing-based sketches}.
Based on this new framework, we develop space-efficient sliding window algorithms for $k$-cover, $k$-clustering and diversity maximization problems.
For each of the above problems, our algorithm achieves $(1\pm \varepsilon)$-approximation.
Compared with the previous work, it improves both the approximation ratio and the space.

\end{abstract}

\section{Introduction}
The success of large-scale computational systems together with the need of solving problems on massive data motivate the development of efficient algorithms on these systems. 
The \emph{streaming model}, where data items arrive one-by-one and can only be accessed by a single pass, was introduced by a seminal work~\cite{alon1999space}.
The goal is to (approximately) solve a problem at the end of the stream while using as little space as possible during the stream.
Due to its theoretical elegance and accurate modeling of real-world streaming computing systems (such as Spark Streaming~\cite{zaharia2012discretized}), the streaming model has attracted a lot of attention in the past decades.
Space-efficient streaming algorithms were developed for a series of fundamental problems including e.g., frequency moments estimation~\cite{alon1999space,indyk2005optimal,braverman2010zero,jayaram2019towards}, $\ell_p$ sampling~\cite{monemizadeh20101,andoni2011streaming,jayaram2021perfect}, clustering~\cite{feldman2011unified,guha2016clustering,braverman2017clustering}, coverage~\cite{badanidiyuru2014streaming,chakrabarti2016incidence,emek2016semi,saha2009maximum}, diversity maximization~\cite{indyk2004algorithms,ceccarello2017mapreduce}, sparse recovery~\cite{nakos2019nearly,nakos2019stronger}, low rank matrix approximation~\cite{ghashami2014relative,liberty2013simple,boutsidis2016optimal,song2017low}, graph problems~\cite{feigenbaum2005graph,ahn2013spectral,andoni2014towards,sun2015tight,liu2020breaking,chen2021near}.
We refer readers to a survey~\cite{muthukrishnan2005data} for more streaming algorithms and applications.

However, the classic streaming model does not completely capture the important aspect of the recency of the data. While this model treats all data items equally over the data stream,
recent data is much more important in many applications. For example, the recommendation system may not want to generate recommendations based on very old user history. Moreover, in some scenarios, data may need to be removed after a certain time for privacy or legal reasons, e.g., data privacy
laws such as the General Data Protection Regulation (GDPR), requires to not retain data beyond a specified period~\cite{upadhyay2019sublinear}.

To study these scenarios, the \emph{sliding window model} was proposed by~\cite{datar2002maintaining}.
This model is similar to the streaming model, except that only the latest $W$ data items in the stream are considered for analysis.
Although sliding window algorithms are more useful in many applications, designing an efficient sliding window algorithm is usually more difficult than designing an efficient streaming algorithm.
The reason is that when $W$ is an upper bound of the size of the stream, the streaming model can be regarded as a special case of the sliding window model. 
To tackle problems in the sliding window model, several frameworks~\cite{datar2002maintaining,braverman2007smooth,braverman2018nearly} were proposed to reduce the sliding window problems to the streaming problems. While such reductions exist for some problems with certain structural properties, many well-studied problems such as clustering, coverage, diversity maximization, low rank matrix approximation, graph sparsification and submodular maximization are not captured by such general frameworks and are studied separately in the sliding window literature~\cite{cohen2016diameter,braverman2015clustering,braverman2016clustering,borassi2020sliding,borassi2019better,braverman2020near,crouch2013dynamic,chen2016submodular,epasto2017submodular}.

In this work, we develop a new framework for sliding window algorithms. In contrast to requiring certain properties of the underlying objectives, our framework asks for an efficient algorithmic primitive called \emph{bucketing-based sketch} for the problem.
Based on this new framework, we develop near-optimal sliding window algorithms for $k$-cover and $k$-clustering problems. We also develop an efficient sliding window algorithm for diversity maximization which almost matches the space and the approximation ratio of the state-of-the-art streaming algorithm.

\paragraph{Notation and Preliminaries}
We start by providing notation and preliminaries. Let $[n]$ denote the set $\{1,2,\cdots,n\}$.
For $A,B,\alpha\geq 0$, we say $B$ is an $\alpha$-approximation of $A$ if $A\leq B\leq \alpha \cdot A$ $(\alpha \geq 1)$ or $\alpha\cdot A\leq B\leq A$ $(\alpha\leq 1)$.
We use $\mathbf{1}(\cdot)$ to denote the indicator function, i.e. $\mathbf{1}(\mathcal{E})=1$ if the event $\mathcal{E}$ happens and $\mathbf{1}(\mathcal{E})=0$ otherwise.
For any set $S$, we use $2^S$ to denote the family of all subsets of $S$.
We use $\wt{O}(f(n))$ to denote $O(f(n)\log(f(n)))$.
For a vector $x\in\mathbb{R}^d$, we use $x(i)$ to denote the $i$-th entry of $x$ for $i\in [d]$.
We use $\e_i\in\mathbb{R}^d$ for $i\in[d]$ to denote the standard unit vector where the $i$-th entry of $\e_i$ is $1$ and all other entries of $\e_i$ are $0$.
For $x\in\mathbb{R}^d,$ we use $\|x\|_2$ to denote the $\ell_2$ norm of $x$, i.e., $\|x\|_2=\sqrt{\sum_{i\in[d]}x(i)^2}$.

In the sliding window model, there is a stream of data items and the $i$-th data item in the stream is $x_i$. At an arbitrary timestamp $N$, the input data set $X$ is implicitly defined by the data stream and a given window size parameter $W>0$ such that $X=\left\{x_{N-W+1},x_{N-W+2},\cdots,x_{N}\right\}$.
The goal is to output a good (approximate) solution for some specific computational problem with respect to $X$ while requiring as small space as possible at any time during the stream.

\subsection{Our Results and Comparison to Prior Work}
There exist a long line of research on the sliding window model.
As the first algorithmic framework for analyzing sliding window model, ~\cite{datar2002maintaining} proposed \emph{exponential histogram} as a technique for approximating the class of ``weakly additive'' functions in this model.
Later, a more general framework \emph{smooth histogram} was developed by~\cite{braverman2007smooth} in which the authors handle the class of ``smooth'' functions where all weakly additive functions are shown to be smooth. This framework is further generalized by~\cite{braverman2018nearly}.
Based on these frameworks, space-efficient $(1\pm\varepsilon)$-approximate sliding window algorithms were developed for many problems including count, sum of integers, $\ell_p$ norm, frequency moments, length of longest subsequence, geometric mean, distinct elements, and heavy hitters.
However, many fundamental problems such as coverage, clustering and diversity maximization are not smooth or not smooth enough to be $(1\pm\varepsilon)$-approximated by smooth histogram (see Appendix~\ref{sec:lack_of_smoothness} for a discussion).
These problems are studied in the sliding window model separately~\cite{braverman2015clustering,braverman2016clustering,borassi2019better,borassi2020sliding}. None of these papers, however, achieve optimal or near-optimal approximation factors or space bounds for these problems.  
In contrast, we develop a new general framework for the sliding window model, and show efficient sliding window algorithms for $k$-cover, $k$-clustering and diversity maximization as natural applications of this unified framework.
Compared to previous work~\cite{borassi2019better,braverman2020near}, all our results can achieve an improved $(1\pm\varepsilon)$-approximation with improved space requirements: the space of our algorithm for $k$-cover and $k$-clustering is near optimal and the space for diversity maximization almost matches the previous best algorithm. Previous work either achieves suboptimal  approximation (e.g., 2-approximation~\cite{borassi2019better} or $O(1)$ \cite{borassi2020sliding}) or sub-optimal space requirements (e.g., quadratic space requirement for $k$-clustering in~\cite{braverman2020near}). 

\paragraph{Algorithmic framework via Bucketing-based sketches.}
We define the notion of \emph{bucketing-based} sketch as a well-structured summary of the data such that it can be easily updated in the sliding window model.
Roughly speaking, a bucketing-based sketch of the data contains a number of buckets where each bucket has a threshold.
For each data item, we process some information and put the item with processed information into some buckets.
If a bucket is full, i.e., the size of the information stored in the bucket is larger than the threshold, then only arbitrary maximal number of items are kept in the bucket.
A good (approximate) solution can be recovered from the sketch.
The formal definition of the bucketing-based sketch is given by Definition~\ref{def:bucketing_based_sketch}, and the sliding window algorithmic framework via bucketing-based sketch is described in Algorithm~\ref{alg:sliding_framework}.
We refer readers to Section~\ref{sec:ideas_alg_sketch} for high level ideas and concrete simple examples of developing sliding window algorithms via bucketing-based sketches.

In contrast to the smooth histogram framework~\cite{braverman2007smooth} which utilizes the  property of the function to make it suitable for sliding window computations,  we ask for an algorithmic primitive, i.e., an efficient bucketing-based sketch, which in turn, yields an efficient sliding window algorithm.

\paragraph{$k$-Clustering.}
We consider the point set $X$ in the discrete Euclidean space $\{1,2,\cdots,\Delta\}^d$ for $\Delta,d\in\mathbb{Z}_{\geq 1}$.
This is without loss of generality, because if the clustering cost is not zero, we can discretize the space by changing the cost by an arbitrary small multiplicative factor (see e.g.,~\cite{indyk2004algorithms,frahling2005coresets,braverman2017clustering,hu2018nearly}).
For $k\geq 1$ and $p>0$, $\ell_p$ $k$-clustering asks for a set of $k$ centers $Z\subset\mathbb{R}^d$ such that the $\ell_p$ clustering cost is minimized, i.e.,
\begin{align*}
\OPT(X):=\min_{Z\subset \mathbb{R}^d:|Z|=k}\sum_{x\in X}\min_{z\in Z} \|x-z\|_2^p.
\end{align*}
The above clustering objective is $k$-median if $p=1$ and is $k$-means if $p=2$.
We extend the sensitivity sampling based sketching technique~\cite{feldman2011unified,braverman2016new,hu2018nearly,braverman2019streaming} and construct an efficient bucketing-based sketch for $\ell_p$ $k$-clustering.
Based on such sketch, we show an efficient sliding window algorithm for $\ell_p$ $k$-clustering in the sliding window model.

\begin{theorem}[A simple version of Theorem~\ref{thm:sliding_window_clustering}]
Consider a point set in $[\Delta]^d$ given by a stream with window size $W\geq 1$.
For $\varepsilon\in(0,0.5)$, there is a sliding window algorithm which outputs a $(1+\varepsilon)$-approximation for $\ell_p$ $k$-clustering with probability at least $0.99$.
The algorithm uses space $(kd+d^{O(p)})/\varepsilon^3\cdot\poly\log(kd\Delta W)$.
The update time for each point during the stream is at most $(kd+d^{O(p)})/\varepsilon^3\cdot\poly\log(kd\Delta W)$.
\end{theorem}

If dimension $d$ is much larger than $\log(k)$, it can be reduced to $O(\log(k)/\varepsilon^2)$ by using $O(kd\log(k)/\varepsilon^2)$ space based on the dimension reduction technique of~\cite{mmr19} (see Appendix~\ref{sec:dim_reduce} for more details). 
Since $\Omega(k)$ space is necessary for any multiplicative approximation of $k$-clustering (see Appendix~\ref{sec:necessity_k_clustering}), our space is optimal up to a $\poly(\log(kd\Delta W)/\varepsilon)$ factor.

Our algorithm can actually output a corset (see Definition~\ref{def:coreset}).
Thus, we output a $(1+\varepsilon)$-approximation by finding the optimal clustering of the coreset.
If we desire a polynomial time approximation algorithm at the end of the stream, we can apply any polynomial time $\alpha$-approximate clustering algorithm over the coreset and we can obtain a $(1+\varepsilon)\alpha$-approximate solution.
For example, polynomial time constant approximation algorithms for $k$-median and $k$-means are known~\cite{ahmadian2019better}.

$k$-Clustering has been studied in the sliding window model by a line of work~\cite{braverman2015clustering,braverman2017clustering,borassi2020sliding}.
To obtain $(1+\varepsilon)$-approximation, \cite{braverman2017clustering} proposed an algorithm that, given a coreset which is maintainable in the streaming model with space $s$, maintains the coreset in the sliding window model using $O(s^2\varepsilon^{-2}\log W)$ space.
Since any coreset has size $\Omega(k)$ (by similar argument of Appendix~\ref{sec:necessity_k_clustering}), the space needed by their algorithm is $\Omega(k^2)$.
\cite{borassi2020sliding} proposed a sliding window algorithm, for arbitrary metric spaces, with space linear in $k$ while they can only achieve a constant approximation with a constant $\gg 10$.
In contrast to their algorithms, our algorithm achieves both optimal space and $(1+\varepsilon)$-approximation ratio simultaneously.

\paragraph{$k$-Cover.}
In the $k$-cover problem, given a ground set $\mathcal{E}$ of $m$ elements and a family $\mathcal{S}\subseteq 2^{\mathcal{E}}$ of $n$ sets, the goal is to choose $k$ sets $\mathcal{P}\subseteq \mathcal{S}$ such that $|\bigcup_{S\in\mathcal{P}} S|$ is maximized.
We consider $k$-cover in the \emph{edge-arrival} model~\cite{bateni2017almost}, i.e., $\mathcal{S}$ is given by a set of pairs $(S_1,e_1),(S_2,e_2),\cdots\in\mathcal{S}\times\mathcal{E}$ where $(S_i,e_i)$ indicates that element $e_i$ is in the set $S_i$. 
We show how to use a bucketing-based sketch to implement the sketch proposed by~\cite{bateni2017almost}.
As a result, we develop an efficient sliding window algorithm for $k$-cover in the edge-arrival sliding window model.
\begin{theorem}[A simple version of Theorem~\ref{thm:sliding_window_kcover_better_approx}]
Consider a $k$-cover instance over $n$ sets given by an edge-arrival streaming model with window size $W\geq 1$.
For $\varepsilon\in(0,0.5)$, there is a sliding window algorithm which outputs a $(1-\varepsilon)$-approximation for $k$-cover with probability at least $0.99$.
The algorithm uses space $n/\varepsilon^3\cdot \poly\log(nW/\varepsilon)$.
The update time for each edge during the stream is at most $k\cdot \poly\log(nW)$.
\end{theorem}
In addition to the above theorem, if polynomial running time is desired to recover an approximation at the end of the stream, we can use near linear time to obtain a $(1-1/e-\varepsilon)$-approximation (see Theorem~\ref{thm:sliding_window_kcover_fast_runtime}).

To the best of our knowledge, sliding window algorithm for $k$-cover has not been studied previously. 
The most relevant work is~\cite{bateni2017almost} which studies $k$-cover in the edge-arrival streaming model.
They achieve $(1-1/e-\varepsilon)$-approximation if the running time at the end of the stream is required to be polynomial and $(1-\varepsilon)$-approximation otherwise. 
The space of their algorithm is proven to be near optimal.
Note that the streaming model is a special case of the sliding window model when we set the window size $W$ to be an upper bound of the length of the stream.
Since the space of our algorithm matches theirs up to a poly-logarithmic factor, the space of our sliding window algorithm is also near optimal.
Coverage problems are also heavily studied in the \emph{set-arrival} model.
Set-arrival model is a special case of edge-arrival model, where each set arrives at a time and brings with it a list of its elements.
Even in the set-arrival model, the best known streaming algorithms for $k$-cover need $\sim n$ space~\cite{badanidiyuru2014streaming,saha2009maximum}.

\paragraph{Diversity maximization.}
We consider the point set $X$ in the discrete Euclidean space $\{1,2,\cdots,\Delta\}^d$ for $\Delta,d\in\mathbb{Z}_{\geq 1}$.
The goal of diversity maximization is to find a subset $Q\subseteq X$ with $|Q|=k$ such that $\div(Q)$ is maximized, where $\div(\cdot)$ can be any diversity function listed in Table~\ref{tab:diversity_functions} which is originally studied by~\cite{indyk2014composable}.

Inspired by the coreset ideas of~\cite{ceccarello2017mapreduce}, we develop an efficient bucketing-based sketch for diversity maximization.
Therefore, we obtain efficient sliding window algorithms for diversity maximization problems.
\begin{table}[h]
    \centering
    \small
    \begin{tabular}{|l|l|}
    \hline
    Name & Diversity Function \\
    \hline
     \textbf{Remote-edge} & $\mathrm{div}(Q)=\min_{u\not =v\in Q} \|u - v\|_2$ \\
     \textbf{Remote-clique} & $\mathrm{div}(Q) = \sum_{u\not = v\in Q} \|u-v\|_2$ \\
     \textbf{Remote-tree} & $\mathrm{div}(Q) = $ the cost of minimum spanning tree of $Q$ \\
     \textbf{Remote-cycle} & $\mathrm{div}(Q)= $ the cost of minimum TSP tour of $Q$ \\
     \textbf{Remote $t$-trees} & $\mathrm{div}(Q)=$ the minimum cost of $t$ trees spanning $Q$ \\
     \textbf{Remote $t$-cycles} & $\mathrm{div}(Q)=$ the minimum cost of $t$ cycles spanning $Q$\\
     \textbf{Remote-star} & $\mathrm{div}(Q) = \min_{u\in Q}\sum_{v\in Q: v\not = u} \|u-v\|_2$ \\
     \textbf{Remote-bipartition} & $\mathrm{div}(Q) = \min_{Q_1,Q_2\subseteq Q:Q2 = Q\setminus Q_1,|Q_1|=\lfloor k/2\rfloor} \sum_{u\in Q_1,v\in Q_1}\|u-v\|_2$ \\
     \textbf{Remote-pseudoforest} & $\mathrm{div}(Q) = \sum_{u\in Q} \min_{v\in Q:v\not =u} \|u-v\|_2$ \\
     \textbf{Remote-matching} & minimum cost of a perfect matching of $Q$ (for even $k$) \\
     \hline
    \end{tabular}
    \caption{Diversity Functions}
    \label{tab:diversity_functions}
\end{table}

\begin{theorem}[A simple version of Theorem~\ref{thm:sliding_window_div}]
Consider a point set in $[\Delta]^d$ given by a stream with window size $W\geq 1$.
For $\varepsilon\in(0,0.5)$, there is a deterministic sliding window algorithm which outputs a $(1-\varepsilon)$-approximation for the diversity maximization problem.
The algorithm uses space $kd\log(d\Delta)\cdot O(\sqrt{d}/\varepsilon)^d$ for remote-edge, remote-tree, remote-cycle, remote $t$-trees and remote $t$-cycles, and uses space $k^2d\log(d\Delta)\cdot O(\sqrt{d}/\varepsilon)^d$ for remote-clique, remote-star, remote-bipartition, remote-pseudoforest and remote-matching.
The update time for each point during the stream is at most $O(d\log(d\Delta))$.
\end{theorem}
We actually output a subset $S\subseteq X$ at the end of the stream such that the maximized diversity of $S$ is a $(1-\varepsilon)$-approximation of the maximized diversity of $X$ (see Section~\ref{sec:div_max_sliding_window}).
Therefore, if we apply any polynomial time $\alpha$-approximate diversity maximization algorithm on $S$, we can get a $(1-\varepsilon)\alpha$-approximation of the maximized diversity of $X$ in polynomial time. 
We refer readers to~\cite{chandra2001approximation,halldorsson1999finding, hassin1997approximation,tamir1991obnoxious} for polynomial time constant approximation algorithms for specific diversity functions.

Diversity maximization was previously studied in the sliding window model by~\cite{borassi2019better}.
In comparison with their algorithm, we achieve $(1-\varepsilon)$-approximation while they only achieve $(1/5-\varepsilon)$-approximation.
On the other hand, our space matches their space while ignoring the factors of $d$ and $\varepsilon$.
In the case where the window size $W$ is an upper bound of the length of the stream, that is, in the streaming model, diversity maximization was studied by~\cite{indyk2014composable,ceccarello2017mapreduce}.
The space needed by~\cite{indyk2014composable} is $\Omega(\sqrt{W})$.
The space used by~\cite{ceccarello2017mapreduce} is almost the same as ours.

\subsection{High-Level Ideas of Sliding Window Algorithms via Bucketing-based Sketches}\label{sec:ideas_alg_sketch}
To deliver the high level ideas of our framework, in this section, we consider two simple problems fitting our framework. 

Consider a (multi-)set $X\subseteq\{0,1\}$.
Let $X_0=\{x\in X\mid x=0\},X_1=\{x\in X\mid x=1\}$.
Suppose $|X_1|\in[o,2\cdot o]$ for some $o>0$.
There is a simple way to construct a bucketing-based sketch $Z(X)$ for estimating $|X_1|$ up to $(1\pm \varepsilon)$-approximation.
Let $Z(X)$ only contain one bucket, and we add each $x$ with $x=1$ into the bucket with probability $p=\min\left(\Theta(\log(1/\delta)/(\varepsilon^2\cdot o)),1\right)$.
By Chernoff bound, $|Z(X)|/p$ is a $(1\pm\varepsilon)$-approximation to $|X_1|$ with probability at least $1-\delta$.
Now consider that $X$ is given by a stream with window size $W$.
Then it corresponds to the following sliding window algorithm for estimating $|X_1|$.
Let threshold $k=10\cdot \log(1/\delta)/\varepsilon^2$.
We maintain a set $H$ during the stream.
For each data item $x$ in the stream, if $x=1$, we add $x$ into $H$ with probability $p$.
If $|H|>k$, we remove the data item from $H$ with the earliest timestamp.
The observation is that if $|X_1|\in[o,2\cdot o]$ at the end of the stream, $Z(X)\subseteq H$ with probability at least $1-\delta$.
Thus, we can recover $Z(X)$ via $H$, i.e., $Z(X)=\{x\in H\mid\text{the timestamp of }x\text{ is at least }N-W+1\}$, where $N$ is the final timestamp.
Thus, a $(1\pm\varepsilon)$-approximation to $|X_1|$ can be obtained via $H$ and the space needed during the stream is $O(k)$.
Notice that we are able to approximately verify whether $|X_1|\in [o,2\cdot o]$ via $H$ for any $X$:
If we cannot use $H$ to recover $Z(X)$, i.e., the earliest timestamp of an item in $H$ is later than $N-W+1$, then $|Z(X)|>|H|$ and $|Z(X)|/p>|H|/p=k/p\gg 2\cdot o$ which implies that $|X_1|\gg 2\cdot o$.
Otherwise, if we can recover $Z(X)$ from $H$ and $|Z(X)|/p=|H|/p\ll o$, then $|X_1|\ll o$.
Thus, by considering $o=1,2,4,8,\cdots,W$ in parallel during the stream, we are able to estimate $|X_1|$ at the end of the stream.
The total space needed is $O(k\log W)$.

The above algorithm can be extended to the following toy $1$-median problem.
We still consider the (multi-)set $X\subseteq \{0,1\}$, the goal is to estimate $\OPT(X)=\min_{z\in\{0,1\}}\sum_{x\in X}|x-z|$.
It is easy to verify that $\OPT(X)=\min(|X_0|,|X_1|)$.
Suppose $\OPT(X)\in[o,2\cdot o]$.
Similar to the problem described in the previous paragraph, we can construct a bucketing-based sketch $Z'(X)$.
$Z'(X)$ has two buckets, and each has a threshold $k=10\cdot \log(1/\delta)/\varepsilon^2$.
We add each $x$ with $x=1$ into the first bucket with probability $p=\min(\Theta(\log(1/\delta)/(\varepsilon^2\cdot o),1)$, and add each $x$ with $x=0$ into the second bucket with probability $p$.
Since $\min(|X_0|,|X_1|)=\OPT(X)\in[o,2\cdot o]$, we know that $\min(|\{x\in Z'(X)\mid x= 0\}|,|\{x\in Z'(X)\mid x= 1\}|)/p$ is a $(1\pm \varepsilon)$-approximation of $\OPT(X)$ with probability at least $1-\delta$.
Now suppose $X$ is given by a stream with window size $W$.
Based on the same argument as the previous paragraph, we can either obtain the sketch $Z'(X)$ and get a $(1\pm\varepsilon)$-approximation of $\OPT(X)$ via $Z'(X)$ at the end of the stream or we can approximately confirm that $\OPT(X)\not\in [o,2\cdot o]$.
By considering $o=1,2,4,8,\cdots,W$ in parallel during the stream, we are able to estimate $\OPT(X)$ at the end of the stream. 
The total space needed is $O(k\log W)$.
In Appendix~\ref{sec:lack_of_smoothness_toy_1_median}, we show that toy $1$-median problem is not smooth enough to be $(1\pm\varepsilon)$-approximated by smooth histogram.

By formalizing above ideas in Section~\ref{sec:framework_via_bucketing_based_sketch}, we obtain our algorithmic framework in Algorithm~\ref{alg:sliding_framework}.

\section{A New Algorithmic Framework for Sliding Window Model}\label{sec:framework_via_bucketing_based_sketch}
Given a set of data items $X$, a natural class of problems is to estimate the value $f(X)$ for some non-negative function $f(\cdot)$.
In many situations, $X$ is very large, and we cannot afford to store entire $X$ in the space.
A common approach is to compute a small (randomized) summary/sketch $Z(X)$ such that we can estimate $f(X)$ by $z(Z(X))$ for some function $z(\cdot)$.
If $z(Z(X))$ is an $\alpha$-approximation to $f(X)$ (with a good probability), then we say $Z(\cdot)$ is an $\alpha$-approximate (randomized) sketch function for $f(\cdot)$, and $z(\cdot)$ is the corresponding recover function.
Similarly, if (with a good probability) $z(Z(X))$ either outputs \textbf{FAIL} or outputs an $\alpha$-approximation to $f(X)$, then we say $Z(\cdot)$ is a weak $\alpha$-approximate (randomized) sketch function for $f(\cdot)$.
If it is clear in the context, we will call a sketch function a sketch for short.
We say $Z(\cdot)$ is a (weak) $o$-restricted $\alpha$-approximate sketch for $f(\cdot)$ if $Z(\cdot)$ is a (weak) $\alpha$-approximate sketch for the restriction of $f(\cdot)$ to $\{X\mid f(X)\geq o\}$.

In the following, we give a definition of a particular type of sketches. 
We call it bucketing-based sketch. 
We will show that this sketch plays an important role in designing our sliding window algorithms.
Figure~\ref{fig:bucketing_based_sketch} shows an explanation of the bucketing-based sketch.

\begin{figure}[t!]
    \centering
    \includegraphics[width=\textwidth]{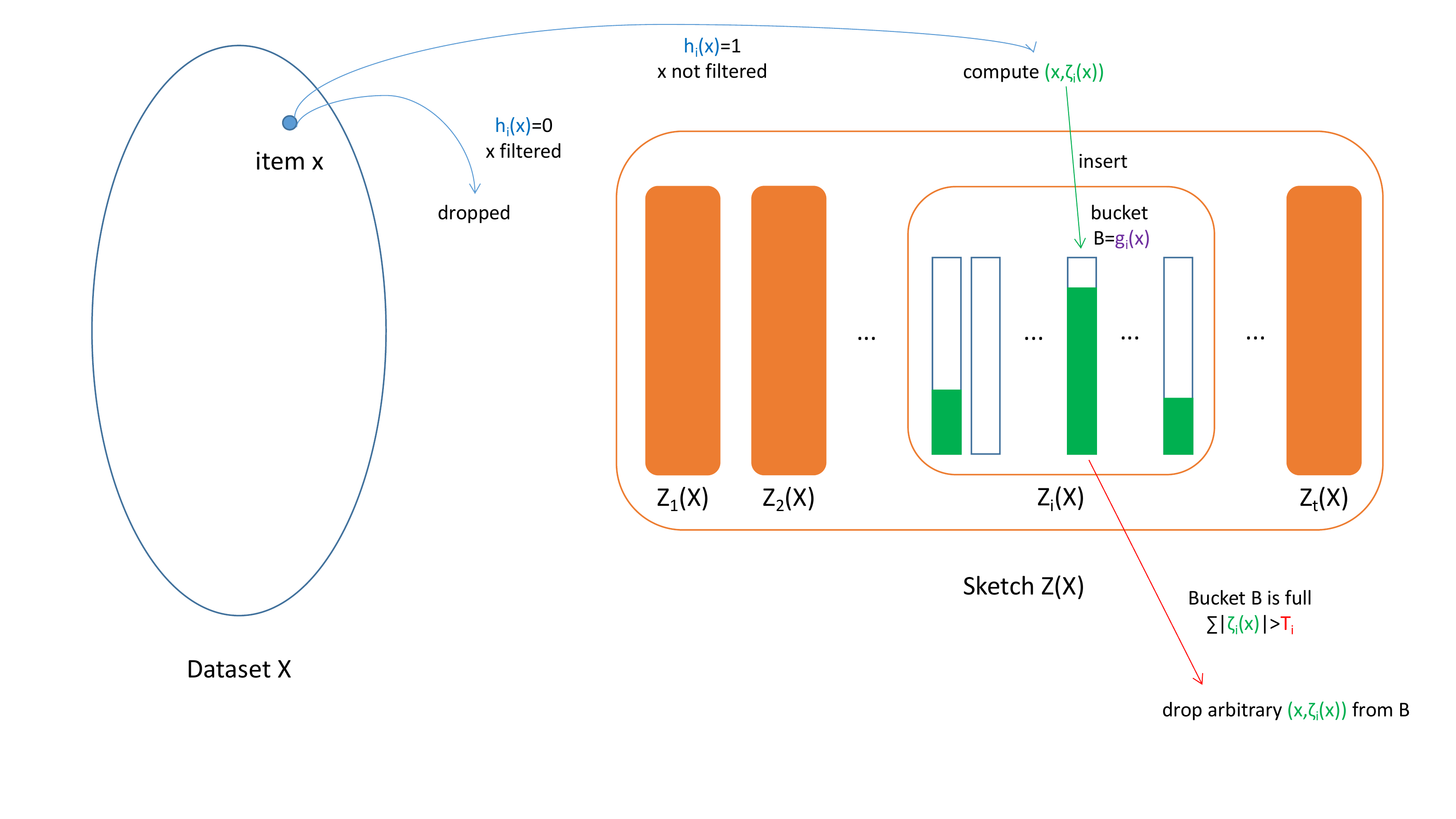}
    \vspace{-0.7in}
    \caption{Bucketing-based sketch.}
    \label{fig:bucketing_based_sketch}
\end{figure}

\begin{definition}[Bucketing-based sketch]\label{def:bucketing_based_sketch}
Let $\mathcal{X}$ be the universe of data items.
We say $Z(\cdot)$ is a (randomized) bucketing-based sketch function, if for any $X\subseteq \mathcal{X}$, $Z(X)$ satisfies following properties:
\begin{enumerate}
    \item $Z(X)$ is a tuple of sub-sketches $(Z_1(X),Z_2(X),\cdots,Z_t(X))$.
    \item For $i\in[t],$ there is a (random) filter function $h_i:\mathcal{X}\rightarrow \{0,1\}$, a (random) bucketing function $g_i:\mathcal{X}\rightarrow \mathcal{B}$, a (random) processing function $\zeta_i:\mathcal{X}\rightarrow \mathcal{I}$ and a threshold $T_i\geq 0$ such that $Z_i(X)\subseteq \mathcal{X}\times \mathcal{I}$ is constructed as the following:
    \begin{enumerate}
        \item For $x\in X$ with $h_i(x)=0$, $(x,\zeta_i(x))$ must not be added into $Z_i(X)$. \label{it:bucket_sketch_2a}
        \item Consider each bucket $B\in\mathcal{B}$. 
        If $\sum_{x\in X: h_i(x)=1,g_i(x)=B} |\zeta_i(x)|\leq T_i$, where $|\zeta_i(x)|$ denotes the space needed to store the information $\zeta_i(x)$, then for every $x\in X$ with $h_i(x)=1$ and $g_i(x)=B$, $(x,\zeta_i(x))$ is added into $Z_i(X)$.
        Otherwise, $(x_1,\zeta_i(x_1)),(x_2,\zeta_i(x_2)),\cdots,(x_s,\zeta_i(x_s))$ are added into $Z_i(X)$ for some arbitrary\footnote{
        An arbitrary tiebreaker is allowed for selecting $x_1,x_2,\cdots,x_{s}$. 
        In our paper, we keep the data items with the latest timestamps during the stream.
        }  points $x_1,x_2,\cdots,x_s$ satisfying $h_i(x_1)=h_i(x_2)=\cdots=h_i(x_s)=1,g_i(x_1)=g_i(x_2)=\cdots=g_i(x_s)=B,\sum_{j=1}^s |\zeta_i(x_j)|\leq T_i$ and $\exists x'\in \{x\in X\mid h_i(x)=1,g_i(x)=B\}\setminus\{x_1,\cdots,x_s\},|\zeta_i(x')|+\sum_{j=1}^s |\zeta_i(x_j)|>T_i$. \label{it:bucket_sketch_2b}
    \end{enumerate}
\end{enumerate}
The space budget of $Z(X)$ is defined as $\sum_{i=1}^t\sum_{B\in\mathcal{B}}\min(T_i, \sum_{x\in X:h_i(x)=1,g_i(x)=B} |\zeta_i(x)|)$.
\end{definition}

It is easy to see that the space budget of $Z(X)$ defined in Definition~\ref{def:bucketing_based_sketch} is always an upper bound of the actual space used by $Z(X)$.
The notion of the space budget is more convenient in the analysis since the space budget of $Z(X)$ is always an upper bound of the space budget of $Z(X')$ if $X'\subseteq X$, while the actual space used by $Z(X)$ can be less than the actual space used by $Z(X')$. 
This is because when a bucket $B\in\mathcal{B}$ is full, data items $x_1,x_2,\cdots,x_s$ in condition~\ref{it:bucket_sketch_2b} of Definition~\ref{def:bucketing_based_sketch} can be any choice.
Different choices of $x_1,x_2,\cdots,x_s$ will result in different actual space used by $Z(X)$.
In addition, it is easy to verify that if the size of $\zeta_i(x)$ is the same for all data items $x$, the space budget of $Z(X)$ is always the same as the actual space used by $Z(X)$.

\begin{algorithm}[ht]
	\small
	\begin{algorithmic}[1]\caption{Sliding Window Algorithmic Framework via Bucketing-based Sketches}\label{alg:sliding_framework}
	    \STATE {\bfseries Pre-determined:}  $f(\cdot)$, $S$, $m$ and $M$. {\hfill //See Theorem~\ref{thm:framework}.}
		\STATE {\bfseries Input:} A stream of data items $x_1,x_2,\cdots\in \mathcal{X}$ and the window size $W$. 
		\STATE $O:=\{m,m\cdot 2,m\cdot 4,m\cdot 8,\cdots, M\}$.
		\STATE For $o\in O$, let $Z^o(\cdot)=(Z^o_1(\cdot),Z^o_2(\cdot),\cdots,Z^o_t(\cdot))$ be a weak $o$-restricted $\alpha$-approximate bucketing-based sketch function for $f(\cdot)$ described by Theorem~\ref{thm:framework}. {\hfill //The notations $Z^o(\cdot)$ are for analysis only.}
		\STATE For $o\in O, i\in[t]$, let $h^o_i(\cdot), g^o_i(\cdot),\zeta^o_i(\cdot)$ and $T^o_i$ be the filter function, bucketing function, processing function and the threshold corresponding to the sub-sketch function $Z_i^o(\cdot)$ respectively. {\hfill //Definition~\ref{def:bucketing_based_sketch}.}
		\STATE For $o\in O$, initialize $l_o:=1,H^o:=(H^o_1,H^o_2,\cdots,H^o_t)$ and $\forall i\in[t],$ $H^o_i:=\emptyset$.
		\STATE Initialize current timestamp $N:=0$.
		\FOR{ the next data item $x_{N+1}$ in the stream} \label{sta:start_process_stream}
		    \STATE $N:=N+1$.
		    \STATE Set the timestamp $\tau(x_N):=N$.
		    \FOR{$o\in O$} 
		        \FOR{$i\in[t]$}\label{sta:start_inserting}
		            \STATE If $h^o_i(x_N)=1$, add $(x_N,\zeta_i^o(x_N))$ into $H_i^o$. {\hfill //$X_N$ is not filtered.} \label{sta:filtering_step}
		            \WHILE {$ \sum_{(x,\zeta_i^o(x))\in H_i^o:g_i^o(x) = g_i^o(x_N)} |\zeta_i^o(x)|>T_i^o$} 
		            \item[] {\hfill //The bucket is full.}
		            \label{sta:start_bucketing_step} 
		            \STATE Remove $(y,\zeta_i^o(y))\in \{(x,\zeta_i^o(x))\in H_i^o\mid g_i^o(x)=g_i^o(x_N)\}$ with the smallest $\tau(y)$ from $H_i^o$.\label{sta:bucketing_step}
		            \ENDWHILE\label{sta:end_bucketing_step}
		        \ENDFOR\label{sta:end_inserting}
		        \WHILE{the size of $H^o$ is larger than $S$}\label{sta:while_start}
		            \STATE For $i\in[t]$, if $\exists (x,\zeta_i^o(x))\in H_i^o$ such that the timestamp $\tau(x)\leq l_o$, remove $(x,\zeta_i^o(x))$ from $H_i^o$. \label{sta:remove_earlist}
		            \STATE $l_o:=l_o+1$. \label{sta:left_pointer_increase}
		        \ENDWHILE\label{sta:while_end}
		    \ENDFOR
		\ENDFOR\label{sta:end_process_stream}
		\STATE For each $o\in O$ with $l_o\leq N-W+1$, let $\hat{H}^o=(\hat{H}_1^o,\cdots,\hat{H}^o_t)$ be $H^{o}=(H^{o}_1,\cdots,H^{o}_t)$ but for each $i\in[t]$ remove all $(x,\zeta_i^o(x))$ with timestamp $\tau(x)\leq N-W$ from $H_i^o$.\label{sta:build_hat_h}
		\STATE Find the smallest $o^*$ such that the recover function over $\hat{H}^{o^*}$ does not output \textbf{FAIL} and return the output of the recover function over $\hat{H}^{o^*}$ as the $\alpha$-approximation to $f(\{x_{N-W+1},x_{N-W+2},\cdots,x_N\})$. \label{sta:final_construction}
	\end{algorithmic}
\end{algorithm}

We show that if $f(\cdot)$ admits a weak $\alpha$-approximate bucketing-based sketch, there is a sliding window algorithm for estimating $f(\cdot)$.
The algorithm is shown in Algorithm~\ref{alg:sliding_framework}.
\begin{theorem}\label{thm:framework}
Let $\mathcal{X}$ be the universe of data items.
Let $\mathcal{D}$ be an arbitrary family of subsets of $\mathcal{X}$.
Consider a function $f:\mathcal{D}\rightarrow \mathbb{R}_{\geq 0}$ satisfying that $\forall X\in \mathcal{D}$, $f(X)\in[m,M]$ for $0<m\leq M$.
Suppose $\forall o\in[m,M]$, there always exists a weak $o$-restricted $\alpha$-approximate bucketing-based sketch $Z^o(\cdot)=(Z_1^o(\cdot),Z_2^o(\cdot),\cdots,Z_t^o(\cdot))$ for $f(\cdot)$
such that 
\begin{enumerate}
    \item the recover function over $Z^o(X)$ either outputs \textbf{FAIL} or outputs an $\alpha$-approximation to $f(X)$ with probability at least $1-\delta$ conditioned on $f(X)\geq o$;
    \item if $f(X)\leq 2\cdot o$, with probability at least $1-\delta$, the space budget of $Z^o(X)$ is upper bounded by $S$ and the recover function over $Z^o(X)$ does not output \textbf{FAIL}.
\end{enumerate}
Let $X\in \mathcal{D}$ be a set of data items given by a sliding window over a stream.
Then there is a sliding window algorithm (Algorithm~\ref{alg:sliding_framework}) which outputs an $\alpha$-approximation to $f(X)$ and uses space at most $O(S\cdot \log(M/m))$.
The success probability of the algorithm is at least $1-O(\delta\cdot \log(M/m))$.
Furthermore, if the time needed to compute a filter function, a bucketing function and a processing function on any data item $x\in\mathcal{X}$ is always upper bounded by $\mathcal{T}$, the update time of the algorithm for each data item is at most $O(t\log(M/m)\cdot \mathcal{T}+S\cdot \log(M/m))$.
\end{theorem}
\begin{proof}
Firstly, let us analyze the total space used by Algorithm~\ref{alg:sliding_framework}.
The total space needed is the space to store $H^o$ for every $o\in O$.
Due to the loop in line~\ref{sta:while_start}-\ref{sta:while_end} of Algorithm~\ref{alg:sliding_framework}, the size of $H^o$ is always upper bounded by $S$.
Thus, the total space of the algorithm is at most $O(|O|\cdot S)=O(S\cdot \log(M/m))$.

Next, let us analyze the correctness.
Let $X^o=\{x_{l_o},x_{l_o+1},\cdots,x_N\}$ during the stream.
\begin{claim}\label{cla:maintain_sketch}
At any time at the end (or before the start) of the iteration of the loop in the line~\ref{sta:start_process_stream}-\ref{sta:end_process_stream} of Algorithm~\ref{alg:sliding_framework}, $\forall o\in O$, $H^o$ is always $Z^o(X^o)$. 
Furthermore, $\forall o\in O, i\in [t]$, $H_i^o$ only contains maximal number of data items with the latest timestamps, i.e., $\forall x\in X^o$ with $h_i^o(x)=1$, if $(x,\zeta_i^o(x))\not\in H_i^o$ then $\sum_{y\in X^o:\tau(y)\geq \tau(x),h_i^o(y)=1,g_i^o(y)=g_i^o(x)} |\zeta_i^o(y)|>T_i^o$.
\end{claim}
\begin{proof}
The proof is by induction.
Let $N$ be the value as at the end (or before the start) of the iteration of the loop in the line~\ref{sta:start_process_stream}-\ref{sta:end_process_stream} of Algorithm~\ref{alg:sliding_framework}.
Consider an arbitrary $o\in O$.
Consider the case when $N=0$. 
$X^o$ is empty and thus the claim holds.
In the following we consider the case when the claim holds for $N-1$.

Let us first consider the behavior of the loop in line~\ref{sta:start_inserting}-\ref{sta:end_inserting}.
Due to line~\ref{sta:filtering_step}, $x_N$ can be added into $H_i^o$ for $i\in [t]$ only when $h_i^o(x_N)=1$.
Thus, the condition~\ref{it:bucket_sketch_2a} of Definition~\ref{def:bucketing_based_sketch} is satisfied.
Fix $i\in [t]$, and let the bucket $B=g_i^o(x_N)$.
Since we may add $(x_N,\zeta_i^o(x_N))$ into $H_i^o$, the bucket $B$ may violate the condition~\ref{it:bucket_sketch_2b} of Definition~\ref{def:bucketing_based_sketch}.
By our induction hypothesis, the condition~\ref{it:bucket_sketch_2b} holds before processing $x_N$.
There are two cases.
In the first case $\sum_{x\in X^o\setminus\{x_N\}:h_i^o(x)=1,g_i^o(x)=B}|\zeta_i^o(x)|\leq T_i^o$.
In this case, $H_i^o=\{(x,\zeta_i^o(x))\mid x\in X^o, h_i^o(x)=1,g_i^o(x)=B\}$ before line~\ref{sta:start_bucketing_step}. 
In the second case $\sum_{x\in X^o\setminus\{x_N\}: h_i^o(x)=1,g_i^o(x)=B} |\zeta_i^o(x)|>T_i^o$, we have $\exists x'\in X_i^o\setminus\{x_N\},(x',\zeta_i^o(x'))\not\in H_i^o,h_i^o(x')=1,g_i^o(x')=B$ and $|\zeta_i^o(x')|+\sum_{(x,\zeta_i^o(x))\in H_i^o} |\zeta_i^o(x)|>T_i^o$ before line~\ref{sta:start_bucketing_step}.
In either case, due to the loop in line~\ref{sta:start_bucketing_step}-\ref{sta:end_bucketing_step}, it is easy to verify that the condition~\ref{it:bucket_sketch_2b} of Definition~\ref{def:bucketing_based_sketch} is always preserved at the end of the loop.
Furthermore, due to the induction hypothesis and the condition of the loop in line~\ref{sta:start_bucketing_step}, we have that $\forall x\in X^o$ with $h_i^o(x)=1$ if $(x,\zeta_i^o(x))\not\in H_i^o$, then $\sum_{y\in X^o:\tau(y)\geq \tau(x),h_i^o(y)=1,g_i^o(y)=g_i^o(x)}|\zeta_i^o(y)|>T_i^o$

Next, let us analyze the behavior of the loop in line~\ref{sta:while_start}-\ref{sta:while_end}.
Let us focus on a particular $o\in O$ and $i\in [t]$.
Since line~\ref{sta:remove_earlist} only removes elements from $H_i^o$, the condition~\ref{it:bucket_sketch_2a} will not be violated.
Furthermore, since line~\ref{sta:remove_earlist} only removes elements with the earliest timestamps, we still have that $\forall x\in X^o$ with $h_i^o(x)=1$, if $(x,\zeta_i^o(x))\not\in H_i^o$ then $\sum_{y\in X^o:\tau(y)\geq \tau(x),h_i^o(y)=1,g_i^o(y)=g_i^o(x)}|\zeta_i^o(y)|>T_i^o$.
Now we are going to prove that condition~\ref{it:bucket_sketch_2b} of Definition~\ref{def:bucketing_based_sketch} holds after line~\ref{sta:left_pointer_increase}.
Suppose that the condition~\ref{it:bucket_sketch_2b} does not hold after line~\ref{sta:left_pointer_increase}.
We can find $x'\in X^o$ such that $h_i^o(x')=1,(x',\zeta_i^o(x'))\not\in H_i^o$ and $\sum_{x\in X^o:h_i^o(x)=1,g_i^o(x)=g_i^o(x'),\tau(x)\geq \tau(x')}\leq T_i^o$ which contradicts to the condition $\forall x\in X^o$ with $h_i^o(x)=1$ if $(x,\zeta_i^o(x))\not\in H_i^o$, then $\sum_{y\in X^o:\tau(y)\geq \tau(x),h_i^o(y)=1,g_i^o(y)=g_i^o(x)} |\zeta_i^o(y)|>T_i^o$.
Thus the condition~\ref{it:bucket_sketch_2b} of Definition~\ref{def:bucketing_based_sketch} must hold after line~\ref{sta:left_pointer_increase}.

Thus, the claimed statement always hold.
\end{proof}

Due to Claim~\ref{cla:maintain_sketch}, we are able to verify that $\hat{H}^o$ for each valid $o$ in line~\ref{sta:build_hat_h} is equivalent to $Z^{o}(\{x_{N-W+1},x_{N-W+2},\cdots,x_N\})$: the condition~\ref{it:bucket_sketch_2a} of Definition~\ref{def:bucketing_based_sketch} holds since line~\ref{sta:build_hat_h} only removes elements, and the condition~\ref{it:bucket_sketch_2b} holds since otherwise we can find $x'\in \{x_{N-W+1},x_{N-W+2},\cdots,x_N\}$, $i\in [t]$ such that $h_i^{o}(x')=1,(x',\zeta_i^{o}(x'))\not\in H_i^{o},\sum_{x\in X^{o}:\tau(x)\geq \tau(x'),h_i^{o}(x)=1,g_i^{o}(x)=g_i^{o}(x')} |\zeta_i^{o}(x)|\leq T_i^{o}$ which violates Claim~\ref{cla:maintain_sketch}.
Since for $o\in O$, $Z^{o}(\cdot)$ is a weak $o$-restricted $\alpha$-approximate sketch for $f(\cdot)$.
In the remaining of the proof, we will show that $o^*\leq f(\{x_{N-W+1},x_{N-W+2},\cdots,x_{N}\})$ with a good probability.

By the construction of $O$, we can find $o'\in O$ such that $f(x)\in[o',2\cdot o']$.
With probability at least $1-\delta$, $Z^{o'}(\{x_{N-W+1},x_{N-W+2},\cdots,x_N\})$ has space budget at most $S$, and the recover function over $Z^{o'}(\{x_{N-W+1},x_{N-W+2},\cdots,x_N\})$ does not output \textbf{FAIL}.
We condition on that $Z^{o'}(\{x_{N-W+1},x_{N-W+2},\cdots,x_N\})$ has space budget at most $S$, and the recover function over $Z^{o'}(\{x_{N-W+1},x_{N-W+2},\cdots,x_N\})$ does not output \textbf{FAIL}.
If $l_{o'}>N-W+1$, by the condition of line~\ref{sta:while_start} and the proof of Claim~\ref{cla:maintain_sketch}, it implies that the space budget of $Z^{o'}(\{N-W+1,N-W+2,\cdots,N'\})$ for some $N'\leq N$ is greater than $S$ which contradicts to that $Z^{o'}(\{x_{N-W+1},\cdots,x_N\})$ has space budget at most $S$.
Thus, we must have $l_{o'}\leq N-W+1$.
Thus, by line~\ref{sta:final_construction}, we know that $o^*\leq o'\leq f(x)$.

For any $o\in O$ with $o\leq f(\{x_{N-W+1},x_{N-W+2},\cdots,x_N\})$, with probability at least $1-\delta$, we are able to obtain an $\alpha$-approximation from $Z^o(\{x_{N-W+1},x_{N-W+2},\cdots,x_N\})$ conditioned on that the recover function over $Z^o(\{x_{N-W+1},x_{N-W+2},\cdots,x_N\})$ does not output \textbf{FAIL}.
By taking union bound over all $o\in O$, with probability at least $1-O(\delta\cdot \log(M/m))$, we are able to obtain an $\alpha$-approximation to $f(\{x_{N-W+1},x_{N-W+2},\cdots,x_N\})$ from $\hat{H}^{o^*}=Z^{o^*}(\{x_{N-W+1},x_{N-W+2},\cdots,x_N\})$.

Finally, let us analyze the update time of the algorithm.
When the algorithm processes $x_N$, the time needed to compute $h_i^o(x),g_i^o(x),\zeta_i^o(x)$ for each $i\in [t],o\in O$ is always upper bounded by $\mathcal{T}$.
Thus the total time needed to compute the filter function, bucketing function and processing function on $x_N$ is at most $O(t\cdot |O|\cdot \mathcal{T})=O(t\cdot \log(M/m)\cdot \mathcal{T})$.
In the remaining of the updating process, the algorithm can only remove items from $H^o$ for $o\in O$.
Since the size of $H^o$ for $o\in O$ is at most $S$, the overall time of removal is at most $O(S\cdot |O|)=O(S\cdot \log(M/m))$.
Thus, the overall update time is $O(t\log(M/m)\cdot \mathcal{T}+S\cdot \log(M/m))$.
\end{proof}

\begin{remark}
Although Theorem~\ref{thm:framework} does not handle the instance $X$ with $f(X)=0$, it is usually easy to verify whether $f(X)=0$ in the sliding window model for many problems (such as all problems studied in this paper).
Therefore, we can use Algorithm~\ref{alg:sliding_framework} to handle the case $f(X)>0$ and run a sliding window procedure to verify whether $f(X)=0$ in parallel.

\end{remark}

\section{$k$-Cover in the Sliding Window Model}\label{sec:k_cover}
In the $k$-cover problem, there is a ground set $\mathcal{E}$ of $m$ elements and a family $\mathcal{S}\subseteq 2^{\mathcal{E}}$ of $n$ subsets of the elements.
Let $\mathcal{P}\subseteq \mathcal{S}$ be a subfamily of subsets.
The \emph{coverage} of $\mathcal{P}$ is denoted as $\mathcal{C}(\mathcal{P})=\left|\bigcup_{S\in \mathcal{P}} S\right|$, i.e., the total number of elements in the union of all subsets in $\mathcal{P}$.
The goal of $k$-coverage is to find a subfamily $\mathcal{P}\subseteq \mathcal{S}$ with size $|\mathcal{P}|=k$ such that $\mathcal{C}(\mathcal{P})$ is maximized.
The $k$-cover problem can also be described by a bipartite graph $G$, where $\mathcal{S}$ corresponds to the vertices on the one part, and $\mathcal{E}$ corresponds to the vertices on the other part. 
There is an edge between a subset $S\in \mathcal{P}$ and an element $i\in\mathcal{E}$ if and only if the element $i$ is in the subset $S$.
For a set $\mathcal{P}$ of vertices in $G$, let $\Gamma(G,\mathcal{P})$ denote the neighbors of $\mathcal{P}$ in $G$. 
In particular, if $\mathcal{P}$ corresponds to a subfamily of subsets in $\mathcal{S}$, then we have $|\Gamma(G,\mathcal{P})|=\mathcal{C}(\mathcal{P})$.
We use $\OPT_k(G)$ to denote the optimal $k$-cover value, i.e., $\OPT_k(G):=\max_{\mathcal{P}\subseteq\mathcal{S}:|\mathcal{P}|=k} |\Gamma(G,\mathcal{P})|$.
If $G$ is clear in the context, we will use $\OPT_k$ to denote $\OPT_k(G)$ for short.

In this paper, we consider $k$-cover in the edge-arrival model~\cite{bateni2017almost}.
In particular, there is a stream of set-element pairs $(S_1,e_1),(S_2,e_2),\cdots$, where $S_i\in\mathcal{S}$ corresponds to a subset, $e_i\in \mathcal{E}$ corresponds to an element.
Given a window size $W$, at timestamp $N$, the goal is to approximate $\OPT_k(G)$, where $G$ is represented by the edges $(S_{N-W+1},e_{N-W+1}),(S_{N-W+2},e_{N-W+2}),\cdots,(S_N,e_N)$.

\subsection{Offline Sketch via Subsampling}\label{sec:kcover_offline_sketch}
Let us review a sketch proposed by~\cite{bateni2017almost}. 
In particular, a sketch graph $H'_p$ for $G$ is created.
Let $\delta\in(0,0.5)$ be a probability parameter.
Let $\hslash$ be an $O(k\log(1/\delta)\log n)$-wise independent hash function which maps each element in $\mathcal{E}$ to $[0,1]$ uniformly at random.
Let $\varepsilon\in(0,0.5)$ be a given accuracy parameter.
For $p\in [0,1]$, the sketch $H'_p$ is constructed as follows.
$H'_p$ is a bipartite graph where the vertices on the one part correspond to $\mathcal{S}$ and the vertices on the other part correspond to the elements $\{e\in\mathcal{E}\mid \hslash(e)\leq p\}$.
For an element $e\in\mathcal{E}$ with $\hslash(e)\leq p$, if the degree of $e$ in $G$ is at most $n\log(1/\varepsilon)/(\varepsilon k)$, we add all edges connected to $e$ in $G$ to $H'_p$, otherwise we add \emph{arbitrary} $n\log(1/\varepsilon)/(\varepsilon k)$ edges conneted to $e$ in $G$ to $H'_p$.
The properties of $H'_p$ is stated in the following lemma.
The randomness is over the choice of the hash function $\hslash$.

\begin{lemma}[\cite{bateni2017almost}]\label{lem:k-cover-offline}
For $\min\left(\frac{k\log(1/\delta)\log(n)}{\varepsilon^2\OPT_k(G)},1\right)\leq p\leq 1$, with probability at least $1-O(\delta)$, for any $\alpha\in [0,1]$ and any $\mathcal{P}\subseteq\mathcal{S}$ with $|\mathcal{P}|=k$ and $|\Gamma(H'_p,\mathcal{P})|\geq \alpha\cdot \OPT_k(H'_p)$, we have 
$
|\Gamma(G,\mathcal{P})| \geq (\alpha-3\varepsilon)  \cdot \OPT_k(G),
$
and $(1+\varepsilon)\cdot \OPT_k(G)\geq \frac{1}{p}\cdot |\Gamma(H'_p,\mathcal{P})|\geq  (\alpha-2\varepsilon)\cdot \OPT_k(G)$.
Furthermore, for any $C\geq 1$, if $p\leq C\cdot \frac{k\log(1/\delta)\log(n)}{\varepsilon^2\OPT_k(G)}$, then with probability at least $1-O(\delta)$, the number of edges in $H_p'$ is at most $C\cdot 100n\log(1/\delta)\log(1/\varepsilon)\log(n)/\varepsilon^3$.
\end{lemma}

\subsection{Bucketing-based Sketch for $k$-Cover Problem}
To design an efficient sliding window algorithm for $k$-cover, we only need to develop an efficient bucketing-based sketch according to Theorem~\ref{thm:framework}.
Fortunately, the sketch graph $H'_p$ described in Section~\ref{sec:kcover_offline_sketch} yields an efficient bucketing-based sketch.
In this section, we will show how to construct the bucketing-based sketch.
Notice that $\OPT_k(G)$ is always in $[1,W]$.
Let $o\in [1,W]$ and let $\delta$ and $\varepsilon$ be the same as described in Section~\ref{sec:kcover_offline_sketch}.
We construct the sketch function $Z(\cdot)$ as following.
$Z(\cdot)$ only contains one sub-sketch $Z_1(\cdot)$, i.e., $Z(\cdot)=(Z_1(\cdot))$.
For convenience, we abuse the notation and denote $Z_1(\cdot)$ as $Z(\cdot)$.
We use $Z(G)$ to denote $Z(X)$ for $X\subseteq \mathcal{S}\times \mathcal{E}$ where $G$ is the bipartite graph corresponding to the edge set $X$.
According to Definition~\ref{def:bucketing_based_sketch}, to describe $Z(\cdot)$, we only need to specify the filter function $h(\cdot)$, the bucketing function $g(\cdot)$, the processing function $\zeta(\cdot)$ and the threshold $T$. 
Let $\hslash$ be the same as described in Section~\ref{sec:kcover_offline_sketch}: an $O(k\log(1/\delta)\log n)$-wise independent hash function which maps each element in $\mathcal{E}$ to $[0,1]$ uniformly at random.
Let $p=\min\left(\frac{k\log(1/\delta)\log(n)}{\varepsilon^2\cdot o},1\right)$.
We construct $h(\cdot),g(\cdot),\zeta(\cdot)$ and $T$ as follows:
\begin{enumerate}
    \item $\forall (S,e)\in \mathcal{S}\times \mathcal{E}$, 
    \begin{align*}
        h(S,e) = \left\{
        \begin{array}{ll}
        1, & \hslash(e)\leq p,\\
        0, & \text{o.w.}
        \end{array}
        \right.
    \end{align*}
    \item Let $\mathcal{B}=\mathcal{E}$, i.e., each element in $\mathcal{E}$ corresponds to a bucket.
    $\forall (S,e)\in\mathcal{S}\times \mathcal{E},g(S,e)=e$.
    \item $\forall (S,e)\in\mathcal{S}\times\mathcal{E}, \zeta(S,e) = (S,e)$.
    \item $T = n\log(1/\varepsilon)/(\varepsilon k)$.
\end{enumerate}
\begin{lemma}\label{lem:recover_Hp}
The space budget of $Z(G)$ is the same as the number of edges of $H'_p$, and $H'_p$ can be constructed via $Z(G)$ in $O(|Z(G)|)$ time.
\end{lemma}
\begin{proof}
Since each item in $Z(G)$ is a tuple $((S,e),\zeta(S,e))$ and $\zeta(S,e)=(S,e)$, we abuse the notation and regard each item in $Z(G)$ as $(S,e)\in\mathcal{S}\times\mathcal{E}$, and it is clear that the space budget of $Z(G)$ is the same as the size of $Z(G)$.
It is easy to verify that there is a one-to-one correspondence between the items of $Z(G)$ and the edges of $H'_p$.
Consider an element $e\in\mathcal{E}$. 
If $\hslash(e)>p$, any edge $(S,e)$ will neither be added to $Z(G)$ ($(S,e)$ is filtered by the filter function $h(\cdot)$) nor $H'_p$.
If $\hslash(e)\leq p$ and the degree of $e$ in $G$ is at most $T=n\log(1/\varepsilon)/(\varepsilon k)$, every edge $(S,e)$ will be added into $Z(G)$ since the bucket $e$ will not be full, and $(S,e)$ will also be added into $H'_p$ due to the construction of $H'_p$.
If $\hslash(e)\leq p$ and the degree of $e$ in $G$ is larger than $T=n\log(1/\varepsilon)/(\varepsilon k)$, arbitrary $T$ edges $(S,e)$ will be added into $Z(G)$ since the bucket $e$ is full, and these edges are also added into $H'_p$ according to the construction of $H'_p$.
\end{proof}

\begin{theorem}[Bucketing-based sketch for $k$-cover under edge-arrival model]\label{thm:bucketsketch_kcover}
Consider the $k$-cover problem over a set of edges in $\mathcal{S}\times\mathcal{E}$ where $\mathcal{E}$ is a ground set of $m$ elements and $\mathcal{S}$ is a family of $n$ sets of elements.
For any $\varepsilon,\delta\in(0,0.5),o\geq 1,\alpha\in[0,1]$, there always exists an $o$-restricted $(\alpha-\varepsilon)$-approximate bucketing-based sketch $Z^o(\cdot)$ such that for any edge set $X\subseteq \mathcal{S}\times \mathcal{E}$,
\begin{enumerate}
    \item the recover function over $Z^o(X)$ outputs an $(\alpha-\varepsilon)$-approximation to the optimal $k$-cover value with probability at least $1-\delta$ conditioned on that the optimal value is at least $o$;
    \item if the optimal value is at most $2\cdot o$, with probability at least $1-\delta$, the space budget of $Z^o(X)$ is upper bounded by $O(n\log(1/\delta)\log(1/\varepsilon)\log(n)/\varepsilon^3)$.
\end{enumerate}
The time needed to compute a filter function, a bucketing function and a processing function on any $(S,e)\in \mathcal{S}\times \mathcal{E}$ is always upper bounded by $O(k\log(1/\delta)\log(n))$.
Furthermore, if $\alpha\leq 1-1/e$, the time to compute recover function over $Z^o(X)$ is at most $\wt{O}(|Z^o(X)|)$.
\end{theorem}
\begin{proof}
Let $G$ be the bipartite graph corresponding to the edge set $X$. 
We will prove that the sketch $Z(G)$ described in this section is the desired sketch $Z^o(X)$.

According to Lemma~\ref{lem:recover_Hp}, we can recover $H'_p$ for $p=\min\left(\frac{k\log(1/\delta)\log(n)}{\varepsilon^2\cdot o},1\right)$ via $Z(G)$ in $O(|Z(G)|)$ time.
If $\OPT_k(G)\geq o$, we have $p\geq \min\left(\frac{k\log(1/\delta)\log(n)}{\varepsilon^2\OPT_k(G)},1\right)$.
Let $\mathcal{P}'\subseteq \mathcal{S}$ with $|\mathcal{P}'|=k$ such that $|\Gamma(H'_p,\mathcal{P}')|\geq \alpha\cdot \OPT_k(H'_p)$.
According to Lemma~\ref{lem:k-cover-offline}, we have $(1+\varepsilon)\cdot \OPT_k(G)\geq \frac{1}{p}\cdot |\Gamma(H'_p,\mathcal{P}')|\geq (\alpha-O(\varepsilon))\cdot \OPT_k(G)$ with probability at least $1-\delta$.
Thus, $Z(G)$ is an $o$-restricted $(\alpha-O(\varepsilon))$-approximate bucketing-based sketch.
If $\alpha=1$, we can find $\mathcal{P}'$ by brute-force search which takes $O(n^k\cdot |Z^o(X)|)$ time.
If $\alpha =1-1/e$, we can find $\mathcal{P}'$ by running a well-known greedy algorithm~\cite{nemhauser1978analysis} on $H'_p$ which takes $\wt{O}(|Z(G)|)$ time.
If $\OPT_k(G)\leq 2\cdot o$, according to Lemma~\ref{lem:k-cover-offline}, with probability at least $1-O(\delta)$, the number of edges of $H'_p$ is at most $200n\cdot\log(1/\delta)\log(1/\varepsilon)\log(n)/\varepsilon^3)$.
According to Lemma~\ref{lem:recover_Hp}, the space budget of $Z(G)$ is the same as the number of edges of $H'_p$, and thus $|Z(G)|\leq O(n\log(1/\delta)\log(1/\varepsilon)\log(n)/\varepsilon^3)$.

The time to perform the filter function $h$ is the same as the time to evaluate the $O(k\log(1/\delta)\log(n))$-wise independent hash function $\hslash$, which is $O(k\log(1/\delta)\log(n))$.
The time to perform the bucketing function $g(\cdot)$ and the processing function $\zeta(\cdot)$ is $O(1)$.

After scaling $\delta$ and $\varepsilon$ by a constant factor, we conclude the proof.
\end{proof}

\subsection{$k$-Cover in Edge-Arrival Sliding Window Model}
By plugging the bucketing-based sketch for $k$-cover into our algorithmic framework (Algorithm~\ref{alg:sliding_framework}), we are able to obtain an efficient sliding window algorithm for $k$-cover problem.
\begin{theorem}[Sliding window $k$-cover with fast running time]\label{thm:sliding_window_kcover_fast_runtime}
For any $\varepsilon,\delta'\in(0,0.5)$, there is a $(1-1/e-\varepsilon)$-approximation algorithm for $k$-cover in the edge-arrival sliding window model with window size $W\geq 1$ using space $\wt{O}(n\log(1/\delta')\log(W)/\varepsilon^3)$.
The update time is $\wt{O}(k\log(1/\delta')\log(n)\log(W))$.
The running time to get the approximation at the end of the stream is at most $\wt{O}(n\log(1/\delta')\log(W)/\varepsilon^3)$.
The success probability is at least $1-\delta'$.
\end{theorem}
\begin{proof}
By plugging the bucketing-based sketch in Theorem~\ref{thm:bucketsketch_kcover} with approximation parameter $\alpha = 1-1/e$ and probability parameter $\delta=\Theta(\delta'/\log(W))$ into Algorithm~\ref{alg:sliding_framework}, we get the desired algorithm.

Since the window size is $W$, the optimal $k$-cover value is between $1$ and $W$.
According to Theorem~\ref{thm:bucketsketch_kcover}, $\forall o\in [1,W]$, there always exists an $o$-restricted $(1-1/e-\varepsilon)$-approximate bucketing-based sketch $Z^o(\cdot)$ for $k$-cover such that a $(1-1/e-\varepsilon)$-approximation to the optimal $k$-cover value can be recovered with probability at least $1-\delta$ conditioned on that the optimal $k$-cover value is at least $o$.
Furthermore, if the optimal $k$-cover value is at most $2\cdot o$, the space budget of the sketch is at most $O(n\log(1/\delta)\log(1/\varepsilon)\log(n)/\varepsilon^3)$ with probability at least $1-\delta$.
According to Theorem~\ref{thm:framework}, after plugging the bucketing-based sketches into Algorithm~\ref{alg:sliding_framework}, we obtain the sliding window algorithm which outputs a $(1-1/e-\varepsilon)$-approximation to the optimal $k$-cover value.
The space needed is at most $O(n\log(1/\delta)\log(1/\varepsilon)\log(n)/\varepsilon^3\cdot \log(W))=\wt{O}(n\log(1/\delta')\log(W)/\varepsilon^3)$.
The success probability is at least $1-O(\delta\cdot \log(W))=1-\delta'$.
According to Theorem~\ref{thm:bucketsketch_kcover}, the time needed to compute a filter function, a bucketing function and a processing function on any data item is always at most $O(k\log(1/\delta)\log(n))$.
By taking a more careful analysis, we can get a update time better than that claimed in Theorem~\ref{thm:framework}.
Consider an iteration of the loop in line~\ref{sta:start_process_stream}-\ref{sta:end_process_stream} of Algorithm~\ref{alg:sliding_framework}.
We only need to compute filter function, bucketing function and processing function $O(\log W)$ times since $|O|=O(\log W)$ and $t=1$.
It takes $O(k\log(1/\delta)\log(n)\log(W))$ time.
Since each tuple in the sketch has size $1$, the loop in line~\ref{sta:start_bucketing_step}-\ref{sta:end_bucketing_step} of Algorithm~\ref{alg:sliding_framework} has at most one iteration each time.
Similarly, the loop in line~\ref{sta:while_start}-\ref{sta:while_end} of Algorithm~\ref{alg:sliding_framework} has at most one iteration each time.
Thus, the overall update time is at most $O(k\log(1/\delta)\log(n)\log(W))=O(k\log(1/\delta')\log(n)\log(W)\log\log(W))$ time.
At the end of the stream, to get the final approximation, we need to evaluate the recover function $O(\log W)$ times according to Algorithm~\ref{alg:sliding_framework}.
Since the size of each sketch is at most $\wt{O}(n\log(1/\delta')/\varepsilon^3)$, the time needed to compute the recover function for each sketch is at most $\wt{O}(n\log(1/\delta')/\varepsilon^3)$ according to Theorem~\ref{thm:bucketsketch_kcover}.
Thus, the overall running time to get the approximation at the end of the stream is at most $\wt{O}(n\log(1/\delta')\log(W)/\varepsilon^3)$.
\end{proof}

\begin{theorem}[Sliding window $k$-cover with $(1-\varepsilon)$-approximation]\label{thm:sliding_window_kcover_better_approx}
For any $\varepsilon,\delta'\in(0,0.5)$, there is a $(1-\varepsilon)$-approximation algorithm for $k$-cover in the edge-arrival sliding window model with window size $W\geq 1$ using space $\wt{O}(n\log(1/\delta')\log(W)/\varepsilon^3)$.
The update time is at most $\wt{O}(k\log(1/\delta')\log(n)\log(W))$.
The success probability is at least $1-\delta'$.
\end{theorem}
\begin{proof}
The proof is almost the same as the proof of Theorem~\ref{thm:sliding_window_kcover_fast_runtime} except that we choose the approximation parameter $\alpha=1$ for the bucketing-based sketch in Theorem~\ref{thm:bucketsketch_kcover}.
\end{proof}

\section{Diversity Maximization in the Sliding Window Model}\label{sec:div_max_sliding_window}
We consider diversity maximization problems in $d$-dimensional Euclidean space. 
In these problems, we are give a set of points $P\subset \mathbb{R}^d$, and the goal is to find a subset $Q$ of $k$ points with the maximum diversity. 
The diversity of a set $Q,$ $\mathrm{div}(Q)$ depends on the distances between points in $Q$.
In this section, we study diversity maximization problems with various well-known diversity functions $\mathrm{div}(\cdot)$~\cite{indyk2014composable,borassi2019better}. 
In particular, the descriptions of these diversity functions are listed in Table~\ref{tab:diversity_functions}.

When $k,d$ and $\mathrm{div}(\cdot)$ are specified, we use $\OPT(P)$ to denote the optimal cost of $P$, i.e., $\OPT(P)=\max_{Q\subseteq P:|Q| = k} \mathrm{div}(Q)$.
Let $\wb{\OPT}(P)$ denote $\OPT(P)$ divided by the number of pairwise distances that are summed in the diversity function $\mathrm{div}(\cdot)$.
In particular, for remote-edge, $\wb{\OPT}(P)=\OPT(P)$; for remote-clique, $\wb{\OPT}(P)=\OPT(P)/(k(k-1)/2)$; for remote-tree and remote-star, $\wb{\OPT}(P)=\OPT(P)/(k-1)$; for remote-cycle, remote $t$-cycles and remote-pseudoforest, $\wb{\OPT}(P)=\OPT(P)/k$; for remote $t$-trees, $\wb{\OPT}(P)=\OPT(P)/(k-t)$; for remote-bipartition, $\wb{\OPT}(P)=\OPT(P)/(\lfloor k/2 \rfloor \cdot \lceil k/2 \rceil)$; and for remote-matching, $\wb{\OPT}(P)=\OPT(P)/(k/2)$.
The goal of estimating $\OPT(P)$ becomes to estimate $\wb{\OPT}(P)$.

\subsection{Offline Sketch via Discretization of the Space}\label{sec:offline_sketch_diversity}
In this section, we show how to get a subset of points (sketch) of $P$ such that a good approximate solution of the obtained subset of points is a good approximate solution for $P$.
Given a parameter $\mu\in\mathbb{R}_{>0}$, let $\mathcal{G}_{\mu}\subset \mathbb{R}^d$ denote the set of regular grid points with side length $\mu$, i.e., $\mathcal{G}_{\mu} = \{(x_1,x_2,\cdots,x_d)\in\mathbb{R}^d\mid \forall i\in [d], x_i=\mu\cdot y_i,y_i\in \mathbb{Z}\}$.
Let $\hslash_{\mu}:\mathbb{R}^d\rightarrow \mathcal{G}_{\mu}$ be a mapping constructed as the following:
\begin{align*}
    \forall x=(x(1),x(2),\cdots,x(d))\in\mathbb{R}^d,\hslash_{\mu}(x) = (\lfloor x(1)/\mu \rfloor, \lfloor x(2)/\mu \rfloor, \cdots \lfloor x(d)/\mu \rfloor)\cdot \mu.
\end{align*}
\begin{fact}\label{fac:net_point}
$\forall p\not= q\in\mathcal{G}_{\mu}$, $\|p-q\|_2\geq \mu$. $\forall p\in \mathbb{R}^d,\|p-\hslash_{\mu}(p)\|_2\leq \sqrt{d}\cdot \mu$.
\end{fact}
Let $T_{\mathrm{div}}\in\mathbb{Z}_{\geq 1}$ be a parameter which only depends on the type of the diversity function.
In particular $T_{\mathrm{div}} = 1$ for remote-edge, remote-tree, remote-cycle, remote $t$-trees and remote $t$-cycles, and $T_{\mathrm{div}} = k$ for remote-clique, remote-star, remote-bipartition, remote-pseudoforest and remote-matching.
The sketch $S_{\mu}$ is constructed as follows. 
$S_{\mu}$ is a subset of points of $P$.
Initialize $S_{\mu}$ to be \emph{arbitrary} $k$ points from $P$ to ensure $|S_{\mu}|\geq k$.
For each point $q\in\mathcal{G}_{\mu}$ in the image of $P$ under $\hslash_{\mu}$, if $|\{p\in P\mid \hslash_{\mu}(p)=q\}|\leq T_{\mathrm{div}}$, add all points $p\in P$ with $\hslash_{\mu}(p) = q$ into $S_{\mu}$, otherwise, add \emph{arbitrary} $T_{\mathrm{div}}$ points $p\in P$ with $\hslash_{\mu}(p) = q$ into $S_{\mu}$.
The properties of $S_{\mu}$ is stated in the following lemma. 

\begin{lemma}\label{lem:div-offline}
Let $\varepsilon \in (0,0.5)$. 
For $0<\mu\leq  \varepsilon\cdot \wb{\OPT}(P)/(10\sqrt{d})$, $\OPT(P)\geq \OPT(S_{\mu})\geq (1-\varepsilon)\cdot \OPT(P)$.
Furthermore, for any $C > 0$, if $\mu \geq C\cdot \varepsilon\cdot \wb{\OPT}(P)/(10\sqrt{d})$,  $|S_{\mu}|\leq k\cdot(2\sqrt{d}+20\sqrt{d}/(C\cdot \varepsilon) + 1)^d\cdot T_{\div}$+k.
\end{lemma}
\begin{proof}
Let us consider the size of $S_{\mu}$ when $\mu \geq C\cdot \varepsilon \cdot \wb{\OPT}(P)/(10\sqrt{d})$ for some $C>0$.
We claim that there exists a subset $Q\subseteq P$ with $|Q|< k$ such that $\forall p\in P,\min_{q\in Q}\|p-q\|_2\leq \wb{\OPT}(P)$.
We use the following procedure to find $Q$:
\begin{enumerate}
    \item Mark every point in $P$ as \emph{uncovered}. Let $Q\gets \emptyset$.
    \item Choose an arbitrary uncovered point $q\in P$ and add $q$ into $Q$.
    \item Mark every point $p\in P$ with $\|p-q\|_2\leq \wb{\OPT}(P)$ as \emph{covered}.
    \item Repeat above two steps until $P$ is empty.
\end{enumerate}
It is easy to verify that $\forall p\in P$, there exists $q\in Q$ such that $\|p-q\|_2\leq \wb{\OPT}(P)$.
If $|Q|\geq k$, we can choose a subset $Q'\subseteq Q$ such that $|Q'|=k$.
Notice that the pairwise distances between points in $Q'$ are always greater than $\wb{\OPT}(P)$ which implies $\div(Q')> \wb{\OPT}(P)$ and thus leads to a contradiction.
Let us consider the size of image set $I=\{\hslash_{\mu}(p) \mid p\in P\}$. 
By Fact~\ref{fac:net_point}, $\forall x\in I,\exists p\in P,\|x-p\|_2\leq \sqrt{d}\cdot \mu$ which implies that $\exists q\in Q, \|x-q\|_2\leq \sqrt{d}\cdot \mu +\wb{\OPT}(P)$.
By Fact~\ref{fac:net_point}, $\forall x,y\in I$, $\|x-y\|_2\geq \mu$.
Let $B(x,r)$ denote the ball centered at $x$ with radius $r$, i.e., $B(x,r)=\{y\in\mathbb{R}^d\mid \|x-y\|_2\leq r\}.$
Then, we know that $\bigcup_{x\in I} B(x,\mu/2) \subseteq \bigcup_{q\in Q} B(q, \sqrt{d}\cdot \mu + \wb{\OPT}(P) + \mu / 2)$.
By analyzing the volume of the balls, we have $|I|\leq |Q|\cdot (2\cdot \sqrt{d}+2\cdot \wb{\OPT}(P)/\mu + 1)^d\leq k\cdot (2\sqrt{d}+ 20\sqrt{d}/(C\cdot \varepsilon ) + 1)^d$.
Since for each $x\in I$, we keep at most $T_{\div}$ points $p\in P$ with $\hslash_{\mu}(p) = x$, the size of $S_{\mu}$ is at most  $k\cdot(2\sqrt{d}+20\sqrt{d}/(C\cdot \varepsilon) + 1)^d\cdot T_{\div} + k$.

Next, let us prove the approximation.
Firstly, consider the cases for remote-clique, remote-star, remote-bipartition, remote-pseudoforest and remote-matching.
In these cases, $T_{\div} = k$.
Let $Q^* = \{q_1^*, q_2^*,\cdots, q_k^*\}$ be the optimal solution, i.e., $\div(Q^*)=\OPT(P)$.
We can find a set $Q'=\{q'_1,q'_2,\cdots,q'_k\}\subseteq S_{\mu}$ such that $\forall i\in [k],\hslash{\mu}(q^*_i)=\hslash_{\mu}(q'_i)$.
By fact~\ref{fac:net_point}, we have $\forall i\in[k], \|q^*_i-q'_i\|_2\leq 2 \sqrt{d}\cdot \mu$.
Thus, $\forall i,j\in[k]$, $ \|q^*_i-q^*_j\|_2\leq \|q'_i-q'_j\|_2 + 4\sqrt{d}\cdot \mu$.
Since $\mu\leq \varepsilon\cdot \wb{\OPT}(P)/(10\sqrt{d})$, we have
\begin{align*}
\wb{\OPT}(P)\geq \wb{\OPT}(S_{\mu})\geq \div(Q')\geq \wb{\OPT}(P) - 4\sqrt{d}\cdot \mu \geq (1-\varepsilon)  \cdot \wb{\OPT}(P)
\end{align*}
which implies that $\OPT(P)\geq \OPT(S_{\mu})\geq (1-\varepsilon)\cdot \OPT(P)$.

Then consider the cases for remote-edge, remote-tree, remote-cycle, remote $t$-trees and remote $t$-cycles.
Let $Q^* = \{q_1^*, q_2^*,\cdots, q_k^*\}$ be the optimal solution, i.e., $\div(Q^*)=\OPT(P)$.
Let $Q'=\{q'_1,q'_2,\cdots,q'_k\}$ be a multi-set such that $\forall i\in[k],\hslash_{\mu}(q^*_i)=\hslash_{\mu}(q'_i)$.
By the similar argument, we can show that $\div(Q')\geq \wb{\OPT}(P)-4\sqrt{d}\cdot \mu\geq (1-\varepsilon)\cdot \wb{\OPT}(P)$.
Notice that if we replace the duplicated points with arbitrary points, the diversity does not decrease. 
So we can find a set $Q''\subseteq S_{\mu}$ such that $\div(Q'')\geq \div(Q')$.
Thus, $\OPT(P)\geq \OPT(S_{\mu})\geq \div(Q'')\geq (1-\varepsilon)\cdot \OPT(P)$.
\end{proof}

\subsection{Bucketing-based Sketch for Diversity Maximization}
To design an efficient sliding window algorithm for diversity maximization, we need to develop an efficient bucketing-based sketch according to Theorem~\ref{thm:framework}.
In this section, we show how to construct $S_{\mu}$ described in Section~\ref{sec:offline_sketch_diversity} via bucketing-based sketch.
We suppose that the point set $P\subseteq [\Delta]^d$.
If $\wb{\OPT}(P)\not = 0,$ we have $\wb{\OPT}(P)\in[1,\sqrt{d}\cdot \Delta]$.
In this section, we consider the case when $\wb{\OPT}(P)>0$.
We will handle the case when $\wb{\OPT}(P)=0$ in our final sliding window diversity maximization algorithm (see the proof of Theorem~\ref{thm:sliding_window_div}).
Let $o\in[1,\sqrt{d}\cdot \Delta]$.
Let $\varepsilon\in(0,0.5)$.
Let $\mu  = \varepsilon\cdot o/(10\sqrt{d})$.
Let $\hslash_{\mu}(\cdot)$, $\mathcal{G}_{\mu}$ and $T_{\div}$ be the same as described in Section~\ref{sec:offline_sketch_diversity}.
We construct our bucketing-based sketch $Z(\cdot)$ as following.
$Z(\cdot)$ is composed by two sub-sketches $(Z_1(\cdot),Z_2(\cdot))$.
The goal of $Z_1(P)$ is to maintain arbitrary $k$ points of $P$.
The goal of $Z_2(P)$ is to maintain other points in $S_{\mu}$.
According to Definition~\ref{def:bucketing_based_sketch}, we need to specify the filter functions $h_1(\cdot),h_2(\cdot)$, the bucketing functions $g_1(\cdot),g_2(\cdot)$, the processing functions $\zeta_1(\cdot),\zeta_2(\cdot)$ and the thresholds $T_1,T_2$.
For $Z_1(\cdot)$, we construct $h_1(\cdot),g_1(\cdot),\zeta_1(\cdot)$ and $T_1$ as follows:
\begin{enumerate}
\item $\forall x\in \mathbb{R}^d,h_1(x)=1$, i.e., no points are filtered.
\item $Z_1(\cdot)$ only has one bucket, i.e., $\forall x\in\mathbb{R}^d$, $g_1(\cdot)$ always maps $x$ to the same bucket in $Z_1(\cdot)$.
\item $\forall x\in \mathbb{R}^d,\zeta_1(x)=x$.
\item $T_1=k$.
\end{enumerate}
For $Z_2(\cdot)$, we construct $h_2(\cdot),g_2(\cdot),\zeta_2(\cdot)$ and $T_2$ as follows:
\begin{enumerate}
\item $\forall x\in\mathbb{R}^d,h_2(x)=1$, i.e., no points are filtered.
\item The buckets $\mathcal{B}=\mathcal{G}_{\mu}$, i.e., each grid point in $\mathcal{G}_{\mu}$ corresponds to a bucket. 
$\forall x\in\mathbb{R}^d,g_2(x)=\hslash_{\mu}(x)$.
\item $\forall x\in \mathbb{R}^d,\zeta_2(x)=x$.
\item $T_2=T_{\div}$.
\end{enumerate}

\begin{lemma}\label{lem:recover_Smu}
$|Z(P)|=|S_{\mu}|$, the space budget of $Z(P)$ is at most $O(|S_{\mu}|\cdot d)$, and $S_{\mu}$ can be constructed via $Z(P)$ in the time linear in the size of $Z(P)$.
\end{lemma}
\begin{proof}
Since $\zeta_1(x)=\zeta_2(x)=x$, we abuse the notation and regard each item in $Z_1(P)$ and $Z_2(P)$ as a point $x$ itself instead of the pair $(x,\zeta_1(x))$ or $(x,\zeta_2(x))$.
The construction of $S_{\mu}$ described in Section~\ref{sec:offline_sketch_diversity} is equivalent to the following process.
\begin{enumerate}
    \item Initialize $S_{\mu}=\emptyset$.
    \item Add all points stored in $Z_1(P)$ into $S_{\mu}$.
    Since the threshold $T_1=k$ and $g_1(\cdot)$ maps points into the same bucket, this step corresponds to adding arbitrary $k$ points from $P$ into $S_{\mu}$.
    \item For each bucket $q\in\mathcal{G}_{\mu}$, add all points $x\in Z_2(P)$  with $g_2(x)=q$ into $S_{\mu}$.
    Since the bucketing function $g_2(x)=\hslash_{\mu}(x)$ and the threshold $T_2=T_{\div}$, this step corresponds to adding all points $p\in P$ with $\hslash_{\mu}(p)=q$ into $S_{\mu}$ if $|\{p\in P\mid \hslash(p)=q\}|\leq T_{\div}$ and adding $T_{\div}$ points $p\in P$ with $\hslash_{\mu}(p)=q$ into $S_{\mu}$ otherwise.
\end{enumerate}
Therefore, $|Z(P)|=|S_{\mu}|$.
Since each point has dimension $d$, the space budget of $Z(P)$ is at most $O(|S_{\mu}|\cdot d)$ 
The construction time is linear.
\end{proof}

\begin{theorem}[Bucketing-based sketch for diversity maximization]\label{thm:bucketsketch_diversity}
Consider any diversity maximization problem over a set of points $P\subseteq [\Delta]^d$ with parameter $k\geq 1$.
Let $T_{\div} = 1$ if the diversity function is remote-edge, remote-tree, remote-cycle, remote $t$-trees or remote $t$-cycles, and let $T_{\div}=k$ if the diversity function is remote-clique, remote-star, remote-bipartition, remote-pseudoforest or remote-matching.
For any $\varepsilon\in(0,0.5),o\geq 1,\alpha\in[0,1]$, there always exists an $o$-restricted $(\alpha-\varepsilon)$-approximate bucketing-based sketch $Z^o(\cdot)$ such that for any $P\subseteq [\Delta]^d$,
\begin{enumerate}
\item the recover function over $Z^o(P)$ outputs an $(\alpha-\varepsilon)$-approximation to $\OPT(P)$ if $\wb{\OPT}(P)\geq o$;
\item if $\wb{\OPT}(P)\leq 2\cdot o$, the space budget of $Z^o(P)$ is at most $k\cdot T_{\div}\cdot O(\sqrt{d}/\varepsilon)^d\cdot d$.
\end{enumerate}
The time needed to compute a filter function, a bucketing function and a processing function on any $x\in\mathbb{R}^d$ is always upper bounded by $O(d)$.
The time to compute recover function over $Z^o(P)$ is at most $\mathcal{T}(|Z^o(P)|)$ where $\mathcal{T}(n')$ denotes the running time needed to compute an $\alpha$-approximation for the diversity maximization for a point set with $n'$ points in $[\Delta]^d$.
\end{theorem}
\begin{proof}
We will prove that the sketch $Z(P)$ described in this section is the desired sketch $Z^o(P)$.

According to Lemma~\ref{lem:recover_Smu}, we can recover $S_{\mu}$ for $\mu= \varepsilon\cdot o/(10\sqrt{d})$ via $Z(P)$ in time linear in the size of $Z(P)$.
If $\wb{\OPT}(P)\geq o$, we have $\mu \leq \varepsilon \cdot \wb{\OPT}(P)/(10\sqrt{d})$.
According to Lemma~\ref{lem:div-offline}, we have $\OPT(P)\geq \OPT(S_{\mu})\geq (1-\varepsilon)\cdot \OPT(P)$.
According to Lemma~\ref{lem:recover_Smu}, we have $|S_{\mu}|=|Z(P)|$.
We can use $\mathcal{T}(|Z(P)|)$ time to find a subset $S\subseteq S_{\mu}\subseteq P$ with $|S|=k$ such that $\OPT(P)\geq \div(S)\geq \alpha \OPT(S_{\mu})\geq \alpha(1-\varepsilon)\OPT(P)\geq (\alpha-\varepsilon)\OPT(P)$.
If $\wb{\OPT}(P)\leq 2\cdot o$, according to Lemma~\ref{lem:div-offline}, we have $|S_{\mu}|\leq k\cdot T_{\div}\cdot O(\sqrt{d}/\varepsilon)^d$.
Thus the time needed to compute recover function over $Z(P)$ is at most $\mathcal{T}(|Z(P)|)$.
According to Lemma~\ref{lem:recover_Smu}, the space budget of $Z(P)$ is at most $O(|S_{\mu}|\cdot d)$, and thus the space budget of $Z(P)$ is at most $k\cdot T_{\div}\cdot O(\sqrt{d}/\varepsilon)^d\cdot d$.

Since each point $x$ has dimension $d$, the time to perform a filter function, a bucketing function and a processing function on $x$ is $O(d)$.
\end{proof}

\subsection{Sliding Window Algorithm for Diversity Maximization}
By plugging the bucketing-based sketch for diversity maximization into our algorithmic framework (Algorithm~\ref{alg:sliding_framework}), we are able to obtain an efficient sliding window algorithm for diversity maximization.

\begin{theorem}\label{thm:sliding_window_div}
For any $\varepsilon\in (0,0.5),\alpha\in [0,1]$ and any diversity function listed in Table~\ref{tab:diversity_functions}, there is a $(\alpha-\varepsilon)$-approximate algorithm for the diversity maximization problem with $k\geq 1$ for a point set from $[\Delta]^d$ in the sliding window model with window size $W\geq 1$ using space $kdT_{\div}\log(d\Delta)\cdot O(\sqrt{d}/\varepsilon)^d$ where $T_{\div}=1$ if the diversity function is remote-edge, remote-tree, remote-cycle, remote $t$-trees or remote $t$-cycles, and $T_{\div}=k$ if the diversity function is remote-clique, remote-star, remote-bipartition, remote-pseudoforest or remote-matching.
The update time is at most $O(d\log(d\Delta))$.
The running time to get the approximation at the end of the stream is at most $\mathcal{T}(k\cdot T_{\div}\cdot O(\sqrt{d}/\varepsilon)^d)\cdot O(\log(d\Delta))$ where $\mathcal{T}(n')$ denotes the running time needed to compute an $\alpha$-approximation for the diversity maximization for a point set with $n'$ points in $[\Delta]^d$.
\end{theorem}
\begin{proof}
Suppose the point set of interest at the end of the stream is $P\subseteq[\Delta]^d$.
Firstly, we need to handle the case if $\OPT(P)=0$.
According to the definition of the diversity functions, $\OPT(P)=0$ can only happen when the number of distinct points in $P$ is at most $k-1$.
We can use the following sliding window procedure to check whether $P$ has at least $k$ points, and if the number of distinct points is at most $k-1$, we can retrieve either all points or at least $T_{\div}$ points at each distinct point.
\begin{enumerate}
    \item Initialize a list of points $L=\emptyset$.
    \item For the latest $x$ in the stream:
        \begin{enumerate}
            \item Add $x$ into $L$.
            \item If there are $T_{\div}+1$ points in $L$ that are equal to $x$, remove the point which is equal to $x$ with the earliest timestamp from $L$.
            \item Otherwise, if $L$ contains $kT_{\div}+1$ points, remove the point with the earliest timestamp from $L$.
        \end{enumerate}
    \item At timestamp $N$, let $L'=\{x\in L\mid~\text{the timestamp is at least }N-W+1,\text{ i.e., }x\in P\}$.
    \item If $L'$ has at most $k-1$ distinct points, return $\OPT(L')$ as $\OPT(P)$.
\end{enumerate}
It is easy to verify that $L$ has the following two properties: 1) if the timestamp of $x$ is earlier than the timestamp of $y$ and $x$ is in $L$, then either $y$ is in $L$ or there are $T_{\div}$ points $y'$ in $L$ such that $y'=y$ and the timestamp of every $y'$ is later than the timestamp of $y$;
2) For each $x\in L$, the number of points $x'\in L$ (including $x$ itself) that is equal to $x$ is at most $T_{\div}$.
Since $L$ can contain at most $kT_{\div}$ points and at most $T_{\div}$ duplicates can be stored in $L$ for each distinct point, if $P$ has at least $k$ distinct points, then $L$ must contain at least $k$ distinct points.
Next, we are going to prove that $L'$ also contains at least $k$ distinct points.
Suppose $x_{i_1},x_{i_2},\cdots,x_{i_k}$ with $i_1<i_2<\cdots<i_k$ are $k$ distinct points in the stream such that $\forall j\geq i_1$, $x_j$ in the stream is equal to $x_{i_q}$ for some $q\in[k]$, and $\forall q\in[k],j>i_q,$ $x_j$ is distinct from $x_{i_q}$.
Since $P$ contains at least $k$ distinct points, all  $x_{i_1},x_{i_2},\cdots,x_{i_k}$ must in $P$, i.e., $i_1\geq N-W+1$.
If $\exists q\in[k]$ such that $x_{i_q}$ is not in $L$, then since $L$ contains at least $k$ distinct points, there must be $x_j\in L$ with $j<i_q$ which contradicts to the first property of $L$.
Thus, all $x_{i_1},x_{i_2},\cdots,x_{i_k}$ must in $L$ which implies that $L'$ contains at least $k$ distinct points $x_{i_1},x_{i_2},\cdots,x_{i_k}$ since $i_1\geq N-W+1$.
Next, we discuss the case when $P$ contains at most $k-1$ distinct points.
For $j\geq N-W+1$, if $x_j$ is not in $L'$, then $x_j$ is not in $L$. 
It implies that there are at least $T_{\div}$ $j'$ with $j'>j$ such that $x_{j'}$ is equal to $x_j$ and $x_{j'}$ is in $L'$ according to the first property of $L$.
Thus, for each distinct point in $P$, $L'$ contains either all duplicates or at least $T_{\div}$ duplicates.
According to the definition of diversity functions in Table~\ref{tab:diversity_functions}, we can verify that $\OPT(L')=\OPT(P)$.
The space needed is $O(kT_{\div}d)$, and the update time is $O(d)$.

In the remaining of the proof, we only need to discuss the case when $\OPT(P)>0$.
In this case, we have $\wb{\OPT}(P)\in[1,\sqrt{d}\Delta]$.
According to Theorem~\ref{thm:bucketsketch_diversity}, $\forall o\in [1,\sqrt{d}\Delta],\alpha\in[0,1],\varepsilon\in(0,0.5)$, there always exists an $o$-restricted $(\alpha-\varepsilon)$-approximate bucketing-based sketch $Z^o(\cdot)$ for diversity maximization such that an $(\alpha-\varepsilon)$-approximation to $\OPT(P)$ can be recovered if $\wb{\OPT}(P)\geq o$.
Furthermore, if $\wb{\OPT}(P)\leq 2\cdot o$, the space budget of the sketch is at most $k\cdot T_{\div}\cdot O(\sqrt{d}/\varepsilon)^d\cdot d$.
According to Theorem~\ref{thm:framework}, after plugging the bucketing-based sketch into Algorithm~\ref{alg:sliding_framework}, we obtain the sliding window algorithm which outputs a $(\alpha-\varepsilon)$-approximation to $\OPT(P)$.
The space needed is at most $k\cdot T_{\div}\cdot O(\sqrt{d}/\varepsilon)^d\cdot d\cdot \log(d\Delta)$.
According to Theorem~\ref{thm:bucketsketch_diversity}, the time needed to compute a filter function, a bucketing function and a processing function on any data item is always at most $O(d)$.
By taking a more careful analysis, we can get a update time better than that claimed in Theorem~\ref{thm:framework}.
Consider an iteration of the loop in line~\ref{sta:start_process_stream}-\ref{sta:end_process_stream} of Algorithm~\ref{alg:sliding_framework}.
We only need to compute filter function, bucketing function and processing function $O(\log(d\Delta))$ times since $|O|=O(\log(d\Delta))$ and $t=2$.
It takes $O(d\log(d\Delta))$ time.
Since each tuple in the sketch has the same size $\Theta(d)$, the loop in line~\ref{sta:start_bucketing_step}-\ref{sta:end_bucketing_step} of Algorithm~\ref{alg:sliding_framework} has at most one iteration each time.
Similarly, the loop in line~\ref{sta:while_start}-\ref{sta:while_end} of Algorithm~\ref{alg:sliding_framework} has at most one iteration each time.
Thus, the overall update time is at most $O(d\log(d\Delta))$ time.
At the end of the stream, to get the final approximation, we need to evaluate the recover function $O(\log(d\Delta))$ times according to Algorithm~\ref{alg:sliding_framework}.
Since each sketch contains at most $k\cdot T_{\div}\cdot O(\sqrt{d}/\varepsilon)^d$ points, the time needed to compute the recover function for each sketch is at most $\mathcal{T}(k\cdot T_{\div}\cdot O(\sqrt{d}/\varepsilon)^d)$.
Thus, the overall running time to get the approximation at the end of the stream is at most $\mathcal{T}(k\cdot T_{\div}\cdot O(\sqrt{d}/\varepsilon)^d)\cdot O(\log(d\Delta))$.

\end{proof}

\section{$k$-Clustering in the Sliding Window Model}
In the $\ell_p$ $k$-clustering problem $(p\in[1,\infty))$, we consider the data universe $\mathcal{X}$ as the (discretized) $d$-dimensional Euclidean sapce $[\Delta]^d$.
Given a point set $X \subseteq [\Delta]^d$ and a parameter $k\in \mathbb{Z}_{\geq 1}$, the goal is to find a subset of $k$ centers $B\subset \mathbb{R}^d$ $(|B|\leq k)$ such that the clustering cost 
\begin{align*}
\cost(X,B):=\sum_{x\in X}\min_{b\in B}\|x-b\|_2^p
\end{align*}
is minimized.
Notice that this problem formulation is a general case of $k$-median clustering $(p=1)$ and $k$-means clustering $(p=2)$.
One popular way to solve $k$-clustering problem in large-scale computational models is to use the coreset technique.
A coreset is a small weighted subset of data points which approximately preserves the clustering cost for any set of centers. 
The formal definition of $\varepsilon$-coreset is defined as the following.

\begin{definition}[$\varepsilon$-coreset]\label{def:coreset}
Let $k\in \mathbb{Z}_{\geq 1},p\in[1,\infty)$. 
Given a set of points $X\subseteq [\Delta]^d$, if a subset $S\subseteq X$ together with the weights $w:S\rightarrow \mathbb{R}_{\geq 0}$ satisfies that 
\begin{align*}
\forall B\subset \mathbb{R}^d \text{ with } |B|\leq k, (1-\varepsilon)\cdot \cost(X,B)\leq \cost(S,w,B)\leq (1+\varepsilon) \cdot \cost(X,B),
\end{align*}
where $\cost(S,w,B):=\sum_{x\in S} w(x)\cdot \min_{b\in B}\|x - b\|_2^p$, then $(S,w)$ is called an $\varepsilon$-coreset of $X$ for the $\ell_p$ $k$-clustering problem.
\end{definition}

In the remaining subsections, we will first review an offline coreset construction algorithm and then we will show how to construct a bucketing-based sketch (see Definition~\ref{def:bucketing_based_sketch}) for an $\varepsilon$-coreset.
Thus, it will imply an efficient sliding window algorithm.

\subsection{Offline Coreset Construction}\label{sec:offline_coreset_k_clustering}
The offline construction is almost the same as the algorithm proposed by~\cite{hu2018nearly}.
We put all analysis into Appendix~\ref{sec:analysis_offline_k_clustering} for completeness.

Suppose the data set is a set of at most $n$ points $X\subseteq [\Delta]^d$.
We first partition the space $[\Delta]^d$ via a hierarchical grid structure~\cite{chen2009coresets,braverman2017clustering,hu2018nearly}.
We sample a random vector $v\in \mathbb{R}^d$ such that each entry $v(i)$ is drawn independently and uniformly at random from $[0,\Delta]$ 
Then we impose a standard hierarchical grids shifted by $-v$.
Let $L=\lceil\log(nd\Delta)\rceil+10$.
The grids have $L+1$ levels.
$\forall i\in\{0,1,\cdots,L\}$, the grid $G_i$ partitions $\mathbb{R}^d$ into cells with side length $\Delta_i = \Delta/2^i$.
In particular, 
\begin{align*}
&\forall i\in \{0,1,\cdots,L\},\\
&G_i=\{C\mid C=[t_1\Delta_i-v(1),(t_1+1)\Delta_i-v(1))\times\cdots\times[t_d-v(d)\Delta_i,(t_d+1)\Delta_i-v(d)), t_1,\cdots,t_d\in \mathbb{Z}\}.
\end{align*}
It is easy to see that each cell $C\in G_i$ is always partitioned by $2^d$ cells in $G_{i+1}$.
For convenience in notation, we also define $G_{-1}$ in the same way where each cell in $G_{-1}$ has side length $\Delta_{-1}=2\Delta$.
It is easy to verify that there is a unique cell in $G_{-1}$ which contains $[\Delta]^d$ entirely.
Consider two cells $C\in G_i$ and $C'\in G_j$.
If $C\subset C'$, we call $C'$ an ancestor of $C$.
If $C$ is a cell in $G_i$ for $i\in \{0,1,\cdots,L\}$ and a cell $C'\in G_{i-1}$ is an ancestor of $C$, then we say that $C$ is a child cell of $C'$.
Consider a point $x\in [\Delta]^d$.
If the cell containing $x$ in $G_i$ is $C$, we define $c_i(x):=C$.
Let $\OPT$ denote the optimal $l_p$ $k$-clustering cost for a point set $X$, i.e.:
$
\OPT:=\min_{B\subset \mathbb{R}^d:|B|\leq k} \cost(X,B).
$
The $\varepsilon$-coreset construction is shown in Algorithm~\ref{alg:offline_coreset_construction} which uses Algorithm~\ref{alg:offline_heavy_cell_decomposition} as a subroutine.

\begin{algorithm}[ht]
	\small
	\begin{algorithmic}[1]\caption{Heavy Cell Partitioning}\label{alg:offline_heavy_cell_decomposition}
	  \STATE \textbf{Predetermined:} $o\geq 1$: a guess of $\OPT$.
	  \STATE \textbf{Input:} A point set $X\subseteq [\Delta]^d$ with at most $n$ points.
	  \STATE Impose randomly shifted grids $G_{-1},G_0,G_1,\cdots,G_L$ with $L:=\lceil \log(nd\Delta)\rceil+10$ levels.
	  \STATE For $i\in \{-1,0,1,\cdots,L\}$, set the threshold $R_i:= 0.01\cdot o/(\sqrt{d}\Delta_i)^p$. {\hfill //$R_{i+1}=R_{i}\cdot2^p$.}
	  \STATE Find the unique cell $C\in G_{-1}$ such that $[\Delta]^d\subseteq C$ and mark $C$ as \emph{heavy}. \label{sta:handle_minusone_level}
	  \FOR{$i:=0\rightarrow L$}
	    \STATE For each cell $C\in G_i$, let $\lambda(C)$ be an estimation of $|C\cap X|$, the number of points in the cell $C$.\label{sta:point_estimation_in_each_cell}
	    \STATE If $\lambda(C)\geq R_i$ and all the ancestors of $C$ are heavy, mark $C$ as heavy. \label{sta:mark_heavy_via_estimation}
	  \ENDFOR
	  \STATE Output all heavy cells.
	\end{algorithmic}
\end{algorithm}

\begin{algorithm}[ht]
	\small
	\begin{algorithmic}[1]\caption{Coreset Construction for $\ell_p$ $k$-Clustering}\label{alg:offline_coreset_construction}
	  \STATE \textbf{Input:} A point set $X\subseteq [\Delta]^d$ with at most $n$ points, and $\delta,\varepsilon\in (0,0.5)$. \label{sta:start_offline_coreset}
	  \STATE Run Algorithm~\ref{alg:offline_heavy_cell_decomposition} and obtain all heavy cells.\label{sta:call_of_heavy_decomposition}
	  \STATE Let $L$, $G_{-1},G_0,\cdots,G_L$ and $R_{-1},R_0,\cdots,R_L$ be the same as in Algorithm~\ref{alg:offline_heavy_cell_decomposition}. \label{sta:post_exe_call_partition} 
	  \STATE For $i\in\{0,1,\cdots,L\}$, for $C\in G_i$, if $C$ is not heavy but all ancestors of $C$ are heavy, mark $C$ as \emph{crucial}.
	  \STATE For $i\in\{0,1,\cdots,L\}$, let $X^i=\{x\in X\mid x\text{ is in a crucial cell in }G_i\}$.
	  Let $\lambda(X^i)$ be an estimation of $|X^i|$.\label{sta:point_estimation_in_each_crucial_level}
	  \STATE Let $I=\{i\mid i\in \{0,1,\cdots,L\},\lambda(X^i)\geq \gamma R_i\}$, where $\gamma:=\varepsilon/(40\cdot 2^{2p+2}\cdot L)$.
	  \STATE $t':=\sum_{i\in I} \lambda(X^i)\cdot \min(2^{2p+1}/R_i,1), m:=\lceil 1000t'\varepsilon^{-2}(\log n\log (2t') + \log 1/\delta)\rceil$, and initialize $S=\emptyset$.\label{sta:mid_offline_coreset}
	  \WHILE{repeat $m$ times} \label{sta:repeated_sampling_start}
	    \STATE Sample a level $i\in I$ with probability $(\lambda(X^i)\cdot \min(2^{2p+1}/R_i,1))/t'$.\label{sta:level_sampling}
	    \STATE Uniformly sample a point $x\in X^i$. \label{sta:uniform_sampling}
	    \STATE Add the point $x$ into $S$ and set the weight $w(x):=t'/(m\cdot \min(2^{2p+1}/R_i,1))$.
	  \ENDWHILE \label{sta:repeated_sampling_end}
	  \STATE Output $(S,w)$ as the coreset.
	\end{algorithmic}
\end{algorithm}

The guarantee of Algorithm~\ref{alg:offline_coreset_construction} is shown in the following theorem.
We include the proof of Theorem~\ref{thm:offline_k_clustering_coreset} in Appendix~\ref{sec:analysis_offline_k_clustering} for completeness.
\begin{theorem}[Generalization of Theorem 8 of~\cite{hu2018nearly}]\label{thm:offline_k_clustering_coreset}
Suppose for every $i\in \{0,1,\cdots, L\}$ and every cell $C\in G_i$, the estimated value $\lambda(C)$ in line~\ref{sta:point_estimation_in_each_cell} of Algorithm~\ref{alg:offline_heavy_cell_decomposition} satisfies $\lambda(C)\in |C\cap X|\pm 0.1R_i$ or $\lambda(C)\in (1\pm 0.01)\cdot |C\cap X|$, and for every $i\in\{0,1,\cdots, L\}$, the estimated value $\lambda(X^i)$ in line~\ref{sta:point_estimation_in_each_crucial_level} of Algorithm~\ref{alg:offline_coreset_construction} satisfies $\lambda(X^i)\in |X^i|\pm 0.1\varepsilon\gamma R_i$ or $\lambda(X^i)\in(1\pm 0.01\varepsilon)\cdot |X^i|$.
If the parameter $o$ in Algorithm~\ref{alg:offline_heavy_cell_decomposition} satisfies $o\in [1,\OPT]$, then $(S,w)$ outputted by Algorithm~\ref{alg:offline_coreset_construction} is an $\varepsilon$-coreset of $X$ for the $\ell_p$ $k$-clustering problem with probability at least $1-\delta$.
\end{theorem}

\subsection{Coreset Construction via bucketing-based Sketches}
To simulate Algorithm~\ref{alg:offline_heavy_cell_decomposition} and Algorithm~\ref{alg:offline_coreset_construction} in the sliding window model, We only need to construct a bucketing-based sketch to simulate the algorithms.
\subsubsection{Heavy Cell Partitioning via Bucketing-based Sketches}\label{sec:heavy_cell_partition_via_bucket_based_sketches}
We first describe a bucketing-based sketch to simulate Algorithm~\ref{alg:offline_heavy_cell_decomposition}.
Our sketch $Z(\cdot)$ is composed by $L+1$ sub-sketches $Z_0(\cdot),Z_1(\cdot),\cdots,Z_L(\cdot)$.
Let $\delta\in (0,1)$ be a probability parameter.
To describe $Z_i(\cdot)$ for $i\in\{0,1,\cdots, L\}$, we need to specify the filter function $h_i(\cdot)$, the bucketing function $g_i(\cdot)$, the processing function $\zeta_i(\cdot)$ and the threshold $T_i$ (see Definition~\ref{def:bucketing_based_sketch}):
\begin{enumerate}
\item $\forall x\in [\Delta]^d$, $h_i(x)$ are independent, and 
\begin{align*}
h_i(x)=\left\{\begin{array}{ll}1, & \text{w.p. }p_i=\min(1, 10^6\ln(nL/\delta)/R_i),\\ 0, & \text{o.w.}\end{array}\right.
\end{align*}
\item Let $\mathcal{B}$ be all cells of randomly shifted grids generated by Algorithm~\ref{alg:offline_heavy_cell_decomposition}, i.e., $\mathcal{B}=G_0\cup G_1\cup\cdots \cup G_L$. 
Then, $g_i(x)$ is the cell in $G_i$ that contains $x$, i.e.,
\begin{align*}
\forall x\in[\Delta]^d,g_i(x) = C:C\in G_i,x\in C.
\end{align*}
\item $\zeta_i(x)=x$.
\item $T_i = 10\cdot p_i\cdot R_i\cdot d$.
\end{enumerate}

\begin{lemma}[Simulation of Algorithm~\ref{alg:offline_heavy_cell_decomposition} via bucketing-based sketch]\label{lem:heavy_cell_decomposition_via_bucketing_based_sketch}
Algorithm~\ref{alg:offline_heavy_cell_decomposition} can be simulated via sketch $Z(X)$ described above. 
Furthermore, with probability at least $1-\delta$, $\forall i\in\{0,1,\cdots,L\},C\in G_i$, the estimated value $\lambda(C)$ in line~\ref{sta:point_estimation_in_each_cell} of Algorithm~\ref{alg:offline_heavy_cell_decomposition} during the simulation satisfies $\lambda(C)\in |C\cap X|\pm 0.1 R_i$ or $\lambda(C)\in (1\pm 0.01)\cdot |C\cap X|$.
The running time of the simulation is $O(|Z(X)|)$.
\end{lemma}
\begin{proof}
In the proof, we will show how to use $Z(X)$ to simulate Algorithm~\ref{alg:offline_heavy_cell_decomposition}.
Let us define $\forall i\in\{0,1,\cdots,L\},C\in G_i,$ $\lambda(C):=|\{x\in X\mid h_i(x)=1,g_i(x)=C\}|/p_i.$
Notice that we do not need to explicitly compute $\lambda(C)$ during the simulation of Algorithm~\ref{alg:offline_heavy_cell_decomposition}.

We first prove that $\lambda(C)$ is a good approximation to $|C\cap X|$.
If $p_i=1$, it is obvious that $\lambda(C)=|C\cap X|$.
We only need to consider the case when $p_i<1$.
Consider $i\in \{0,1,\cdots,L\}$ and cell $C\in G_i$.
We have $\E[|\{x\in X\mid h_i(x)=1,g_i(x)=C\}|] = p_i\cdot |C\cap X|$.
There are two cases. 
If $|C\cap X| >R_i$, by Bernstein inequality, we have:
\begin{align*}
&\Pr\left[\left||\{x\in X\mid h_i(x)=1,g_i(x)=C\}|-p_i\cdot |C\cap X|\right|\geq 0.01\cdot p_i\cdot |C\cap X|\right]\\
\leq & 2 \exp\left(- \frac{ 0.5\cdot 0.01^2\cdot p_i^2\cdot |C\cap X|^2 }{p_i\cdot |C\cap X|+\frac{1}{3}\cdot 0.01\cdot p_i\cdot |C\cap X|}\right)\\
\leq & \delta/(n(L+1)),
\end{align*}
where the last inequality follows from that $|C\cap X|>R_i$ and $p_i\geq 10^6 \ln(nL/\delta)/R_i$.
In the second case, we have $|C\cap X|\leq R_i$.
In this case, by Bernstein inequality, we have:
\begin{align*}
&\Pr\left[\left||\{x\in X\mid h_i(x)=1,g_i(x)=C\}|-p_i\cdot |C\cap X|\right|\geq 0.1\cdot p_i\cdot R_i\right]\\
\leq & 2 \exp\left(- \frac{ 0.5\cdot 0.1^2\cdot p_i^2\cdot R_i^2 }{p_i\cdot |C\cap X|+\frac{1}{3}\cdot 0.1\cdot p_i\cdot R_i}\right)\\
\leq & \delta/(n(L+1)),
\end{align*}
where the last inequlaity follows from that $|C\cap X|\leq R_i$ and $p_i \geq 10^4  \ln(nL/\delta)/ R_i$.
Notice that $\lambda(C)=|\{x\in X\mid h_i(x)=1,g_i(x)=C\}|/p_i$.
Thus, by taking union bound over all $i\in\{0,1,\cdots, L\}$ and all non-empty cells $C\in G_i$, with probability at least $1-\delta$, $\forall i\in \{0,1,\cdots, L\},C\in G_i$, $\lambda(C)$ satisfies $\lambda(C)\in |C\cap X|\pm 0.1R_i$ or $\lambda(C)\in (1\pm 0.01)\cdot |C\cap X|$.

Now let us go back to Algorithm~\ref{alg:offline_heavy_cell_decomposition}. 
All steps of Algorithm~\ref{alg:offline_heavy_cell_decomposition} can be implemented without any information of $X$ except line~\ref{sta:point_estimation_in_each_cell} and line~\ref{sta:mark_heavy_via_estimation}.
Let us focus on line~\ref{sta:point_estimation_in_each_cell} and line~\ref{sta:mark_heavy_via_estimation} of Algorithm~\ref{alg:offline_heavy_cell_decomposition}.
The goal to determine whether $\lambda(C)\geq R_i$. 
Since $\lambda(C)=|\{x\in X\mid h_i(x)=1,g_i(x)=C\}|/p_i$, it is equivalent to determine whether $|\{x\in X\mid h_i(x)=1,g_i(x)=C\}|\geq p_i\cdot R_i$.
Let $\bar{\lambda}(C)$ be $1/p_i$ times the number of tuples $(x,\zeta_i(x))$ stored in $Z_i(X)$.
Notice that $\bar{\lambda}(C)$ can be computed during the simulation via the sketch $Z(X)$.
Consider the condition~\ref{it:bucket_sketch_2b} of Definition~\ref{def:bucketing_based_sketch}.
There are two cases. 
In the first case, $\sum_{x\in X:h_i(x)=1,g_i(x)=C}|\zeta_i(x)|\leq T_i$.
In this case, $\forall x\in X$ with $h_i(x)=1,g_i(x)=C$, we have $(x,\zeta_i(x))\in Z_i(x)$.
Therefore, we have $\bar{\lambda}(C)=\lambda(C)$ and thus line~\ref{sta:point_estimation_in_each_cell} and line~\ref{sta:mark_heavy_via_estimation} of Algorithm~\ref{alg:offline_heavy_cell_decomposition} can be successfully simulated.
In the second case, $\sum_{x\in X:h_i(x)=1,g_i(x)=C}|\zeta_i(x)|> T_i$.
By our definition of $\lambda(C)$, it implies that $\lambda(C)\geq R_i$.
In this case, it is also easy to verify that $\bar{\lambda}(C)\geq R_i$ according to the condition~\ref{it:bucket_sketch_2b} of Definition~\ref{def:bucketing_based_sketch}.
Thus, the outcome of using $\bar{\lambda}(C)$ to simulate line~\ref{sta:point_estimation_in_each_cell} and line~\ref{sta:mark_heavy_via_estimation} of Algorithm~\ref{alg:offline_heavy_cell_decomposition} is always the same as the outcome of using $\lambda(C)$.

Finally let us consider the running time of the simulation process. The overall running time is linear in the number of times that line~\ref{sta:point_estimation_in_each_cell} and line~\ref{sta:mark_heavy_via_estimation} are executed. 
Observe that each tuple $(x,\zeta_i(x))\in Z_i(X)$ can be considered at most once. 
The total running time is at most linear in the size of $Z(X)$.

\end{proof}

\begin{lemma}[Size of the sketch $Z(X)$]\label{lem:size_heavy_cell_partioning_sketch}
With probability at least $1-2\delta$, the space budget of $Z(X)$ is at most $O((kd+d^{O(p)}\cdot \OPT/o)\cdot L\log(nL/\delta)/\delta)$.
\end{lemma}

Before proving Lemma~\ref{lem:size_heavy_cell_partioning_sketch}, let us prove the following useful lemma. 
Let $B^*\subset \mathbb{R}^d$ be the optimal centers for $X$, i.e., $\cost(X,B^*)=\OPT$.
A \emph{center cell} $C\in G_i$ in level $i\in\{0,1,\cdots,L\}$ is a cell such that $\inf_{x\in C,y\in B^*}\|x-y\|_2\leq \Delta_i/(2d)$
The following Lemma shows that the total number of center cells cannot be too large since the grids are randomly shifted. 
This is an observation made by~\cite{chen2009coresets,braverman2017clustering,hu2018nearly}. 
We include the proof here for completeness.
\begin{lemma}\label{lem:center_cells}
With probability at least $1-\delta$, the total number of center cells is at most $6kL/\delta$.
\end{lemma}
\begin{proof}
Let $B^*=\{b_1^*,b_2^*,\cdots,b_k^*\}$.
Consider the grid $G_i$ in level $i\in \{0,1,\cdots, L\}$.
Let $Y_{j,\alpha}$ denote an indicator variable of the event that the Euclidean distance from $b_j$ to the boundary of $G_i$ in the $\alpha$-th dimension is at most $\Delta_i/(2d)$.
Notice that if a center is close to a boundary of $G_i$ in a dimension, it will contribute a factor of at most $2$ to the number of center cells.
Thus, the number of cells that have distance to $b^*_j$ at most $\Delta_i/(2d)$ is at most
\begin{align*}
2^{\sum_{\alpha=1}^d Y_{j,\alpha}}.
\end{align*}
Notice that
\begin{align*}
\E\left[2^{\sum_{\alpha=1}^d Y_{j,\alpha}}\right]=\prod_{\alpha=1}^d\E[2^{Y_{j,\alpha}}]= (1+1/d)^d\leq e\leq 3,
\end{align*}
where the second inequality follows from that $\Pr[Y_{j,\alpha}=1]= (2\cdot \Delta_i/(2d))/\Delta_i=1/d$ and $2^{Y_{j,\alpha}}=1+Y_{j,\alpha}$.
Since there are $k$ centers, the expected total number of center cells in $G_i$ is at most $3\cdot k$.
Since there are $L+1$ levels. the expected total number of center cells over all levels is at most $3k\cdot (L+1)\leq 6kL$.
By Markov's inequality, with probability at least $1-\delta$, the number of center cells in all levels is at most $6kL/\delta$.
\end{proof}

\paragraph{Proof of Lemma~\ref{lem:size_heavy_cell_partioning_sketch}.} We are going to prove Lemma~\ref{lem:size_heavy_cell_partioning_sketch}.
The space budget of $Z(X)$ can be upper bounded as the following:
\begin{align}
&\text{The space budget of $Z(X)$}\notag\\
=&\sum_{i=0}^L \sum_{C\in G_i} \min\left(T_i,\sum_{x\in X:h_i(x)=1,g_i(x)=C} |\zeta_i(x)|\right)\notag\\
\leq &\sum_{i=0}^L \left(\sum_{C\in G_i:C\text{ is a center cell}} T_i +\sum_{C\in G_i:C\text{ is not a center cell}} |\{x\in X\mid h_i(x)=1,g_i(x)=C\}| \cdot d\right)\notag\\
= & \sum_{i=0}^L \sum_{C\in G_i:C\text{ is a center cell}} T_i\label{eq:first_part}\\
&+d\cdot \sum_{i=0}^L \sum_{x\in X:x\text{ is not in any center cell in }G_i} h_i(x).\label{eq:second_part}
\end{align}
According to Lemma~\ref{lem:center_cells}, part~\eqref{eq:first_part} is at most $6kL/\delta\cdot 10^7\ln(nL/\delta)\cdot d$ with probability at least $1-\delta$. In the remaining of the proof let us bound part~\eqref{eq:second_part}.

Consider $i\in\{0,1,\cdots,L\}$. 
If $x\in X$ is not in any center cell in $G_i$, then $\min_{b\in B^*} \|x-b\|_2\geq \Delta_i/(2d)$.
Thus, $|\{x\in X\mid x\text{ is not in any center cell in }G_i\}|\leq \frac{\OPT}{(\Delta_i/(2d))^p}\leq 100\cdot \left(2d^{1.5}\right)^p\cdot R_i\cdot \frac{\OPT}{o}$ where the last inequality follows from that $R_i=0.01o/(\sqrt{d}\Delta_i)^p$.
Since $p_i=\min(1,10^6\ln(nL/\delta)/R_i)$, we have
\begin{align*}
\E\left[ \sum_{x\in X:x\text{ is not in any center cell in }G_i} h_i(x)\right]\leq 10^8\cdot (2d^{1.5})^p\cdot \ln(nL/\delta)\cdot \frac{\OPT}{o}.
\end{align*}
By Markov's inequality, with probability at least $1-\delta$, we have:
\begin{align*}
\sum_{i=0}^L \sum_{x\in X:x\text{ is not in any center cell in }G_i} h_i(x) \leq 10^8\cdot (2d^{1.5})^p\cdot (L+1)\cdot \ln(nL/\delta)/\delta\cdot \frac{\OPT}{o}.
\end{align*}
Thus, we can bound part~\eqref{eq:second_part} by $2\cdot 10^8\cdot (2d^{1.5})\cdot dL\cdot \ln(nL/\delta)/\delta\cdot \frac{\OPT}{o}$.
By combining the part~\eqref{eq:first_part} with the part~\eqref{eq:second_part}, we conclude the proof.
\qed

\subsubsection{Sampling Process via Bucketing-based Sketches}\label{sec:sampling_process}
Now we will describe the simulation of Algorithm~\ref{alg:offline_coreset_construction} via bucketing-based sketches.
We divided Algorithm~\ref{alg:offline_coreset_construction} into two parts.
The first part contains line~\ref{sta:start_offline_coreset}-\ref{sta:mid_offline_coreset} of Algorithm~\ref{alg:offline_coreset_construction}.
The second part contains remaining steps of Algorithm~\ref{alg:offline_coreset_construction}.

Let us first consider the first part. 
In Section~\ref{sec:heavy_cell_partition_via_bucket_based_sketches}, we described how to simulate Algorithm~\ref{alg:offline_heavy_cell_decomposition} via bucketing-based sketches.
Thus, to simulate line~\ref{sta:start_offline_coreset}-\ref{sta:mid_offline_coreset} of Algorithm~\ref{alg:offline_coreset_construction}, we only need to obtain $\lambda(X^i)$ for all $i\in\{0,1,\cdots,L\}$ via bucketing-based sketches.
We describe a bucketing-based sketch $Z'(\cdot)$ as the following.
Our sketch $Z'(\cdot)$ is composed by $L+1$ sub-sketches $Z'_0(\cdot),Z'_1(\cdot),\cdots,Z'_L(\cdot)$.
Let $\delta\in(0,0.5)$ be a probability parameter.
Let $\varepsilon\in(0,0.5)$ be an error parameter.
To describe $Z'_i(\cdot)$ for $i\in\{0,1,\cdots,L\}$, we need to specify the filter function $h'_i(\cdot)$, the bucketing function $g'_i(\cdot)$, the processing function $\zeta_i'(\cdot)$ and the threshold $T'_i$ (see Definition~\ref{def:bucketing_based_sketch}):
\begin{enumerate}
    \item $\forall x\in [\Delta]^d,h'_i(x)$ are independent, and
    \begin{align*}
        h'_i(x)=\left\{\begin{array}{ll}1, & \text{w.p. } p'_i=\min(1,10^6\ln(nL/\delta)/(\varepsilon^2\gamma R_i)),\\ 0 & \text{o.w.}\end{array}\right.
    \end{align*}
    \item Let $\mathcal{B}$ be all cells of grids described by Algorithm~\ref{alg:offline_coreset_construction}, i.e., $\mathcal{B}=G_0\cup G_1\cup \cdots \cup G_L$. 
    Then $g'_i(x)$ is the cell in $G_i$ that contains $x$, i.e.,
    \begin{align*}
        \forall x\in[\Delta]^d,g'_i(x)=C:C\in G_i,x\in C.
    \end{align*}
    \item $\zeta'_i(x)=x$.
    \item $T'_i=10\cdot p'_i\cdot R_i\cdot d$.
\end{enumerate}
Note that $Z'(\cdot)$ is very similar to $Z(\cdot)$.
The main differences are that $h'_i$ has a higher sampling rate than $h_i$ and $T'_i>T_i$.

\begin{lemma}[Obtaining $\lambda(X^i)$ of Algorithm~\ref{alg:offline_coreset_construction} via bucketing-based sketch]\label{lem:simulate_sampling_first_part}
If during the simulation of Algorithm~\ref{alg:offline_heavy_cell_decomposition} in line~\ref{sta:call_of_heavy_decomposition} of Algorithm~\ref{alg:offline_coreset_construction}, $\forall i\in \{0,1,\cdots,L\},C\in G_i$ $\lambda(C)$ satisfies $\lambda(C)\in |C\cap X|\pm 0.1R_i$ or $\lambda(C)\in (1\pm 0.01) \cdot |C\cap X|$, then $\forall i\in \{0,1,\cdots,L\}$, $\lambda(X^i)$ can be obtained via $Z'(X)$ such that with probability at least $1-2\delta$, $\forall i\in \{0,1,\cdots,L\},$ $\lambda(X^i)\in |X^i|\pm 0.1\varepsilon \gamma R_i$ or $\lambda(X^i) \in (1\pm 0.01\varepsilon)\cdot |X^i|$.
Line~\ref{sta:post_exe_call_partition}-\ref{sta:mid_offline_coreset} of Algorithm~\ref{alg:offline_coreset_construction} can be simulated once $\lambda(X^i)$ for all $i\in\{0,1,\cdots,L\}$ are obtained.
The running time to simulate line~\ref{sta:post_exe_call_partition}-\ref{sta:mid_offline_coreset} of Algorithm~\ref{alg:offline_coreset_construction} is at most $O(|Z'(X)|+L)$.
\end{lemma}

\begin{proof}
Consider $i\in \{0,1,\cdots,L\}$ and $C\in G_i$.
If $C$ is a crucial cell, since $\lambda(C)$ is a good estimation of $|C\cap X|$, we know that $|C\cap X|\leq 1.1 R_i$.
Notice that $\E[|\{x\in X\mid h'_i(x)=1,g'_i(x)=C\}|]=p'_i\cdot |C\cap X|$.
By Bernstein inequality, we have:
\begin{align*}
&\Pr\left[|\{x\in X\mid h'_i(x)=1,g'_i(x)=C\}|-p'_i\cdot |C\cap X|\geq 0.1\cdot p'_i\cdot R_i\right]\\
\leq & \exp\left(-\frac{0.5\cdot 0.1^2\cdot {p'_i}^2\cdot R_i^2}{p'_i\cdot |C\cap X|+\frac{1}{3}\cdot 0.1\cdot p'_i\cdot R_i}\right)\\
\leq & \delta/(n(L+1)),
\end{align*}
where the last inequality follows from that $|C\cap X|\leq 1.1R_i$ and $p'_i\geq 10^4\ln(nL/\delta)/R_i$.
Thus, for a crucial cell $C$, with probability at least $1-\delta/(n(L+1))$, $|\{x\in X\mid h'_i(x)=1,g'_i(x)=C\}|\leq 1.1\cdot p_i'\cdot R_i$.
Thus, by taking union bound over all non-empty crucial cells, with probability at least $1-\delta$, $\forall i\in\{0,1,\cdots,L\},\forall C\in G_i$ that is a crucial cell, $|\{x\in X\mid h'_i(x)=1,g'_i(x)=C\}|\leq 1.1\cdot p'_i\cdot R_i$.
Thus, with probability at least $1-\delta$, $\forall i\in\{0,1,\cdots,L\},\forall C\in G_i$ that is a crucial cell, $\sum_{x\in X: h'_i(x)=1,g'_i(x)=C}|\zeta'_i(x)|\leq T'_i$.
According to condition~\ref{it:bucket_sketch_2b} of Definition~\ref{def:bucketing_based_sketch} and the construction of $X^i$ in Algorithm~\ref{alg:offline_coreset_construction}, it implies that with probability at least $1-\delta$, $\forall i\in \{0,1,\cdots, L\},\forall x\in X^i$, if $h'_i(x)=1$, then $(x,\zeta'_i(x))\in Z'_i(X)$.

Notice that for each $(x,\zeta'_i(x))\in Z'_i(X)$, we are able to know whether $x\in X^i$ since we know all heavy cells.
By the construction of $X^i$, $x\in X^i$ if and only if $x$ is not in a heavy cell of $G_i$ and $x$ is in a heavy cell of $G_{i-1}$.
Thus, we are able to compute $\lambda(X^i)=|\{(x,\zeta'_i(x))\in Z'_i(X)\mid x\in X^i\}|/p_i'=|\{x\in X^i\mid h'_i(x)=1\}|/p'_i$.
Next, we will show that $\lambda(X^i)$ is a good approximation to $|X^i|$.
There are two cases.
If $|X^i|>\gamma R_i$, by Bernstein inequality, we have:
\begin{align*}
&\Pr[\left| |\{x\in X^i\mid h'_i(x)=1\}| - p'_i\cdot |X^i| \right| \geq 0.01\cdot \varepsilon\cdot p'_i\cdot |X^i|]\\
\leq & 2 \exp\left(-\frac{0.5\cdot 0.01^2\cdot \varepsilon^2\cdot {p'_i}^2\cdot |X^i|^2}{p'_i\cdot |X^i|+\frac{1}{3}\cdot 0.01\cdot \varepsilon\cdot p'_i\cdot |X^i|}\right)\\
\leq & \delta/(L+1),
\end{align*}
where the last inequality follows from that $|X^i|>\gamma R_i$ and $p'_i\geq 10^6\ln(L/\delta)/(\varepsilon^2\gamma R_i)$.
In the second case, we have $|X^i|\leq \gamma R_i$.
In this case, by Bernstein inequality, we have:
\begin{align*}
&\Pr\left[\left||\{x\in X^i\mid h'_i(x)=1\}|-p'_i\cdot |X^i|\right|\geq 0.1\cdot \varepsilon\cdot p'_i\cdot \gamma\cdot R_i\right]\\
\leq & \exp\left(-\frac{0.5\cdot 0.1^2\cdot \varepsilon^2 \cdot {p'_i}^2\cdot \gamma^2\cdot R_i^2}{p'_i\cdot |X^i| + \frac{1}{3}\cdot 0.1\cdot \varepsilon\cdot p'_i\cdot \gamma\cdot R_i }\right)\\
\leq & \delta/(L+1),
\end{align*}
where the last inequality follows from that $|X^i|\leq \gamma R_i$ and $p'_i\geq 10^4\ln(L/\delta)/(\varepsilon^2\gamma R_i)$.
Thus, by taking union bound over all $i\in\{0,1,\cdots,L\}$, with probability at least $1-\delta$, $\forall i\in\{0,1,\cdots,L\},\lambda(X^i)$ satisfies $\lambda(X^i)\in |X^i|\pm 0.1 \varepsilon \gamma R_i$ or $\lambda(X^i)\in (1\pm 0.01\varepsilon)|X^i|$.

By using such $\lambda(X^i)$, we are able to simulate line~\ref{sta:post_exe_call_partition}-line~\ref{sta:mid_offline_coreset} of Algorithm~\ref{alg:offline_coreset_construction}.
It is easy to verify that we only need to scan each item in $Z'(X)$ at most once to compute $\lambda(X^i)$ for all $i\in\{0,1,\cdots,L\}$.
Thus the total running time to simulate line~\ref{sta:post_exe_call_partition}-line~\ref{sta:mid_offline_coreset} of Algorithm~\ref{alg:offline_coreset_construction} is at most $O(L+|Z'(X)|)$.
\end{proof}

\begin{lemma}[Size of $Z'(X)$]\label{lem:size_of_Zprime}
With probability at least $1-2\delta$, the space budget of $Z'(X)$ is at most $O((kd+d^{O(p)}\cdot \OPT/o)\cdot 2^{O(p)}L^2\log(nL/\delta)/(\delta\varepsilon^3))$.
\end{lemma}
\begin{proof}
The proof is similar to the proof of Lemma~\ref{lem:size_heavy_cell_partioning_sketch}.
We will still use the concept of center cells (see Lemma~\ref{lem:center_cells}).
\begin{align}
&\text{The space budget of }Z'(X)\notag\\
=&\sum_{i=0}^L\sum_{C\in G_i}\min\left(T'_i,\sum_{x\in X:h'_i(x)=1,g'_i(x)=C} |\zeta'_i(x)|\right)\notag\\
\leq & \sum_{i=0}^L\left(\sum_{C\in G_i:C\text{ is a center cell}} T_i + \sum_{C\in G_i:C\text{ is not a center cell}} |\{x\in X\mid h'_i(x)=1,g'_i(x)=C\}|\cdot d\right)\notag\\
=&\sum_{i=0}^L \sum_{C\in G_i:C\text{ is a center cell}} T'_i \label{eq:prime_part_one}\\
& +d\cdot \sum_{i=0}^L \sum_{x\in X:x\text{ is not in any center cell in }G_i} h'_i(x). \label{eq:prime_part_two}
\end{align}
According to Lemma~\ref{lem:center_cells}, part~\eqref{eq:prime_part_one} is at most $6kL/\delta\cdot 10^7\ln(nL/\delta)/(\varepsilon^2\gamma)\cdot d$ with probability at least $1-\delta$.
In the remaining of the proof let us upper bound part~\eqref{eq:prime_part_two}.

Recall that $B^*$ is the optimal centers for $X$, i.e. $B^*$ satisfies $\cost(X,B^*)=\OPT$.
Consider $i\in \{0,1,\cdots,L\}$.
If $x\in X$ is not in any center cell in $G_i$, then $\min_{b\in B^*} \|x-b\|_2\geq \Delta_i/(2d)$.
Thus, $|\{x\in X\mid x \text{ is not in any center cell in }G_i\}|\leq \frac{\OPT}{(\Delta_i/(2d))^p}\leq 100\cdot (2d^{1.5})^p\cdot R_i\cdot \frac{\OPT}{o}$ where the last inequality follows from that $R_i=0.01o/(\sqrt{d}\Delta_i)^p$.
Since $p'_i=\min(1,10^6\ln(nL/\delta)/(\varepsilon^2\gamma R_i))$, we have
\begin{align*}
\E\left[\sum_{x\in X:x\text{ is not in any center cell in }G_i} h'(x)\right] \leq 10^8\cdot (2d^{1.5})^p\cdot \ln(nL/\delta)/(\varepsilon^2\gamma)\cdot \frac{\OPT}{o}.
\end{align*}
By Markov's inequality, with probability at least $1-\delta$, we have:
\begin{align*}
\sum_{i=0}^L \sum_{x\in X:x\text{ is not in any center cell in }G_i} h'_i(x) \leq 10^8\cdot (2d^{1.5})^p\cdot \ln(nL/\delta)/(\varepsilon^2\gamma)\cdot \frac{\OPT}{o}\cdot (L+1)/\delta.
\end{align*}
Thus, we can upper bound part~\eqref{eq:prime_part_two} by $10^8\cdot (2d^{1.5})^p\cdot \ln(nL/\delta)/(\varepsilon^2\gamma)\cdot \frac{\OPT}{o}\cdot d(L+1)/\delta$.
By combining with part~\eqref{eq:prime_part_one} with part~\eqref{eq:prime_part_two}, and since $\gamma = \Theta(\varepsilon/(2^{O(p)}L))$ we conclude the proof.
\end{proof}

In the remaining of the section, we will describe how to use a bucketing-based sketch to simulate the sampling procedure shown in line~\ref{sta:repeated_sampling_start}-\ref{sta:repeated_sampling_end} of Algorithm~\ref{alg:offline_coreset_construction}.
The most challenging part is line~\ref{sta:uniform_sampling} since each time we need to draw a uniform sample from $X^i$ while we cannot explicitly store the entire $X^i$.
We use the following idea to draw a uniform sample from $X^i$: we choose a small proper sampling probability $p_i''$  and we construct independent random subsets $Y^i_1, Y^i_2,\cdots,Y^i_{\hat{m}}\subseteq X^i$, where $\forall j\in [\hat{m}]$, each $x\in X^i$ is independently added into $Y_j^i$ with probability $p_i''$. 
Since $p_i''$ is very small, $Y_j^i$ may be an empty set.
But for a non-empty set $Y_j^i$, if we choose a uniform sample from $Y_j^i$, then such sample is also a uniform sample from $X^i$.
Thus, when each time we want to draw an independent uniform sample from $X^i$, we just choose an arbitrary unused non-empty set $Y^i_j$ and draw a uniform sample from $Y^i_j$ as a uniform sample from $X^i$.
Since $p_i''$ is properly chosen, we may afford to maintain all the subsets $Y_1,Y_2,\cdots, Y^i_{\hat{m}}$ via our sketch and the number of non-empty subsets is large enough.
In the following, we formalize the above idea by first introducing an another bucketing-based sketch $Z''(\cdot)$.

Our sketch $Z''(\cdot)$ is composed by $L+1$ sub-sketches $Z''_0(\cdot),Z''_1(\cdot),\cdots,Z''_L(\cdot)$ where $Z''_i(\cdot)$ will be used to handle the uniform sampling over $X^i$.
Let $\delta\in (0,0.5)$ be a probability parameter.
Let $\varepsilon\in(0,0.5)$ be an error parameter.
To describe $Z''_i(\cdot)$ for $i\in\{0,1,\cdots,L\}$, we need to specify the filter function $h''_i(\cdot)$, the bucketing function $g''_i(\cdot)$, the processing function $\zeta_i''(\cdot)$ and the threshold $T''_i(\cdot)$ (see Definition~\ref{def:bucketing_based_sketch}):
\begin{enumerate}
\item $\forall x\in[\Delta]^d,$ $h''_i(x)$ is determined by $\zeta''_i(x)$: $h''_i(x)=1$ if $\zeta''_i(x)\not=\emptyset$; $h''_i(x)=0$ otherwise.
\item Let $\mathcal{B}$ be all cells of grids described by Algorithm~\ref{alg:offline_coreset_construction}, i.e., $\mathcal{B}=G_0\cup G_1\cup\cdots\cup G_L$.
Then $g''_i(x)$ is the cell in $G_i$ that contains $x$, i.e., 
\begin{align*}
\forall x\in[\Delta]^d,g''_i(x)=C:C\in G_i,x\in C.
\end{align*}
\item Let 
\begin{align*}
\hat{m}= 10^{9}\cdot 2^{2p+2}\cdot \varepsilon^{-3}\cdot (kL+(2d^{1.5})^p)\cdot (p\log n\log(kLd)+\log(1/\delta))\cdot \frac{L}{\delta}.
\end{align*}
Let $\zeta''_i(x)$ be a random subset of $[\hat{m}]$ such that $\forall j\in [\hat{m}]$, $j$ is added into $\zeta''_i(x)$ independently with probability $p''_i= \min((2000\cdot (kL+(2d^{1.5})^p)\cdot R_i)^{-1},1)$.
\item $T''_i= 10 \cdot p''_i\cdot \hat{m}\cdot R_i\cdot d$.
\end{enumerate}

We describe the simulating process of line~\ref{sta:repeated_sampling_start}-\ref{sta:repeated_sampling_end} via $Z''(X)$ as the following:
\begin{enumerate}
    \item For each $i\in\{0,1\cdots,L\},j\in[\hat{m}]$, construct $\hat{Y}^i_j=\{x\in X^i\mid (x,\zeta''_i(x))\in Z''(X),j\in \zeta''_i(x)\}$.
    Notice that we are able to verify whether $x\in X^i$ given all heavy cells. 
    Due to the construction of $X^i$, $x\in X^i$ if and only if $x$ is in a non-heavy cell in $G_i$ and is in a heavy cell in $G_{i-1}.$
    \item Initialize $A^0=A^1=\cdots=A^L=[\hat{m}]$.
    \item Repeat $m$ times:
        \begin{enumerate}
            \item Sample a level $i\in I$ with probability $\lambda(X^i)\cdot \min(2^{2p+1}/R_i,1)/t'$.
            \item Pick an arbitrary $j\in A^i$ satisfying $\hat{Y}^i_j\not=\emptyset$. 
            If such $j$ does not exist, return \textbf{FAIL}.
            \item Uniformly sample a point $x\in \hat{Y}^i_j$ and remove $j$ from $A^i$, i.e., $A^i\gets A^i\setminus\{j\}$.
            \item Add the point $x$ into $S$ and set the weight $w(x):=t'/(m\cdot \min(2^{2p+1}/R_i,1))$.
        \end{enumerate}
\end{enumerate}

For the convenience of notation in our analysis, for $i\in\{0,1,\cdots,L\},j\in[\hat{m}]$, we define $Y^i_j=\{x\in X^i \mid j\in \zeta_i''(x)\}$.
It is clear that if $Y^i_j$ is non-empty, then a uniform sample from $Y^i_j$ is a uniform sample for $X^i$.
Notice that $\hat{Y}^i_j\subseteq Y^i_j$, a uniform sample from $\hat{Y}^i_j$ is a unform sample from $X^i$ only when $\hat{Y}^i_j=Y^i_j$.

\begin{lemma}[$\hat{Y}^i_j=Y^i_j$ with a good probability]\label{lem:hatYijisgood}
If line~\ref{sta:start_offline_coreset}-\ref{sta:mid_offline_coreset} of Algorithm~\ref{alg:offline_coreset_construction} are simulated such that $\forall i\in\{0,1,\cdots,L\},C\in G_i$, $\lambda(C)\in |C\cap X|\pm 0.1 R_i$ or $\lambda(C)\in (1\pm 0.01)\cdot |C\cap X|$ and $\forall i\in\{0,1,\cdots,L\}$, $\lambda(X^i)\in |X^i|\pm 0.1\varepsilon\gamma R_i$ or $\lambda(X^i)\in(1\pm 0.01\varepsilon)\cdot |X^i|$, then with probability at least $1-\delta$, $\forall i\in\{0,1,\cdots,L\},j\in [\hat{m}],\hat{Y}^i_j=Y^i_j$.
\end{lemma}
\begin{proof}
Suppose all $\lambda(C)$ and $\lambda(X^i)$ are good enough as claimed in the statement.

Consider $i\in\{0,1,\cdots,L\}$ and a crucial cell $C\in G_i$.
Since $\lambda(C)$ is a good estimation of $|C\cap X|$, we know that $|C\cap X|\leq 1.1R_i$.
Notice that
\begin{align*}
\E\left[\sum_{x\in X:h''_i(x)=1,g''_i(x)=C} |\zeta_i''(x)|\right] = p_i''\cdot \hat{m}\cdot  |C\cap X|.
\end{align*}
By Bernstein inequality, we have:
\begin{align*}
&\Pr\left[\sum_{x\in X:h''_i(x)=1,g''_i(x)=C} |\zeta_i''(x)| - p_i''\cdot \hat{m}\cdot |C\cap X|\geq 0.1 \cdot p_i''\cdot \hat{m} \cdot R_i\right]\\
\leq& \exp\left(-\frac{0.5\cdot 0.1^2\cdot {p_i''}^2\cdot \hat{m}^2\cdot R_i^2}{p_i''\cdot \hat{m}\cdot |C\cap X| + \frac{1}{3}\cdot 0.1\cdot {p_i''}\cdot \hat{m}\cdot R_i}\right)\\
\leq& \delta/(n(L+1)),
\end{align*}
where the last inequality follows from that $|C\cap X|\leq 1.1R_i$ and $p_i''\cdot \hat{m}\geq 10^4\ln(nL/\delta)/R_i$
By taking union bound over all $i\in \{0,1,\cdots,L\}$ and all non-empty crucial cells $C\in G_i$, with probability at least $1-\delta$,
\begin{align*}
\sum_{x\in X:h''_i(x)=1,g''_i(x)=C} |\zeta''_i(x)| \leq 1.2\cdot p_i''\cdot\hat{m}\cdot R_i \leq T''_i.
\end{align*}
Thus, according to the condition~\ref{it:bucket_sketch_2b} of Definition~\ref{def:bucketing_based_sketch}, $\forall x\in X^i,$ if $\zeta''_i(x)\not=\emptyset$, then $(x,\zeta''_i(x))\in Z''_i(X)$.
According to the construction of $\hat{Y}_j^i$, we have $\hat{Y}_j^i=Y^i_j$.
\end{proof}

\begin{lemma}[Good simulation if not \textbf{FAIL}]\label{lem:simulation_second_part}
If line~\ref{sta:start_offline_coreset}-\ref{sta:mid_offline_coreset} of Algorithm~\ref{alg:offline_coreset_construction} are simulated such that $\forall i\in\{0,1,\cdots,L\},C\in G_i$, $\lambda(C)\in |C\cap X|\pm 0.1 R_i$ or $\lambda(C)\in (1\pm 0.01)\cdot |C\cap X|$ and $\forall i\in\{0,1,\cdots,L\}$, $\lambda(X^i)\in |X^i|\pm 0.1\varepsilon\gamma R_i$ or $\lambda(X^i)\in(1\pm 0.01\varepsilon)\cdot |X^i|$, then conditioned on that the simulating process via $Z''(X)$ does not output \textbf{FAIL}, it simulates line~\ref{sta:repeated_sampling_start}-\ref{sta:repeated_sampling_end} of Algorithm~\ref{alg:offline_coreset_construction} with probability at least $1-\delta$.
\end{lemma}
\begin{proof}
To show the lemma statement, we only need to show that each sample drawn via $\hat{Y}^i_j$ is a uniform sample from $X^i$.
According to Lemma~\ref{lem:hatYijisgood}, with probability at least $1-\delta$, $\forall i\in \{0,1,\cdots, L\},j\in[\hat{m}]$, $\hat{Y}^i_j=Y^i_j$.
Thus, a uniform sample from a non-empty $\hat{Y}^i_j$ is a uniform sample from $X^i$.
If the simulating process does not output \textbf{FAIL}, then the simulating process via $Z''(X)$ is exactly the same as line~\ref{sta:repeated_sampling_start}-\ref{sta:repeated_sampling_end} of Algorithm~\ref{alg:offline_coreset_construction}.
\end{proof}

\begin{lemma}[Size of $Z''(X)$]\label{lem:size_Zprimeprime}
With probability at least $1-2\delta$, the space budget of $Z''(X)$ is at most 
\begin{align*}
O\left(2^{2p+2}\cdot \varepsilon^{-3}\left(kd+d^{O(p)}\cdot \frac{\OPT}{o}\right)\cdot (p\log n\log(kLd)+\log(1/\delta))\cdot L^2/\delta^2\right).
\end{align*}
\end{lemma}
\begin{proof}
The proof is similar to the proof of Lemma~\ref{lem:size_heavy_cell_partioning_sketch}.
We will still use the concept of center cells (see Lemma~\ref{lem:center_cells}).
\begin{align}
&\text{The space budget of }Z''(X)\notag\\
=&\sum_{i=1}^L\sum_{C\in G_i} \min\left(T''_i,\sum_{x\in X:h''_i(x)=1,g''_i(x)=C} |\zeta''_i(x)|\right)\notag\\
\leq & \sum_{i=0}^L \left(\sum_{C\in G_i:C\text{ is a center cell}} T''_i + \sum_{C\in G_i:C\text{ is not a center cell}}\sum_{x\in C\cap X}|\zeta'_i(x)|\right)\notag\\
=& \sum_{i=0}^L \sum_{C\in G_i:C\text{ is a center cell}} T''_i\label{eq:last_part_one}\\
&+\sum_{i=0}^L\sum_{x\in X:x\text{ is not in any center cell in $G_i$}}\sum_{j=1}^{\hat{m}}\mathbf{1}(j\in\zeta''_i(x))\label{eq:last_part_two}.
\end{align}
According to Lemma~\ref{lem:center_cells}, part~\eqref{eq:last_part_one} is at most $6kL/\delta\cdot 10^7\cdot 2^{2p+2}\cdot \varepsilon^{-3}\cdot (p\log(n)\log(kLd)+\log(1/\delta))\cdot Ld/\delta$ with probability at least $1-\delta$.
In the remaining of the proof let us upper bound part~\eqref{eq:last_part_two}.

Recall that $B^*$ is the optimal centers for $X$, i.e., $B^*$ satisfifes $\cost(X,B^*)=\OPT$.
Consider $i\in\{0,1,\cdots,L\}$.
If $x\in X$ is not in any center cell in $G_i$, then $\min_{b\in B^*}\|x-b\|_2\geq \Delta_i/(2d)$.
Thus, $|\{x\in X\mid x\text{ is not in any center cell in }G_i\}|\leq \frac{\OPT}{(\Delta_i/(2d))^p}\leq 100\cdot (2d^{1.5})^p\cdot R_i\cdot \frac{\OPT}{o}$ where the last inequality follows from that $R_i=0.01o/(\sqrt{d}\Delta_i)^p$.
Since $p_i''\cdot \hat{m}\leq 10^6\cdot 2^{2p+2}\cdot \varepsilon^{-3}\cdot (p\log n\log(kLd)+\log(1/\delta))\cdot L/\delta \cdot R_i^{-1}$,
\begin{align*}
&\E\left[\sum_{x\in X:x\text{ is not in any center cell in }G_i}|\zeta''_i(x)|\right]\\
\leq &10^8\cdot 2^{2p+2}\cdot \varepsilon^{-3}\cdot (2d^{1.5})^p\cdot (p\log n\log(kLd)+\log(\delta^{-1}))\cdot \frac{L}{\delta}\cdot \frac{\OPT}{o}.
\end{align*}
By Markov's inequality, with probability at least $1-\delta$, the part~\eqref{eq:last_part_two} can be upper bounded by
\begin{align*}
10^8\cdot 2^{2p+2}\cdot \varepsilon^{-3}\cdot (2d^{1.5})^p\cdot (p\log n\log(kLd)+\log(\delta^{-1}))\cdot \frac{L^2}{\delta^2}\cdot \frac{\OPT}{o}.
\end{align*}
By combining the upper bound of part~\eqref{eq:last_part_one} with the upper bound of part~\eqref{eq:last_part_two}, we complete the proof.
\end{proof}

\begin{lemma}[Number of non-empty $Y^i_j$ is large] \label{lem:num_Yij_large}
Consider $o\geq \OPT/16$.
Suppose for every $i\in \{0,1,\cdots, L\}$ and every cell $C\in G_i$, the estimated value $\lambda(C)$ in line~\ref{sta:point_estimation_in_each_cell} of Algorithm~\ref{alg:offline_heavy_cell_decomposition} satisfies $\lambda(C)\in |C\cap X|\pm 0.1R_i$ or $\lambda(C)\in (1\pm 0.01)\cdot |C\cap X|$, and for every $i\in\{0,1,\cdots, L\}$, the estimated value $\lambda(X^i)$ in line~\ref{sta:point_estimation_in_each_crucial_level} of Algorithm~\ref{alg:offline_coreset_construction} satisfies $\lambda(X^i)\in |X^i|\pm 0.1\varepsilon\gamma R_i$ or $\lambda(X^i)\in(1\pm 0.01\varepsilon)\cdot |X^i|$. 
With probability at least $1-2\delta$, 
\begin{align*}
\forall i \in I, \sum_{j=1}^{\hat{m}} \mathbf{1}(Y^i_j\not=\emptyset)  \geq 10^5\cdot 2^{2p+1}\cdot \varepsilon^{-2}\cdot (p\log n\log(kLd)+\log(1/\delta))\cdot \frac{L}{\delta}\cdot \frac{\lambda(X^i)}{R_i}
\end{align*}
\end{lemma}
\begin{proof}
Consider $i\in I$.
According to Lemma~\ref{lem:center_cells}, with probability at least $1-\delta$, the number of center cells is at most $6kL$.
According to the construction of $X^i$ by Algorithm~\ref{alg:offline_coreset_construction}, none of $x\in X^i$ is in a heavy cell.
Therefore,
\begin{align}
|X^i|&\leq \sum_{C\in G_i:C\text{ is a center cell and $C$ is crucial}} |C\cap X| + \sum_{C\in G_i:C\text{ is not a center cell}} |C\cap X|\notag\\
&\leq \sum_{C\in G_i:C\text{ is a center cell}} 1.1 R_i +  \sum_{C\in G_i:C\text{ is not a center cell}} |C\cap X|\notag\\
&\leq 6kL\cdot 1.1R_i+ \frac{\OPT}{(\Delta_i/(2d))^p}\notag\\
&\leq 7kL\cdot R_i + 100\cdot (2d^{1.5})^p\cdot \frac{\OPT}{o} \cdot R_i\notag\\
&\leq 1600 (kL + (2d^{1.5})^p) \cdot R_i, \label{eq:single_size_of_Xi}
\end{align}
where the second step follows from that a crucial cell contains at most $1.1R_i$ poionts, the third step follows from that the number of center cells is at most $6kL$ and the distance from any point outside a center cell to the optimal centers is at least $(\Delta_i/(2d))^p$, the forth step follows from that $R_i=0.01\cdot o/(\sqrt{d}\cdot \Delta_i)^p$ and the last step follows from $\OPT\leq 16\cdot o$.

By union bound over all $x\in X^i$, we have that 
\begin{align}
\forall j\in [\hat{m}], \Pr[Y^i_j\not=\emptyset ] \leq |X^i|\cdot p''_i < 1, \label{eq:less_than_one_equation}
\end{align}
where the last inequality follows from Equation~\eqref{eq:single_size_of_Xi} and $p''_i< 1/(1600(kL + (2d^{1.5})^p)R_i)$.
Since, $|X^i|\cdot p''_i < 1$, we have $\lceil\frac{1}{|X^i|\cdot p_i''}\rceil\leq \frac{2}{|X^i|\cdot p''_i}$. 
Let $b=10\cdot \lceil \frac{1}{|X^i|\cdot p''_i} \rceil$.
Let $r=\lfloor\hat{m}/b\rfloor\geq \frac{1}{2}\cdot \hat{m}\cdot |X^i|\cdot p_i''-1$.
Since $X^i\not=\emptyset$, it means that there is at least one non-empty crucial cell in $G_i$.
By the construction of $X^i$, we know that a non-empty crucial cell in $G_i$ implies that $R_i\geq \frac{1}{2}$.
Since $p_i''= \min(1/(2000\cdot (kL+(2d^{1.5})^p)\cdot R_i),1)$, we have $p_i''= 1/(2000\cdot (kL+(2d^{1.5})^p)\cdot R_i)$.
Thus,
\begin{align}
r&\geq \frac{1}{2}\cdot \hat{m} \cdot \frac{1}{2000(kL+(2d^{1.5})^p)}\cdot \frac{|X^i|}{R_i}-1\notag\\
&\geq \frac{1}{4}\cdot \hat{m}\cdot \frac{1}{2000(kL+(2d^{1.5})^p)}\cdot \frac{|X^i|}{R_i},\label{eq:lb_of_r}
\end{align}
where the last inequality follows from that $|X^i|\geq \frac{\gamma}{2}\cdot R_i$ since $i\in I$ and $\hat{m}\geq 10^4 (kL+(2d^{1.5})^p)/\gamma$.
For $s\in [r]$, we can define a random variable $Q_s$:
\begin{align*}
Q_s= \sum_{j=(s-1)\cdot b+1}^{s\cdot b} \sum_{x\in X^i} \mathbf{1}\left(j\in \zeta''_i(x)\right). 
\end{align*}
We have that $\E[Q_s] = b\cdot |X^i|\cdot p_i''$.
By our choice of $b$, we know that $\E[Q_s]\in [10,20]$.
Since $Q_s$ is a sum of independent random variables from $\{0,1\}$, the variance $\Var[Q_s]\leq \E[Q_s]\leq 20$.
By Chebyshev's inequality, we have $\Pr[Q_s<1]\leq \frac{1}{4}$.

Notice that $Q_s\geq 1$ means that $\exists j\in \{(s-1)\cdot b+1,(s-1)\cdot b+2,\cdots,s\cdot b\}$ such that $Y^i_j\not=\emptyset$.
Define random variable $P_s=\mathbf{1}(Q_s\geq 1)$.
Then we have $\sum_{s\in [r]} P_s\geq \sum_{j=1}^{\hat{m}}\mathbf{1}(Y^i_j\not=\emptyset)$.
Notice that $\E[\sum_{s\in[r]} P_s]\geq 0.75r$.
By Chernoff bound, we have
\begin{align*}
\Pr\left[\sum_{s\in[r]} P_s \leq 0.5 r\right]\leq 2^{-r/20} \leq \delta/(L+1),
\end{align*}
where the last inequality follows from that $r\geq \frac{1}{4}\cdot \hat{m}\cdot \frac{1}{2000(kL+(2d^{1.5})^p)}\cdot \frac{|X^i|}{R_i}\geq 20\log(2L/\delta)$ since $\frac{|X^i|}{R_i}\geq\frac{\gamma}{2}$ and $\hat{m}\geq 10^7(kL+(2d^{1.5})^p)\log(L/\delta)/\gamma$.

Since $\hat{m}\geq 10^{9}\cdot 2^{2p+2}\cdot \varepsilon^{-2}\cdot (kL+(2d^{1.5})^p)\cdot (p\log n\log(kLd)+\log(1/\delta))\cdot \frac{L}{\delta}$ and $\lambda(X^i)\leq 2|X^i|$, according to Equation~\eqref{eq:lb_of_r}, we have
\begin{align*}
r\geq 10^5\cdot 2^{2p+2}\cdot \varepsilon^{-2}\cdot (p\log n\log(kLd)+\log(1/\delta))\cdot \frac{L}{\delta} \cdot \frac{\lambda(X^i)}{R_i}.
\end{align*}
Thus, by taking union bound over all $i\in I$, with probability at least $1-\delta$,
\begin{align*}
\forall i\in I,\sum_{j=1}^{\hat{m}} \mathbf{1}(Y_j^i\not=\emptyset) \geq 10^5\cdot 2^{2p+1}\cdot \varepsilon^{-2}\cdot (p\log n\log(kLd)+\log(1/\delta))\cdot \frac{L}{\delta}\cdot \frac{\lambda(X^i)}{R_i}.
\end{align*}
\end{proof}

\begin{lemma}[Number of samples needed]\label{lem:num_samples_needed}
Suppose for every $i\in \{0,1,\cdots, L\}$ and every cell $C\in G_i$, the estimated value $\lambda(C)$ in line~\ref{sta:point_estimation_in_each_cell} of Algorithm~\ref{alg:offline_heavy_cell_decomposition} satisfies $\lambda(C)\in |C\cap X|\pm 0.1R_i$ or $\lambda(C)\in (1\pm 0.01)\cdot |C\cap X|$, and for every $i\in\{0,1,\cdots, L\}$, the estimated value $\lambda(X^i)$ in line~\ref{sta:point_estimation_in_each_crucial_level} of Algorithm~\ref{alg:offline_coreset_construction} satisfies $\lambda(X^i)\in |X^i|\pm 0.1\varepsilon\gamma R_i$ or $\lambda(X^i)\in(1\pm 0.01\varepsilon)\cdot |X^i|$.
Consider Algorithm~\ref{alg:offline_coreset_construction}.
With probability at least $1-2\delta$, all of the following things happen:
\begin{enumerate}
\item $t'=\sum_{i\in I}\lambda(X^i)\cdot \min(2^{2p+1}/R_i,1) \leq 200\cdot 2^{2p+2}\cdot \left(k+(2d^{1.5})^p\cdot \frac{\OPT}{o}\right)\cdot L$.
\item The  times that $i\in I$ is sampled by line~\ref{sta:level_sampling} of Algorithm~\ref{alg:offline_coreset_construction} is at most
\begin{align*}
4000\cdot 2^{2p+1}\cdot \varepsilon^{-2} (\log n\log(2t')+\log(1/\delta))\cdot  \frac{\lambda(X^i)}{R_i}\cdot \frac{L}{\delta}.
\end{align*}
\end{enumerate}
\end{lemma}
\begin{proof}
Let us first bound $t'$.
According to Lemma~\ref{lem:center_cells}, with probability at least $1-\delta$, the number of center cells is at most $6kL$.
According to the construction of $X^i$ by Algorithm~\ref{alg:offline_coreset_construction}, none of the point $x\in X^i$ for $i\in\{0,1,\cdots,L\}$ is in a heavy cell.
Therefore, 
\begin{align}
\sum_{i=0}^L|X^i|/R_i&\leq \sum_{i=0}^L\sum_{C\in G_i:C \text{ is a center cell and $C$ is crucial}} |C\cap X|/R_i + \sum_{i=0}^L\sum_{C\in G_i:C \text{ is not a center cell}} |C\cap X|/R_i\notag\\
&\leq \sum_{i=0}^L \sum_{C\in G_i:C \text{ is a center cell}} 1.1+\sum_{i=0}^L\sum_{C\in G_i:C \text{ is not a center cell}} |C\cap X|/R_i\notag\\
&\leq 6kL\cdot 1.1+\sum_{i=0}^L\frac{\OPT}{(\Delta_i/(2d))^p\cdot R_i}\notag\\
&\leq 7kL\cdot+100\cdot (2d^{1.5})^p\cdot \frac{ \OPT}{o}\cdot (L+1), \label{eq:bound_Xi}
\end{align}
where the second step follows from that a crucial cell contains at most $1.1R_i$ points, the third step follows from that the number of center cells is at most $6kL$ and the distance from any point outside a center cell to the optimal centers is at least $(\Delta_i/(2d))^p$, and the forth step follows from that $R_i=0.01\cdot o/(\sqrt{d}\cdot \Delta_i)^p$.

Then we have:
\begin{align*}
t'=&\sum_{i\in I} \lambda(X^i)\cdot \min(2^{2p+1}/R_i,1)\\
\leq & \sum_{i\in I} 2\cdot |X^i|\cdot \min(2^{2p+1}/R_i,1)\\
\leq & \sum_{i=0}^L 2\cdot |X^i|\cdot \min(2^{2p+1}/R_i,1)\\
\leq & 200\cdot 2^{2p+2}\cdot \left(k+(2d^{1.5})^p\cdot \frac{\OPT}{o}\right)\cdot L,
\end{align*}
where the second step follows from that $\lambda(X^i)$ is a good estimation of $|X^i|$ and $\forall i\in I,\lambda(X^i)\geq \gamma R_i$, the third step follows from that $I\subseteq\{0,1,\cdots,L\}$ and the last step follows from Equation~\eqref{eq:bound_Xi}.

Next, let us consider the times that $i\in I$ is sampled by line~\ref{sta:level_sampling} of Algorithm~\ref{alg:offline_coreset_construction}.
The expected number of times that $i\in I$ is sampled is
\begin{align*}
&m\cdot \frac{\lambda(X^i)\cdot \min(2^{2p+1}/R_i,1)}{t'}\\
\leq & 2000\cdot 2^{2p+1}\cdot \varepsilon^{-2} (\log n\log(2t')+\log(1/\delta))\cdot \frac{\lambda(X^i)}{R_i}.
\end{align*}
By Markov's inequality, with probability at least $1-\delta/(L+1)$, $i\in I$ is sampled at most 
\begin{align*}
2000\cdot 2^{2p+1}\cdot \varepsilon^{-2} (\log n\log(2t')+\log(1/\delta))\cdot  \frac{\lambda(X^i)}{R_i}\cdot \frac{L+1}{\delta}
\end{align*}
times.
Thus, by taking union bound over all $i\in I$, with probability at least $1-\delta$, $\forall i\in I$, $i$ is sampled 
\begin{align*}
4000\cdot 2^{2p+1}\cdot \varepsilon^{-2} (\log n\log(2t')+\log(1/\delta))\cdot  \frac{\lambda(X^i)}{R_i}\cdot \frac{L}{\delta}
\end{align*}
times.
\end{proof}

\begin{lemma}[Simulating process does not output \textbf{FAIL} with a good probability]\label{lem:simulating_succcess}
Consider $o\geq \OPT/16$.
Suppose for every $i\in \{0,1,\cdots, L\}$ and every cell $C\in G_i$, the estimated value $\lambda(C)$ in line~\ref{sta:point_estimation_in_each_cell} of Algorithm~\ref{alg:offline_heavy_cell_decomposition} satisfies $\lambda(C)\in |C\cap X|\pm 0.1R_i$ or $\lambda(C)\in (1\pm 0.01)\cdot |C\cap X|$, and for every $i\in\{0,1,\cdots, L\}$, the estimated value $\lambda(X^i)$ in line~\ref{sta:point_estimation_in_each_crucial_level} of Algorithm~\ref{alg:offline_coreset_construction} satisfies $\lambda(X^i)\in |X^i|\pm 0.1\varepsilon\gamma R_i$ or $\lambda(X^i)\in(1\pm 0.01\varepsilon)\cdot |X^i|$.
The probability that the simulating process via $Z''(X)$ does not output \textbf{FAIL} is at least $1-5\delta$.
\end{lemma}
\begin{proof}
The simulating process outputs \textbf{FAIL} if and only if $\exists i\in I$ such that the times that level $i$ is sampled is more than the number of non-empty $\hat{Y}^i_j$.
According to Lemma~\ref{lem:hatYijisgood}, with probability at least $1-\delta$, $\forall i\in I,j\in[\hat{m}],\hat{Y}^i_j=Y^i_j$.
According to Lemma~\ref{lem:num_Yij_large} and Lemma~\ref{lem:num_samples_needed}, with probability at least $1-4\delta$, $\forall i\in I$, the times that level $i$ is sampled is at most the number of non-empty sets $Y^i_j$.
By union bound, with probability at least $1-5\delta$, the simulating process via $Z''(X)$ does not output \text{FAIL}.
\end{proof}

\begin{theorem}[Bucketing-based sketch for $\ell_p$ $k$-clustering]\label{thm:bucket_sketch_clustering}
Let $k\geq 1,p\geq 1$.
Consider the $\ell_p$ $k$-clustering problem over a point set in $[\Delta]^d$ with at most $n$ points. 
Let $\varepsilon,\delta\in(0,0.5)$.
For any $o\geq 1,\alpha\geq 1$, there always exists a weak $o$-restricted $(1+\varepsilon)\alpha$-approximate bucketing-based sketch $Z^o(\cdot)$ such that for any data set $X\subseteq[\Delta]^d$,
\begin{enumerate}
    \item the recover function over $Z^o(X)$ either outputs \textbf{FAIL} or outputs a $(1+\varepsilon)\alpha$-approximation to the optimal $\ell_p$ $k$-clustering cost with probability at least $1-\delta$ conditioned on that the optimal cost is at least $o$;
    \item if the optimal cost is at most $2\cdot o$, with probability at least $1-\delta$, the space budget of $Z^o(X)$ is upper bounded by $\wt{O}_p\left(\varepsilon^{-3}\delta^{-2}\left(kd+d^{O(p)}\right)\log^3(n\Delta)\right)$ and the recover function over $Z^o(X)$ does not output \textbf{FAIL}.
\end{enumerate}
Furthermore, the time needed to compute a filter function, a bucketing function and a processing function on any data item $x\in[\Delta]^d$ is always upper bounded by $\wt{O}_p(\varepsilon^{-3}\delta^{-1}\cdot (k+d^{O(p)})\cdot \log^2(n\Delta))$.
The time to compute recover function over $Z^o(X)$ is at most $\mathcal{T}(m')$ where $m'$ is the space budget of $Z^o(X)$ and $\mathcal{T}(n')$ denotes the running time needed to compute an $\alpha$-approximation for the $\ell_p$ $k$-clustering for a weighted point set with $n'$ points in $\mathbb{R}^d$.
\end{theorem}
\begin{proof}
We simply set $Z^o(X)$ as the composition of $Z(X),Z'(X)$ and $Z''(X)$ described in Section~\ref{sec:heavy_cell_partition_via_bucket_based_sketches} and Section~\ref{sec:sampling_process}, i.e., 
\begin{align*}
Z^o(X)=(Z_0(X),Z_1(X),\cdots,Z_L(X),Z'_0(X),Z'_1(X),\cdots,Z'_L(X),Z''_0(X),Z''_1(X),\cdots,Z''_L(X)).
\end{align*}
According to Lemma~\ref{lem:heavy_cell_decomposition_via_bucketing_based_sketch}, Lemma~\ref{lem:simulate_sampling_first_part} and Lemma~\ref{lem:simulation_second_part}, if the simulating process via $Z''(X)$ does not output \textbf{FAIL}, then with probability at least $1-4\delta$, Algorithm~\ref{alg:offline_heavy_cell_decomposition} and Algorithm~\ref{alg:offline_coreset_construction} can be simulated via $Z(X),Z'(X),Z''(X)$ and every $\lambda(C)$ satisfies $\lambda(C)\in |C\cap X|\pm 0.1R_i$ or $\lambda(C)\in(1\pm 0.01)\cdot |C\cap X|$ and every $\lambda(X^i)$ satisfies $\lambda(X^i) \in |X^i|\pm 0.1\varepsilon\gamma R_i$ or $\lambda(X^i)\in (1\pm 0.01\varepsilon)\cdot |X^i|$.
If $o\leq \OPT$, according to Theorem~\ref{thm:offline_k_clustering_coreset}, the output $(S,w)$ is an $\varepsilon$-coreset for $X$ with probability at least $1-\delta$.
Thus, if $o\leq \OPT$ and the simulating process via $Z''(X)$ does not output \textbf{FAIL}, we are able to obtain a $(1+\varepsilon)\alpha$-approximation to $\OPT$ via $Z^o(X)$ with probability at least $1-5\delta$ by computing an $\alpha$-approximation for the $\ell_p$ $k$-clustering for $(S,w)$.
Since the number of points in $S$ is at most $m'$ --- the space budget of $Z^o(X)$, the running time to get an $\alpha$-approximation for $(S,w)$ is at most $\mathcal{T}(m')$.
According to Lemma~\ref{lem:size_heavy_cell_partioning_sketch}, Lemma~\ref{lem:size_of_Zprime} and Lemma~\ref{lem:size_Zprimeprime}, if $\OPT\leq 2\cdot o$, with probability at least $1-6\delta$, the space budget of $Z^o(X)$ is at most 
\begin{align*}
&O\left(2^{O(p)}\varepsilon^{-3}\left(kd+d^{O(p)}\right)\cdot (p\log(n)\log(kd\log(nd\Delta))+\log(1/\delta))\cdot \frac{\log^2(nd\Delta)}{\delta^2}\right)\\
=&\wt{O}_p\left(\varepsilon^{-3}\delta^{-2}\left(kd+d^{O(p)}\right)\log^3(n\Delta)\right).
\end{align*}
If $\OPT\leq 2\cdot o$, according to Lemma~\ref{lem:simulating_succcess}, the simulating process via $Z''(X)$ does not output \textbf{FAIL} with probability at least $1-5\delta$.

The bottleneck of processing an data point is to compute $\zeta_i''(x)$.
The running time is upper bounded by $O(d+\hat{m})=\wt{O}_p(\varepsilon^{-3}\delta^{-1}\cdot (k+d^{O(p)})\cdot \log^2(n\Delta))$.
\end{proof}

\subsection{Sliding Window Algorithm for $k$-Clustering}
By plugging the bucketing-based sketch for $\ell_p$ $k$-clustering into our algorithmic framework (Algorithm~\ref{alg:sliding_framework}), we are able to obtain an efficient sliding window algorithm for the clustering problem.
\begin{theorem}\label{thm:sliding_window_clustering}
For any $\varepsilon,\delta'\in(0,0.5),\alpha\geq 1,k\geq 1,p\geq 1$, there is a $(1+\varepsilon)\alpha$-approximate algorithm for the $\ell_p$ $k$-clustering problem for a point set from $[\Delta]^d$ in the sliding window model with window size $W\geq 1$ using space $\wt{O}_p(\varepsilon^{-3}{\delta'}^{-2}(kd+d^{O(p)})\log^6(W\Delta))$.
The update time is at most $\wt{O}_p(\varepsilon^{-3}{\delta'}^{-2}(kd+d^{O(p)})\log^7(W\Delta))$.
The running time to get the approximate solution at the end of the stream is at most $\mathcal{T}(\wt{O}_p(\varepsilon^{-3}{\delta'}^{-2}(kd+d^{O(p)})\log^6(W\Delta)))\cdot O_p(\log(Wd\Delta))$ where $\mathcal{T}(n')$ denotes the running time needed to compute an $\alpha$-approximation for the $\ell_p$ $k$-clustering for a weighted point set with $n'$ points in $[\Delta]^d$.
The success probability is at least $1-\delta'$.
\end{theorem}
\begin{proof}
Suppose the point set of interest at the end of the stream is $X\subseteq[\Delta]^d$.
We first need to handle the case if $\OPT=0$.
We have $\OPT=0$ if and only if $X$ contains at most $k$ distinct points.
We can use the following sliding window procedure to check whether $X$ has at most $k$ points:
\begin{enumerate}
    \item Initialize a list of points $L=\emptyset$.
    \item For the latest $x$ in the stream:
        \begin{enumerate}
            \item Add $x$ into $L$.
            \item If there are $2$ points (including $x$ itself) in $L$ that are equal to $x$, remove the point which is equal to $x$ with earlier timestamp from $L$.
            \item Otherwise, if $L$ has $k+2$ points, remove the point with the earliest timestamp from $L$.
        \end{enumerate}
    \item At timestamp $N$, let $L'=\{x\in L\mid\text{ the timestamp is at least }N-W+1\text{, i.e., }x\in X\}$.
    \item If $L'$ has at most $k$ distinct points, return $0$.
\end{enumerate}

In the remaining of the proof, we only need to discuss the case when $\OPT>0$.
Since $X\subseteq[\Delta]^d$ and $|X|\leq W$, we have $\OPT\in[1,W\cdot(\sqrt{d}\Delta)^p]$.
Let $\delta=O_p(\delta'/\log(Wd\Delta))$.
According to Theorem~\ref{thm:bucket_sketch_clustering}, $\forall o\in[1,W\cdot(\sqrt{d}\Delta)^p], \alpha\in [0,1],\varepsilon,\delta\in(0,0.5)$, there always exists a weak $o$-restricted $(1+\varepsilon)\alpha$-approximate bucketing-based sketch $Z^o(\cdot)$ for $\ell_p$ $k$-clustering such that the recover function over $Z^o(X)$ either outputs \textbf{FAIL} or outputs a $(1+\varepsilon)\alpha$-approximation to $\OPT$ with probability at least $1-\delta$ conditioned on that $\OPT\geq o$.
Furthermore, if $\OPT\leq 2\cdot o$, with probability at least $1-\delta$, the space budget of $Z^o(X)$ is at most $\wt{O}_p(\varepsilon^{-3}\delta^{-2}(kd+d^{O(p)})\log^3(W\Delta))$ and the recover function over $Z^o(X)$ does not output \textbf{FAIL}.
According to Theorem~\ref{thm:framework}, after plugging the bucketing-based sketch into Algorithm~\ref{alg:sliding_framework}, we obtain the sliding window algorithm which outputs a $(1+\varepsilon)\alpha$-approximation to $\OPT$.
The space needed is at most $\wt{O}_p(\varepsilon^{-3}\delta^{-2}(kd+d^{O(p)})\log^3(W\Delta))\cdot O_p(\log(Wd\Delta))=\wt{O}_p(\varepsilon^{-3}\delta^{-2}(kd+d^{O(p)})\log^4(W\Delta))=\wt{O}_p(\varepsilon^{-3}{\delta'}^{-2}(kd+d^{O(p)})\log^6(W\Delta))$.
The success probability is at most $1-\delta\cdot O_p(\log(Wd\Delta))=1-\delta'$.
According to Theorem~\ref{thm:bucket_sketch_clustering}, the time needed to compute a filter function, a bucketing function and a processing function on any data item is always at most $\wt{O}_p(\varepsilon^{-3}\delta^{-1}\cdot(k+d^{O(p)})\cdot \log^2(W\Delta))$.
According to Theorem~\ref{thm:framework}, the update time is at most $\wt{O}_p(\log(Wd\Delta)\cdot\log(Wd\Delta)\cdot \varepsilon^{-3}\delta^{-1}\cdot(k+d^{O(p)})\cdot \log^2(W\Delta)+\varepsilon^{-3}{\delta'}^{-2}(kd+d^{O(p)})\log^6(W\Delta)\cdot \log(Wd\Delta))=\wt{O}_p(\varepsilon^{-3}{\delta'}^{-2}(kd+d^{O(p)})\log^7(W\Delta))$.
At the end of the stream, to gent the final approximation, we need to evaluate the recover function $O_p(\log(Wd\Delta))$ times according to Algorithm~\ref{alg:sliding_framework}.
Since the size of each sketch is at most $\wt{O}_p(\varepsilon^{-3}{\delta'}^{-2}(kd+d^{O(p)})\log^6(W\Delta))$, according to Theorem~\ref{thm:bucket_sketch_clustering}, the time needed to compute the recover function for each sketch is at most $\mathcal{T}(\wt{O}_p(\varepsilon^{-3}{\delta'}^{-2}(kd+d^{O(p)})\log^6(W\Delta)))$.
Thus, the overall running time to get the approximate solution at the end of the stream is at most $\mathcal{T}(\wt{O}_p(\varepsilon^{-3}{\delta'}^{-2}(kd+d^{O(p)})\log^6(W\Delta)))\cdot O_p(\log(Wd\Delta))$.
\end{proof}

\addcontentsline{toc}{section}{References}
\bibliographystyle{alpha}
\bibliography{ref}

\appendix
\section{Discussion of the Smoothness of $k$-Cover, Diversity Maximization and Clustering}\label{sec:lack_of_smoothness}

In this section, we briefly review the smooth histogram framework.
Then we will show examples that $k$-cover, diversity maximization and clustering are not smooth enough to obtain $(1\pm\varepsilon)$-approximation via smooth histogram.

\subsection{Smooth Function and Smooth Histogram}
Exponential histogram is an algorithmic framework for sliding window algorithms and is proposed by~\cite{datar2002maintaining}.
It shows that \emph{weakly additive functions} can be approximated by efficient sliding window algorithms via exponential histogram.
Later, \cite{braverman2007smooth} proposes the smooth histogram framework which can be used to develop efficient sliding window algorithms for approximating  \emph{smooth functions}.
It was shown that the class of weakly additive functions is a sub-class of the class of smooth functions.
Thus, smooth histogram is a more general framework than exponential histogram.

Suppose $A$ and $B$ are two streams of data items.
Let $B\subseteq_{r} A$ indicate that $B$ is a suffix of $A$.
The definition of the smooth function is as follows.
\begin{definition}[$(\alpha,\beta)$-Smooth function~\cite{braverman2007smooth}]
For $\alpha\in(0,1),\beta\in(0,\alpha]$, function $f(\cdot)$ is $(\alpha,\beta)$-smooth, if it holds the following properties:
\begin{enumerate}
    \item $\exists$ constant $c>0$, $\forall A,0\leq f(A)\leq |A|^c$.
    \item $\forall B\subseteq_{r} A$, $f(A)\geq f(B)$.
    \item If $B\subseteq_{r} A$ and $(1-\beta) f(A)\leq f(B)$ then $(1-\alpha)f(A\cup C)\leq f(B\cup C)$ for any subsequent $C$.
\end{enumerate}
\end{definition}
Roughly speaking, for a smooth function, if the function over a suffix of a stream of the data is a good approximation of the function over the entire stream, then after the arrival of any new data items, the function over such suffix is always a good approximation over the entire stream of data.
Therefore, the idea of smooth histogram is to maintain a set of start points over the stream.
The sliding window algorithm maintains a sketch of data items from the start point to the current end of the stream for each start point.
A start point is created when a new data item is arrived and a start point can be dropped if the stream from the next start point can approximate the stream from the previous start point very well, i.e., the value of the function over the stream starting from the next start point is at least $(1-\beta)$ times the value of the function over the stream starting from the previous start point.
If the function is an $(\alpha,\beta)$-smooth function, the function values produced by two adjacent start points are always within $(1-\alpha)$.
Since a start point can be dropped if the function value produced by the next start point is within $(1-\beta)$ of the function value produced by the previous start point, the smooth histogram only needs to maintain at most $O\left(\frac{1}{\beta}\cdot \log W\right)$ start points.
The formal statement of the guarantee of the sliding window algorithm via smooth histogram is shown as the following.
\begin{theorem}[\cite{braverman2007smooth}]
Let $f(\cdot)$ be an $(\alpha,\beta)$-smooth function.
If there is a streaming algorithm which outputs a $(1-\varepsilon')$-approximation (or $1/(1-\varepsilon')$-approximation) of $f(\cdot)$ using $S(\varepsilon')$ space, then there exists a sliding window algorithm which outputs a $(1-\alpha-\varepsilon')$-approximation (or $1/(1-\alpha-\varepsilon')$-approximation) of $f(\cdot)$ using $O(\log(W)/\beta\cdot S(\varepsilon'))$ space where $W\geq 1$ is the window size.
\end{theorem}
Thus, if $f(\cdot)$ is $(\alpha,\beta)$-smooth with smaller $\alpha$, the sliding window algorithm via smooth histogram can output a better approximation.
In the remaining of the section, we will show that there is a constant lower bound of $\alpha$ for $k$-cover and diversity maximization.
$\ell_p$ $k$-Clustering is not $(\alpha,\beta)$-smooth for any $\alpha,\beta\geq 0$.

\subsection{$(\Omega(1),\beta)$-Smoothness of $k$-Cover}

The following lemma shows that we cannot use the smooth histogram to obtain a sliding window algorithm for $k$-cover in the edge-arrival model with approximation ratio better than $1/2$.
\begin{lemma}
$k$-Cover in edge-arrival model cannot be $(\alpha,\beta)$-smooth for any $\alpha<1/2$.
\end{lemma}
\begin{proof}
Consider $k=1$.
There are two sets $S_1,S_2$ and $2m$ elements $e_1,e_2,\cdots,e_{2m}$.
The stream $A$ contains edges $(S_1,e_1),(S_1,e_2),\cdots,(S_1,e_m),(S_2,e_{m+1}),(S_2,e_{m+2}),\cdots,(S_2,e_{2m})$.
The suffix stream $B$ contains edges $(S_2,e_{m+1}),(S_2,e_{m+2}),\cdots,(S_2,e_{2m})$.
The optimal $1$-cover of the both $A$ and $B$ are the same, i.e., $\OPT_1(A)=\OPT_1(B)=m$.
Thus, for any $\beta\geq 0,(1-\beta)\OPT_1(A)\leq \OPT_1(B)$.
Let $C$ be the edges $(S_1,e_{m+1}),(S_1,e_{m+2}),\cdots,(S_1,e_{2m})$.
Then we have $(1-\alpha)\cdot \OPT_1(A\cup C)=(1-\alpha)\cdot 2m>  \OPT_1(B\cup C)=m$.
Thus, it cannot be $(\alpha,\beta)$-smooth for any $\alpha<1/2$.
To generalize to $k>1$, we only need to make $k$ copies of $S_1,S_2,e_1,e_2,\cdots,e_{2m}$.
\end{proof}

\subsection{$(\Omega(1),\beta)$-Smoothness of Diversity Maximization}
Recall that the diversity functions that we are considering is listed in Table~\ref{tab:diversity_functions}.
The definition of $\wb{\OPT}$ is given in Section~\ref{sec:div_max_sliding_window}.
Let us consider the case $k=2$.
In this case, $\wb{\OPT}(P)$ is always the same for every diversity function, i.e., $\wb{\OPT}(P)=\max_{Q\subseteq P:|Q|=2}\min_{u\not=v\in Q}\|u-v\|_2$.
The following lemma shows that we cannot use smooth histogram to obtain a sliding window algorithm for diversity maximization for $k=2$ and any diversity function listed in Table~\ref{tab:diversity_functions} with approximation ratio better than $1/\sqrt{2}$.

\begin{lemma}
Diversity maximization with $k=2$ and any diversity function listed in Table~\ref{tab:diversity_functions} cannot be $(\alpha,\beta)$-smooth for any $\alpha<1-1/\sqrt{2}$.
\end{lemma}
\begin{proof}
Let $A=\{(0,1,1,0),(1,1,0,0),(1,0,1,0)\}$.
Let $B$ be the suffix of $A$, i.e., $B=\{(1,1,0,0),(1,0,1,0)\}$.
Let $C=\{(1,0,0,1)\}$.
It is easy to verify that $\wb{\OPT}(A)=\wb{\OPT}(B)=\sqrt{2}$ and $\wb{\OPT}(A\cup C)=2,\wb{\OPT}(B\cup C)=\sqrt{2}$.
Thus we have $(1-\beta)\wb{\OPT}(A)\leq \wb{\OPT}(B)$ but $(1-\alpha)\wb{\OPT}(A\cup C)>\wb{\OPT}(B\cup C)$.
\end{proof}

In the following, we extend the above lemma to the case when $k>2$.
\begin{lemma}
Consider $k>2$.
\begin{enumerate}
\item Remote-edge cannot be $(\alpha,\beta)$-smooth for any $\alpha<1-1/\sqrt{2}$.
\item Remote-clique cannot be $(\alpha,\beta)$-smooth for any $\alpha<(2-\sqrt{2})/(k(k-1))$.
\item Remote-tree cannot be $(\alpha,\beta)$-smooth for any $\alpha<(2-\sqrt{2})/(2(k-1))$.
\item Remote-cycle cannot be $(\alpha,\beta)$-smooth for any $\alpha<(2-\sqrt{2})/(2k)$.
\item Remote $t$-trees cannot be $(\alpha,\beta)$-smooth for any $\alpha<(2-\sqrt{2})/(2(k-t))$.
\item Remote $t$-cycle cannot be $(\alpha,\beta)$-smooth for any $\alpha<(2-\sqrt{2})/(2k)$.
\item Remote-star cannot be $(\alpha,\beta)$-smooth for any $\alpha<(2-\sqrt{2})/(2(k-1))$.
\item Remote-bipartition cannot be $(\alpha,\beta)$-smooth for any $\alpha<(2-\sqrt{2})/(2\lfloor k/2\rfloor \lceil k/2 \rceil)$.
\item Remote-pseudoforest cannot be $(\alpha,\beta)$-smooth for any $\alpha<(2-\sqrt{2})/(2k)$.
\item Remote-matching cannot be $(\alpha,\beta)$-smooth for any $\alpha<(2-\sqrt{2})/k$.
\end{enumerate}
\end{lemma}
\begin{proof}
Let $A$ contain $k+1$ $(2k)$-dimensional vectors:
\begin{align*}
A=\{&v_1=&(1,0,0,\cdots,0,1,0,0,\cdots,0),\\
&v_2=&(1,1,0,\cdots,0,0,0,0,\cdots,0),\\
&v_3=&(0,1,1,\cdots,0,0,0,0,\cdots,0),\\
&\cdots&\\
&v_{k+1}=&(0,0,0,\cdots,1,1,0,0,\cdots,0)\},
\end{align*}
i.e., $v_1=\e_1+\e_{k+1},v_2=\e_2+\e_1,v_3=\e_3+\e_2,\cdots,v_{k+1}=\e_{k+1}+\e_k$.
Let $B$ contain $k$ $(2k)$-dimensional vectors: $B=\{v_2,v_3,\cdots,v_{k+1}\}$.
Let $C$ contain $k-1$ $(2k)$-dimensional vectors:
\begin{align*}
C=\{&u_1=&(0,1,0,\cdots,0,0,1,0,\cdots,0),\\
&u_2=&(0,0,1,\cdots,0,0,0,1,\cdots,0),\\
&\cdots&\\
&u_{k-1}=&(0,0,0,\cdots,1,0,0,0,\cdots,1)\},
\end{align*}
i.e., $u_1=\e_2+\e_{k+2},u_2=\e_3+\e_{k+3},\cdots,u_{k-1}=\e_k+\e_{2k}$.

By symmetry, if we choose arbitrary $k$ vectors $Q\subseteq A$, then $\min_{p\not=q\in Q} \|p-q\|_2=\sqrt{2}$.
Thus, if we choose $Q=B$, then $\min_{p\not=q\in Q} \|p-q\|_2=\sqrt{2}$.
Consider $A\cup C$. if we choose $Q=\{v_1,u_1,u_2,\cdots,u_{k-1}\}\subseteq A\cup C$, then $\min_{p\not=q\in Q}\|p-q\|_2=2$.

Next we consider $B\cup C$.
\begin{claim}
For any $i\in [k-1]$, let $Q_i$ be arbitrary $i+1$ vectors from $\{v_2,v_3,\cdots,v_{i+2},u_1,u_2,\cdots,u_i\}$.
Then $\min_{p\not=q\in Q_i} \|p-q\|_2= \sqrt{2}$.
\end{claim}
\begin{proof}
The proof is by induction.
Consider the base case: $Q_1$ contains arbitrary $2$ vectors from $\{v_2,v_3,u_1\}$.
It is easy to verify that  $\min_{p\not=q\in Q_1} \|p-q\|_2=\sqrt{2}$.
Suppose the claim is true for $i-1$.
Consider $Q_i$. 
There are two cases.
In the first case, $Q_i$ contains at least $i$ vectors from $\{v_2,v_3,\cdots,v_{i+1},u_1,u_2,\cdots,u_{i-1}\}$.
By induction hypothesis, since $Q_i$ contains at least $i$ vectors from $\{v_2,v_3,\cdots,v_{i+1},u_1,u_2,\cdots,u_{i-1}\}$, we have $\min_{p\not = q\in Q_i}\|p-q\|_2= \sqrt{2}$.
In the second case, $Q_i$ contains at most $i-1$ vectors from $\{v_2,v_3,\cdots,v_{i+1},u_1,u_2,\cdots,u_{i-1}\}$.
In this case, both $v_{i+2}=\e_{i+2}+\e_{i+1}$ and $u_i=\e_{i+1}+\e_{i+k}$ are in $Q_i$.
Thus, $\min_{p\not = q\in Q_i}\|p-q\|_2=\|\e_{i+2}-\e_{i+k}\|_2= \sqrt{2}$.
\end{proof}
By use $i=k-1$ in the above claim, we know that any $Q\subseteq B\cup C$ with $|Q|=k$, we have $\min_{p\not = q\in Q}\|p-q\|_2=\sqrt{2}$.
Thus,
\begin{enumerate}
\item For remote-edge: $\OPT(A)=\OPT(B)=\sqrt{2},\OPT(A\cup C)=2,\OPT(B\cup C)=\sqrt{2}$.
\item For remote-clique: $\OPT(A)=\OPT(B)=(k-1)\cdot \sqrt{2}+(k-1)(k-2)/2\cdot 2,\OPT(A\cup C)=k(k-1),\OPT(B\cup C)\leq k(k-1)-(2-\sqrt{2})$.
\item For remote-tree: $\OPT(A)=\OPT(B)=(k-1)\cdot \sqrt{2},\OPT(A\cup C)=2(k-1),\OPT(B\cup C)\leq 2(k-1)-(2-\sqrt{2})$.
\item For remote-cycle: $\OPT(A)=\OPT(B)=k\cdot \sqrt{2},\OPT(A\cup C)=2k,\OPT(B\cup C)\leq 2k-(2-\sqrt{2})$.
\item For remote $t$-trees: $\OPT(A)=\OPT(B)=(k-t)\cdot \sqrt{2},\OPT(A\cup C)=2(k-t),\OPT(B\cup C)\leq 2(k-t)-(2-\sqrt{2})$.
\item For remote $t$-cycle: $\OPT(A)=\OPT(B)=k\cdot \sqrt{2},\OPT(A\cup C)=2k,\OPT(B\cup C)\leq 2k-(2-\sqrt{2})$.
\item For remote-star: $\OPT(A)=\OPT(B)=(k-1)\cdot \sqrt{2},\OPT(A\cup C)=2(k-1),\OPT(B\cup C)\leq 2(k-1)-(2-\sqrt{2})$.
\item For remote-bipartition: $\OPT(A)=\OPT(B)=\lfloor k/2 \rfloor\cdot \lceil k/2 \rceil \cdot \sqrt{2},\OPT(A\cup C)=2\cdot \lfloor k/2 \rfloor\cdot \lceil k/2 \rceil,\OPT(B\cup C)\leq 2\cdot \lfloor k/2 \rfloor\cdot \lceil k/2 \rceil-(2-\sqrt{2})$.
\item For remote-pseudoforest: $\OPT(A)=\OPT(B)=k\cdot \sqrt{2},\OPT(A\cup C)=2k,\OPT(B\cup C)\leq 2k-(2-\sqrt{2})$.
\item For remote-matching: $\OPT(A)=\OPT(B)=k/2\cdot \sqrt{2},\OPT(A\cup C)=2\cdot k/2,\OPT(B\cup C)\leq 2\cdot k/2-(2-\sqrt{2})$.
\end{enumerate}
\end{proof}

\subsection{Non-smoothness of $\ell_p$ $k$-Clustering}
As shown by~\cite{braverman2016clustering}, $\ell_p$ $k$-clustering is not $(\alpha,\beta)$-smooth for any $\alpha$ and $\beta$.
Thus, smooth histogram cannot be used to develop a sliding window algorithm for $\ell_p$ $k$-clustering for any multiplicative approximation.
\begin{lemma}[\cite{braverman2016clustering}]
Consider $\ell_p$ $k$-clustering problem for $p\geq 1$.
There exists sets of points $A,B,C\subseteq [\Delta]^d$ such that $\forall 0<\beta\leq \alpha <1$, $(1-\beta)\OPT(A\cup B)\leq \OPT(B)$ but $(1-\alpha)\OPT(A\cup B\cup C)>(1-\alpha)\OPT(B\cup C)$.
\end{lemma}

\subsection{$(\Omega(1),\beta)$-Smoothness of Toy $1$-Median Problem}\label{sec:lack_of_smoothness_toy_1_median}
In the toy $1$-median problem, given a (multi-)set $X\subseteq\{0,1\}$, the goal is to compute $\OPT(X)=\min_{z\in\{0,1\}} \sum_{x\in X}|x-z|$.
Let $A=\{0,0\},B=\{0\},C=\{1,1,1\}$.
Then we have $(1-\beta)\OPT(A)=0\leq \OPT(B)$.
But $\frac{1}{2}\cdot \OPT(A\cup C)=\frac{1}{2}\cdot 2 = 1\geq  \OPT(B\cup C)=1$

\section{Analysis of Offline Coreset Construction for $k$-Clustering}\label{sec:analysis_offline_k_clustering}
We give the analysis of the coreset construction shown in Section~\ref{sec:offline_coreset_k_clustering}.
The goal is to prove Theorem~\ref{thm:offline_k_clustering_coreset}.
The analysis is similar to~\cite{hu2018nearly}.
We include the analysis in this section for completeness.

\begin{definition}[Good estimation of number of points]
For $i\in\{0,1,\cdots,L\},C\in G_i$ if the estimated value $\lambda(C)$ in line~\ref{sta:point_estimation_in_each_cell} of Algorithm~\ref{alg:offline_heavy_cell_decomposition} satisfies $\lambda(C)\in |C\cap X|\pm0.1R_i$ or $\lambda(C)\in (1\pm 0.01)\cdot |C\cap X|$, then we say that $\lambda(C)$ is good.
For $i\in \{0,1,\cdots, L\}$, if the estimated value $\lambda(X^i)$ in line~\ref{sta:point_estimation_in_each_crucial_level} of Algorithm~\ref{alg:offline_coreset_construction} satisfies $\lambda(X^i)\in |X^i|\pm 0.1\varepsilon\gamma R_i$ or $\lambda(X^i)\in (1\pm 0.01\varepsilon)\cdot |X^i|$, then we say that $\lambda(X^i)$ is good.
\end{definition}

\begin{fact}\label{fac:points_in_heavy_cells}
If $o\leq \OPT$ and the estimation $\lambda(C)$ is good for every cell $C$, then for level $i\in \{-1,0,\cdots,L\}$, every heavy cell $C\in G_i$ contains at least $0.9R_i$ points, i.e., $|C\cap X|\geq 0.9 R_i$, and every crucial cell $C\in G_i$ contains at most $1.1R_i$ points, i.e., $|C\cap X|\leq 1.1 R_i$.
\end{fact}
\begin{proof}
For level $i\in \{0,1,\cdots,L\}$, the claim follows directly from the construction of heavy and crucial cells in Algorithm~\ref{alg:offline_heavy_cell_decomposition} and Algorithm~\ref{alg:offline_coreset_construction}.
For $i=-1$, $G_{-1}$ does not contain any crucial cell.
The unique heavy cell $C\in G_{-1}$ contains every point in $X$.
Notice that each point can contribute to the clustering cost at most $(\sqrt{d}\cdot \Delta)^p$.
Thus, $|X|\geq \OPT/(\sqrt{d}\cdot \Delta)^p\geq o/(\sqrt{d}\cdot \Delta)^p\geq R_{-1}$.
Thus $|C\cap X|\geq 0.9R_{-1}$.
\end{proof}

\begin{fact}[No heavy cell in $G_L$]\label{fac:no_heavy_cell_GL}
If the estimated value $\lambda(C)$ is good for every cell $C$, none of the cell in $G_L$ is marked as heavy by Algorithm~\ref{alg:offline_heavy_cell_decomposition}.
\end{fact}
\begin{proof}
Since $L=\lceil\log(nd\Delta)\rceil+10$, we have $R_L=0.01\cdot o/(\sqrt{d}\cdot \Delta_L)^p\geq 0.01/(\sqrt{d}\cdot \Delta/2^L)^p\geq 2n$.
Consider each cell $C\in G_L$.
According to Fact~\ref{fac:points_in_heavy_cells}, if $C\in G_L$ is a heavy cell, then we have $|C\cap X|\geq 0.9R_{-1}\geq 1.8n$ which contradicts to $|X|\leq n$.
Thus, any cell $C\in G_L$ cannot be heavy.
\end{proof}

\begin{fact}[$X$ is partitioned by crucial cells]\label{fac:partition_of_X}
If the estimated value $\lambda(C)$ is good for every cell $C$, $X^0,X^1,\cdots,X^L$ in Algorithm~\ref{alg:offline_coreset_construction} is a partition of $X$, i.e., $\forall x\in X$, there is a unique $i\in\{0,1,\cdots,L\}$ such that $x\in X^i$.
\end{fact}
\begin{proof}
Notice that any ancestor of a crucial cell must be a heavy cell by the construction of crucial cells.
Thus, any $x\in X$ can be in at most one crucial cell.
Consider an arbitrary point $x\in X$. 
Let $i\in \{0,1,\cdots, L\}$ be the smallest value such that cells $c_{-1}(x),c_{0}(x),\cdots,c_{i-1}(x)$ are heavy and $c_i(x)$ is not heavy.
Notice that Fact~\ref{fac:no_heavy_cell_GL} implies that $c_L(x)$ is not heavy.
On the other hand, the cell $c_{-1}(x)$ must be heavy according to line~\ref{sta:handle_minusone_level} of Algorithm~\ref{alg:offline_heavy_cell_decomposition}.
Thus we are guaranteed to find such $i$.
We can verify that $c_i(x)$ is a crucial cell.
Thus, $x\in X^i$.
\end{proof}

\begin{fact}
If $C$ is a heavy cell in $G_i$ for some $i\in\{-1,0,\cdots,L-1\}$, then every child cell $C'$ of $C$ is either heavy or crucial. 
\end{fact}
\begin{proof}
Follows directly from the construction of heavy and crucial cells in Algorithm~\ref{alg:offline_heavy_cell_decomposition} and Algorithm~\ref{alg:offline_coreset_construction}.
\end{proof}

\begin{fact}\label{fac:good_estimation_Xi}
If $\lambda(X^i)$ is good for every $i\in\{0,1,\cdots,L\}$, then $\forall i\in I, |X^i|\geq 0.9\gamma R_i$ and $\forall i\not\in I, |X^i|\leq 1.1\gamma R_i$.
\end{fact}
\begin{proof}
Follows directly from the construction of heavy and crucial cells in Algorithm~\ref{alg:offline_heavy_cell_decomposition} and Algorithm~\ref{alg:offline_coreset_construction}.
\end{proof}

\begin{lemma}[Levels outside $I$ can be ignored]\label{lem:levels_outside_I_can_be_ignored}
Suppose all estimated values $\lambda(C)$ and $\lambda(X^i)$ in Algorithm~\ref{alg:offline_heavy_cell_decomposition} and Algorithm~\ref{alg:offline_coreset_construction} are good.
Let $X^I:=\bigcup_{i\in I} X^i$.
If $o\in[1,\OPT]$, then for any $B\subset \mathbb{R}^d$ with $|B|=k$, we have
\begin{align*}
\cost(X^I,B)\leq \cost(X,B)\leq (1+\varepsilon/10)\cdot \cost(X^I,B).
\end{align*}
\end{lemma}
The above lemma shows that we can construct a coreset of $X^I$ instead.
The coreset of $X^I$ will automatically become a coreset of $X$.
Before proving Lemma~\ref{lem:levels_outside_I_can_be_ignored}, we need to prove the following claim.
\begin{claim}\label{cla:lot_of_points_in_heavy_cell}
Suppose all estimated values $\lambda(C)$ and $\lambda(X^i)$ in Algorithm~\ref{alg:offline_heavy_cell_decomposition} and Algorithm~\ref{alg:offline_coreset_construction} are good.
Let $X^I:=\bigcup_{i\in I} X^i$ and let $\gamma$ be the same as in Algorithm~\ref{alg:offline_coreset_construction}.
Suppose $o\in [1,\OPT]$.
For $i\in \{0,1,\cdots, L\}$, if $C\in G_{i-1}$ is marked as heavy, then we have $|C\cap X^I|\geq (1-2^{p+1}(L-i)\gamma)\cdot |C\cap X|$.
\end{claim}
\begin{proof}
The proof is by induction.
Consider the base case $i=L$.
If there is no heavy cell in $G_{L-1}$, the claim holds for $i=L$ case.
Otherwise consider any heavy cell $C\in G_{L-1}$.
According to Fact~\ref{fac:no_heavy_cell_GL}, none of the children cells of $C$ is heavy.
According to the construction of crucial cells, every child cell of $C$ is crucial.
By the construction of $X^L$, we know that $|C\cap X^L|=|C\cap X|$.
Since $\lambda(C)$ is good, we have $|C\cap X^L|=|C\cap X|\geq 0.9R_{L-1}\geq R_L/2^{p+1}$.
Since $\lambda(X^L)$ is good, we have 
\begin{align*}
\lambda(X^L)&\geq \min(|X^L|-0.1\varepsilon \gamma R_L,(1-0.01\varepsilon)|X^L|)\\
&\geq \min(|X^L\cap C|-0.1\varepsilon \gamma R_L,(1-0.01\varepsilon)|X^L\cap C|)\\
&=\min(|X\cap C|-0.1\varepsilon \gamma R_L,(1-0.01\varepsilon)|X\cap C|)\\
&\geq \gamma R_L,
\end{align*}
where the last step follows from that $\gamma \leq 1/2^{p+2}$.
Thus, by the construction of $I$, we have $L\in I$ and thus $|C\cap X^I|=|C\cap X^L|=|C\cap X|$.

Now suppose the claim is true for $i+1,i+2,\cdots,L$.
If there is no heavy cell in $G_{i-1}$, the claim directly holds for $i$.
Otherwise consider any heavy cell $C\in G_{i-1}$.
There are two cases.
In the first case, $i\in I$, i.e., $X^i\subseteq X^I$.
In this case, we have
\begin{align*}
|C\cap X^I| &= \sum_{C'\in G_i:\text{$C'$ is a heavy child of $C$}} |C'\cap X^I| + \sum_{C'\in G_i:\text{$C'$ is a crucial child of $C$}} |C'\cap X^I|\\
&=\sum_{C'\in G_i:\text{$C'$ is a heavy child of $C$}} |C'\cap X^I| + |C\cap X^i|\\
&\geq (1-2^{p+1}(L-i-1)\gamma)\sum_{C'\in G_i:\text{$C'$ is a heavy child of $C$}} |C'\cap X| + |C\cap X^i|\\
&\geq (1-2^{p+1}(L-i-1)\gamma)|C\cap X|\\
&\geq (1-2^{p+1}(L-i)\gamma)|C\cap X|,
\end{align*}
where the second step follows from $i\in I$ and the construction of $X^i$, the third step follows from the induction hypothesis and the forth step follows from that $(C\cap X^i)\cup \bigcup_{C'\in G_i:C'\text{ is a heavy child of }C}(C'\cap X)=C\cap X$.

In the second case, $i\not\in I$, i.e., $X^i\cap X^I=\emptyset$.
In this case, we have:
\begin{align}
\sum_{C'\in G_i:C'\text{ is a heavy child of }C}|C'\cap X| &= |C\cap X| - \sum_{C'\in G_i:C'\text{ is a crucial child of }C} |C'\cap X| \notag\\
&\geq |C\cap X|-|X^i|\notag\\
&\geq |C\cap X|-1.1\gamma R_i\notag\\
&\geq (1-2^{p+1}\gamma) |C\cap X|, \label{eq:lb_points_heavy_cell}
\end{align}
where the first step follows from that each child of $C$ is either heavy or crucial, the second step follows from that $C'\cap X\subseteq X^i$ for any crucial cell $C'\in G_i$, the third step follows from $i\not\in I$ and Fact~\ref{fac:good_estimation_Xi}, and the last step follows from that $C\in G_{i-1}$ is a heavy cell which implies that $|C\cap X|\geq 0.9R_{i-1}\geq 0.9R_i/2^p$ by Fact~\ref{fac:points_in_heavy_cells}.
Thus, we have
\begin{align*}
|C\cap X^I|&\geq \sum_{C'\in G_i:C'\text{ is a heavy child of }C}|C'\cap X^I|\\
&\geq (1-2^{p+1}(L-i-1)\gamma)\sum_{C'\in G_i:C'\text{ is a heavy child of }C}|C'\cap X|\\
&\geq (1-2^{p+1}(L-i-1)\gamma)(1-2^{p+1}\gamma)|C\cap X|\\
&\geq (1-2^{p+1}(L-i)\gamma)|C\cap X|,
\end{align*}
where the second step follows from the induction hypothesis, the third step follows from Equation~\eqref{eq:lb_points_heavy_cell}.
\end{proof}

\paragraph{Proof of Lemma~\ref{lem:levels_outside_I_can_be_ignored}.}
Now we are able to prove Lemma~\ref{lem:levels_outside_I_can_be_ignored}.
Observe that $X^I\subseteq X$, it is easy to verify that $\cost(X^I,B)\leq \cost(X,B)$ for any $B\subset \mathbb{R}^d$ with $|B|\leq k$.
In the remaining of the proof, we will show that $\cost(X,B)\leq (1+\varepsilon/10)\cdot \cost(X^I,B)$ for any $B\subset \mathbb{R}^d$ with $|B|\leq k$.

Consider an arbitrary level $i\in\{0,1,\cdots,L\}\setminus I$, i.e., $X^i\cap X^I=\emptyset$.
Consider a point $x\in X^i$.
By the construction of $X^i$ by Algorithm~\ref{alg:offline_coreset_construction}, we know that $c_{i-1}(x)$ is a heavy cell.
By Claim~\ref{cla:lot_of_points_in_heavy_cell}, $|c_{i-1}(x)\cap X^I|\geq (1-2^{p+1}L\gamma)\cdot |c_{i-1}(x)\cap X|$.
Since $c_{i-1}(x)$ is a heavy cell, we have $|c_{i-1}(x)\cap X|\geq 0.9R_{i-1}$ by Fact~\ref{fac:points_in_heavy_cells}.
Therefore, 
\begin{align}
|c_{i-1}(x)\cap X^I|\geq (1-2^{p+1}L\gamma)\cdot 0.9R_{i-1}\geq 0.5 R_{i-1}
\end{align}
since $\gamma \leq \frac{\varepsilon}{40\cdot 2^{2p+2}L}$.
By an averaging argument, we can find a point $y\in c_{i-1}(x)\cap X^I$ such that
\begin{align}
\min_{b\in B}\|x-b\|_2^p&\leq 2^{p-1}\cdot \left(\min_{b\in B}\|y-b\|_2^p+\|x-y\|_2^p\right)\notag\\
&\leq 2^{p-1}\cdot \left(\min_{b\in B}\|y-b\|_2^p+ (\sqrt{d}\Delta_{i-1})^p\right)\notag\\
&\leq 2^{p-1} \cdot \left(\frac{\sum_{x'\in c_{i-1}(x)\cap X^I} \min_{b\in B}\|x'-b\|_2^p}{|c_{i-1}(x)\cap X^I|}+ (\sqrt{d}\Delta_{i-1})^p\right),\label{eq:average_argument}
\end{align}
where the first step follows from the convexity of $\|\cdot\|_2^p$, the second step follows from that $x$ and $y$ are both in the cell $c_{i-1}(x)\in G_{i-1}$, and the last step follows from an averaging argument.

Let $X^N=X\setminus X^I$.
We have the following:
\begin{align*}
\cost(X,B)&=\cost(X^I,B)+\cost(X^N,B)\\
&=\cost(X^I,B) + \sum_{i\not\in I}\sum_{x\in X^i} \min_{b\in B}\|x-b\|_2^p\\
&\leq\cost(X^I,B) + 2^{p-1}\cdot\sum_{i\not\in I}\sum_{x\in X^i}\left(\frac{\sum_{x'\in c_{i-1}(x)\cap X^I} \min_{b\in B}\|x'-b\|_2^p}{|c_{i-1}(x)\cap X^I|}+ (\sqrt{d}\Delta_{i-1})^p\right)\\
&\leq\cost(X^I,B) + 2^{p-1}\cdot\sum_{i\not\in I}\sum_{x\in X^i}\left(\frac{\sum_{x'\in c_{i-1}(x)\cap X^I} \min_{b\in B}\|x'-b\|_2^p}{0.5R_{i-1}}+ (\sqrt{d}\Delta_{i-1})^p\right)\\
&\leq\cost(X^I,B) + 2^{p-1}\cdot\sum_{i\not\in I}\sum_{x\in X^i}\left(\frac{\sum_{x'\in X^I} \min_{b\in B}\|x'-b\|_2^p}{0.5R_{i-1}}+ (\sqrt{d}\Delta_{i-1})^p\right)\\
&\leq \cost(X^I,B) + 2^{p-1}\cdot\sum_{i\not\in I}1.1\gamma R_i\left(\frac{\sum_{x'\in X^I} \min_{b\in B}\|x'-b\|_2^p}{0.5R_{i-1}}+ (\sqrt{d}\Delta_{i-1})^p\right)\\
&\leq \cost(X^I,B) + 2^{p-1}\cdot 4L\gamma R_i\left(\frac{\sum_{x'\in X^I} \min_{b\in B}\|x'-b\|_2^p}{0.5R_{i-1}}+ (\sqrt{d}\Delta_{i-1})^p\right)\\
&=\cost(X^I,B) + 2^{p-1}\cdot 4L\gamma R_i\left(\frac{\cost(X^I,B)}{0.5R_{i-1}}+ (\sqrt{d}\Delta_{i-1})^p\right)\\
&=\cost(X^I,B) + 2^{p-1}\cdot 4L\gamma R_i\cdot\frac{\cost(X^I,B)}{0.5R_{i-1}}+ 2^{2p+1}L\gamma \cdot 0.01o\\
&=\cost(X^I,B) + 2^{2p+2}\cdot L\gamma \cdot\cost(X^I,B)+ 2^{2p+1}L\gamma \cdot 0.01o\\
&\leq (1+2^{2p+2}\cdot L\gamma)\cost(X^I,B) + 2^{2p+1} L \gamma \cdot 0.01\cdot  \cost(X,B),
\end{align*}
where the first step follows from $X=X^I\cup X^N$ and $X^I\cap X^N=\emptyset$, the second step follows from the definition of $\cost(X^N,B)$, the third step follows from Equation~\eqref{eq:average_argument}, the fourth step follows from Equation~\ref{eq:average_argument}, the fifth step follows from $(c_{i-1}(x)\cap X^I)\subseteq X^I$, the sixth step follows from that $|X^i|\leq 1.1\gamma R_i$ for $i\not\in I$ by Fact~\ref{fac:good_estimation_Xi}, the seventh step follows from that $|\{0,1,\cdots,L\}\setminus I|\leq L+1\leq 2L$, the eighth step follows from the definition of $\cost(X^I,B)$, the ninth step follows from that $R_i=0.01o/(0.5\sqrt{d}\Delta_{i-1})^p$, the tenth step follows from that $R_i=R_{i-1}\cdot 2^p$, and the last step follows from that $o\leq \OPT\leq \cost(X,B)$.

Since $\gamma\leq \frac{\varepsilon}{40\cdot 2^{2p+2}L}$, the above inequality implies that
\begin{align*}
\frac{\cost(X,B)}{\cost(X^I,B)}=\frac{1+2^{2p+2}\cdot L\gamma}{1-0.01\cdot 2^{2p+1}L\gamma}\leq \frac{1+\varepsilon/40}{1-\varepsilon/40}\leq 1+\varepsilon/10.
\end{align*}
\qed

\begin{lemma}[Sensitivity upper bound of each point]\label{lem:sensitivity_ub}
Suppose all estimated values $\lambda(C)$ and $\lambda(X^i)$ in Algorithm~\ref{alg:offline_heavy_cell_decomposition} and Algorithm~\ref{alg:offline_coreset_construction} are good.
Consider $X^I$ computed by Algorithm~\ref{alg:offline_coreset_construction}.
If $o\in[1,\OPT]$, then for every $i\in I$ and every $x\in X^i$,
\begin{align*}
\max_{B\subset\mathbb{R}^d:|B|\leq k} \frac{\min_{b\in B}\|x-b\|_2^p}{\cost(X^I,B)}\leq \min\left(2^{2p+1}\cdot \frac{1}{R_i},1\right).
\end{align*}
\end{lemma}

\begin{proof}
Since $x\in X^I$, it is clear that $\max_{B\subset\mathbb{R}^d:|B|\leq k} \frac{\min_{b\in B}\|x-b\|_2^p}{\cost(X^I,B)}\leq 1$.
In the remaining of the proof, we will show that $\max_{B\subset\mathbb{R}^d:|B|\leq k} \frac{\min_{b\in B}\|x-b\|_2^p}{\cost(X^I,B)}\leq 2^{2p+1}\cdot \frac{1}{R_i}.$

Let $B\subset \mathbb{R}^d$ be a subset of at most $k$ centers.
Consider an arbitrary point $x\in X$.
By Fact~\ref{fac:partition_of_X}, there is a unique $i\in\{0,1,\cdots,L\}$ such that $x\in X^i$, i.e., there is a crucial cell in $G_i$ that contains $x$.
Let $C$ be $c_{i-1}(x)$, i.e., the crucial cell that contains $x$ is a child cell of $C$.
By the construction of crucial cells, it is easy to verify that $C$ is a heavy cell in $G_{i-1}$.
According to Claim~\ref{cla:lot_of_points_in_heavy_cell} and Fact~\ref{fac:points_in_heavy_cells}, we have
\begin{align}\label{eq:heavy_cell_size}
|C\cap X^I|\geq (1-2^{p+1}L\gamma)\cdot |C\cap X| \geq 0.5R_{i-1}, 
\end{align}
where the last inequality follows from that $|C\cap X|\geq 0.9R_{i-1}$ and $\gamma = \varepsilon/(40\cdot 2^{2p+2}\cdot L)$.
According to Claim~\ref{cla:lot_of_points_in_heavy_cell}, the cell $C$ contains at least one point from $X^I$, i.e., $C\cap X^I\not=\emptyset$.
By an average argument, we can find a point $y\in C\cap X^I$ such that
\begin{align}
\min_{b\in B}\|y-b\|_2^p \leq \frac{\sum_{x'\in C\cap X^I}\min_{b\in B}\|x'-b\|_2^p}{|C\cap X^I|}.\label{eq:average_argument_2}
\end{align}
Then we have:
\begin{align*}
\frac{\min_{b\in B}\|x-b\|_2^p}{\cost(X^I,B)}&\leq 2^{p-1}\cdot \left(\frac{\min_{b\in B}\|y-b\|_2^p}{\cost(X^I,B)} + \frac{\|x-y\|_2^p}{\cost(X^I,B)}\right)\\
&\leq  2^{p-1}\cdot \left( \frac{\sum_{x'\in C\cap X^I} \min_{b\in B}\|x'-b\|_2^p}{|C\cap X^I|\cdot \cost(X^I,B)}  + \frac{\|x-y\|_2^p}{\cost(X^I,B)}  \right)\\
&\leq  2^{p-1}\cdot \left( \frac{\sum_{x'\in X^I} \min_{b\in B}\|x'-b\|_2^p}{|C\cap X^I|\cdot \cost(X^I,B)}  + \frac{\|x-y\|_2^p}{\cost(X^I,B)}  \right)\\
& = 2^{p-1}\cdot \left(\frac{1}{|C\cap X^I|} + \frac{\|x-y\|_2^p}{\cost(X^I,B)}\right)\\
&\leq 2^{p-1}\cdot \left(\frac{1}{|C\cap X^I|} + \frac{(\sqrt{d}\Delta_{i-1})^p}{\cost(X^I,B)}\right)\\
&\leq 2^{p-1}\cdot \left(\frac{1}{|C\cap X^I|} + (1+\varepsilon/10)\cdot \frac{(\sqrt{d}\Delta_{i-1})^p}{\cost(X,B)}\right)\\
&\leq 2^{p-1}\cdot \left(\frac{1}{|C\cap X^I|} + (1+\varepsilon/10)\cdot \frac{(\sqrt{d}\Delta_{i-1})^p}{\OPT}\right)\\
&\leq 2^{p-1}\cdot \left(\frac{1}{|C\cap X^I|} + (1+\varepsilon/10)\cdot \frac{(\sqrt{d}\Delta_{i-1})^p}{o}\right)\\
&=2^{p-1}\cdot \left(\frac{1}{|C\cap X^I|} + (1+\varepsilon/10)\cdot \frac{1}{100\cdot R_{i-1}}\right)\\
&\leq 2^{p-1}\cdot \left(\frac{1}{0.5R_{i-1}} + (1+\varepsilon/10)\cdot \frac{1}{100\cdot R_{i-1}}\right)\\
&\leq 2^{p+1}\cdot\frac{1}{R_{i-1}}\\
&=2^{2p+1}\cdot\frac{1}{R_i},
\end{align*}
where the first step follows from the convexity of $\|\cdot\|_2^p$, the second step follows from Equation~\eqref{eq:average_argument_2}, the third step follows from $(C\cap X^I)\subseteq X^I$, the forth step follows from that $\cost(X^I,B)=\sum_{x'\in X^I} \min_{b\in B}\|x'-b\|_2^p$, the fifth step follows from that both $x,y$ are in the cell $C\in G_{i-1}$ and the side length of $C$ is $\Delta_{i-1}$, the sixth step follows from Lemma~\ref{lem:levels_outside_I_can_be_ignored} that $\cost(X,B)\leq (1+\varepsilon/10)\cdot \cost(X^I,B)$, the seventh step follows from $\OPT\leq \cost(X,B)$, the eighth step follows from $o\leq \OPT$, the ninth step follows from $R_{i-1}=0.01\cdot o/(\sqrt{d}\Delta_{i-1})^p$, the tenth step follows from Equation~\eqref{eq:heavy_cell_size}, and the twevlth step follows from $R_i=R_{i-1}\cdot 2^p$.
\end{proof}

\paragraph{Proof of Theorem~\ref{thm:offline_k_clustering_coreset}.}
Now, we are going to prove Theorem~\ref{thm:offline_k_clustering_coreset}.
Before we proceeds to more details, we need to introduce the sensitivity sampling theorem.
Given a point set $X\subset \mathbb{R}^d$, the sensitivity of a point $x\in X$ with respect to the $\ell_p$ $k$-clustering problem is defined as
\begin{align*}
s(x):=\sup_{B\subset \mathbb{R}^d:|B|\leq k} \frac{\min_{b\in B}\|x-b\|_2^p}{\cost(X,B)}.
\end{align*}
The following theorem shows the guarantee of the sensitivity sampling based coreset construction.
\begin{theorem}[\cite{braverman2016new},\cite{braverman2019streaming}]\label{thm:sensitivity_sampling}
Let $X\subset\mathbb{R}^d$ be a set of at most $n$ points.
Let $s':X\rightarrow \mathbb{R}_{\geq 0}$ satisfy that $\forall x\in X,$ $s'(x)$ is an upper bound of the sensitivity of $x$, i.e., $s'(x)\geq s(x)$.
Let $\hat{t}=\sum_{x\in X} s'(x)$.
Let $\delta,\varepsilon\in(0,0.5)$.
Consider a multiset $S$ of $m$ i.i.d. samples from $X$, where each sample chooses $x\in X$ with probability $\mathrm{prob}(x)=s'(x)/\hat{t}$.
For each sampled point $x$, we assign it an arbitrary weight $w(x)\in (1\pm \varepsilon/4)\cdot 1/(m\cdot\mathrm{prob}(x))$. 
If $m\geq \lceil 50\hat{t}/\varepsilon^2\cdot (\log n\log \hat{t}+\log(1/\delta)) \rceil$, $(S,w)$ is an $\varepsilon$-coreset for $X$ with probability at least $1-\delta$.
\end{theorem}

Consider the sampling procedure in Algorithm~\ref{alg:offline_coreset_construction}.
Each sample $x\in X^i\subseteq X^I$ is drawn with probability 
\begin{align*}
\mathrm{prob}(x)=\frac{\lambda(X^i)\cdot \min(2^{2p+1}/R_i,1)}{\sum_{j\in I} \lambda(X^j)\cdot \min(2^{2p+1}/R_j,1)}\cdot \frac{1}{|X^i|} & = \frac{\frac{\lambda(X^i)}{|X^i|}\cdot \min(2^{2p+1}/R_i,1)}{\sum_{j\in I}\sum_{x'\in X^j} \frac{\lambda(X^j)}{|X^j|}\cdot \min(2^{2p+1}/R_j,1)}.
\end{align*}
The weight $w(x)$ of the sample $x$ is
\begin{align*}
w(x)=\frac{t'}{m\cdot \min(2^{2p+1}/R_i,1)} =\frac{1}{m}\cdot \frac{\sum_{j\in I}\sum_{x'\in X^j}\frac{\lambda(X^j)}{|X^j|}\cdot \min\left(2^{2p+1}/R_j,1\right)}{\frac{\lambda(X^i)}{|X^i|}\cdot \min(2^{2p+1}/R_i,1)}\cdot \frac{\lambda(X^i)}{|X^i|}.
\end{align*}
Thus, we have $w(x)=\frac{\lambda(X^i)}{|X^i|}\cdot 1/(m\cdot \mathrm{prob}(x))$.
Since $i\in I$, we have $|X^i|\geq 0.9\gamma R_i$ by Fact~\ref{fac:good_estimation_Xi}.
Since $\lambda(X^i)$ is a good estimated value of $|X^i|$, we have $\lambda(X^i)/|X^i|\in (1\pm \varepsilon/8)$.
Therefore, we have $w(x)\in (1\pm \varepsilon/8)\cdot 1/(m\cdot \mathrm{prob}(x))$.
Let $s'(x):=\frac{\lambda(X^i)}{|X^i|}\cdot 2\cdot \min\left(2^{2p+1}/R_i,1\right)$.
Since $\lambda(X^i)/|X^i|\geq 1/2$, we have $s'(x)\geq s(x)$ by Lemma~\ref{lem:sensitivity_ub}.
Let $\hat{t}=\sum_{x\in X^I} s'(x)$. 
We have $\hat{t}=2\cdot t'$ and $\forall x\in X^I, \mathrm{prob}(x)=s'(t)/\hat{t}$.
According to Theorem~\ref{thm:sensitivity_sampling}, $(S,w)$ outputted by Algorithm~\ref{alg:offline_coreset_construction} is an $(\varepsilon/2)$-coreset for $X^I$ with probability at least $1-\delta$.
According to Lemma~\ref{lem:levels_outside_I_can_be_ignored}, when $(S,w)$ is an $(\varepsilon/2)$-coreset for $X^I$, $(S,w)$ is an $\varepsilon$-coreset for $X$.
\qed

\section{$\Omega(k)$ Space is Necessary for $k$-Clustering}\label{sec:necessity_k_clustering}
There is a simple reduction from a communication problem to the $k$-clustering problem to show that any multiplicative approximate sliding window algorithm for $k$-clustering needs at least $\Omega(k)$ space.
In the INDEX~\cite{kn06} problem, Alice has a $k$-bit binary string and Bob has an index $i\in[k]$.
Alice can send a message to Bob, and Bob needs to answer whether the $i$-th bit of the Alice's string is $1$ without sending any message to Alice.
If Bob needs to answer correctly with probability at least $2/3$, Alice must send $\Omega(k)$ bits to Bob.
Now consider the reduction to the sliding window $k$-clustering in $1$-dimensional space.
Let the window size be $k+1$.
For the $j$-th bit of Alice's string, if the bit is $0$ Alice adds a point $2\cdot j-1$ into the stream, and if the bit is $1$ Alice adds a point $2\cdot j$ into the stream.
Then Alice simulates the sliding window algorithm and sends the memory to Bob.
Bob adds a point $2\cdot i$ into the stream, and continues simulating the sliding window algorithm.
If the $i$-th bit is $1$, then there are only $k$ distinct points in the stream and the clustering cost is $0$.
If the $i$-th bit is $0$, then there are $k+1$ distinct points in the stream, and the clustering cost is at least $1$.
Thus, if the sliding window algorithm can output any multiplicative approximation, Bob can determine whether the $i$-th bit of Alice's string is $1$.
Thus, the sliding window algorithm for $k$-clustering needs $\Omega(k)$ space.

\section{Dimension Reduction for $k$-Clustering}\label{sec:dim_reduce}
According to~\cite{mmr19}, there is a simple way to reduce the dimension of points for $\ell_p$ $k$-clustering problems.
\begin{lemma}[Theorem 1.3 in~\cite{mmr19}]
Consider $\ell_p$ $k$-clustering problem for $p\geq 1$.
Consider an arbitrary point set $P=\{p_1,p_2,\cdots,p_n\}\subset\mathbb{R}^d$.
Let $\varepsilon\in (0,0.5)$.
Let $S\in\mathbb{R}^{d'\times d}$ be a random matrix where $d'=\Theta(\varepsilon^{-2}\log k)$ and each entry of $S$ is drawn uniformly at random from $\{-1/\sqrt{d'},1/\sqrt{d'}\}$.
Let $P'=\{S\cdot p_1,S\cdot p_2,\cdots,S\cdot p_n\}$.
Then, with probability at least $0.99$,
\begin{align*}
(1-\varepsilon)\cdot \min_{Z\subset\mathbb{R}^d:|Z|=k}\cost(P,Z)\leq \min_{Z'\subset\mathbb{R}^{d'}:|Z'|=k}\cost(P',Z')\leq (1+\varepsilon)\cdot \min_{Z\subset\mathbb{R}^d:|Z|=k}\cost(P,Z).
\end{align*}
\end{lemma}
According to the above lemma, we only need to solve the clustering problem for the point set $P'$.
For a sliding window algorithm, we can sample a matrix $S$ before the algorithm starts.
For each point $p$ during the stream, we can compute $S\cdot p$ and apply the algorithm for $S\cdot p$.
Thus, the dimension of a point that needed to be handled is at most $O(\log(k)/\varepsilon^2)$.
The entire algorithm only needs $O(d'\cdot d)=O(kd/\varepsilon^2)$ additional space to store $S$.

\end{document}